\documentclass[11pt,english,a4paper]{article}

\usepackage[T1]{fontenc} 
\usepackage[latin9]{inputenc}   

\usepackage{lmodern}
\usepackage[DIV=8]{typearea} 

\usepackage{geometry}
\usepackage{microtype,xparse,ifmtarg}  
\usepackage{array}
\usepackage{multirow,setspace}
\usepackage{amsmath} 
\usepackage{amssymb}
\usepackage{amsthm} 
\usepackage{thmtools}
\usepackage{csquotes}
\usepackage{subfigure}

\usepackage{cleveref}

\usepackage
{algorithm2e}

\usepackage[svgnames]{xcolor}  
\usepackage{xspace}
\usepackage[shortlabels]{enumitem}    
\usepackage{tikz}   

\usepackage{comment} 

\usepackage{graphicx} 
\usepackage{caption} 

\usepackage{authblk}

\definecolor{ForestGreen}{rgb}{0.1333,0.5451,0.1333}
\definecolor{DarkRed}{rgb}{0.8,0,0}
\definecolor{Red}{rgb}{1,0,0}

\allowdisplaybreaks[1]

\makeatletter
\g@addto@macro\bfseries{\boldmath}
\makeatother

\makeatletter
\g@addto@macro\mdseries{\unboldmath}
\g@addto@macro\normalfont{\unboldmath}
\g@addto@macro\rmfamily{\unboldmath}
\g@addto@macro\upshape{\unboldmath}
\makeatother

\renewcommand{\paragraph}[1]{\medskip\noindent{\bfseries #1}\xspace}

\declaretheorem[numberwithin=section,refname={Theorem,Theorems},Refname={Theorem,Theorems},name={Theorem}]{thm}
\declaretheorem[numberlike=thm,refname={Theorem,Theorems},Refname={Theorem,Theorems},name={Theorem}]{theorem}
\declaretheorem[numberlike=thm,refname={Lemma,Lemmas},Refname={Lemma,Lemmas},name={Lemma}]{lem}
\declaretheorem[numberlike=thm,refname={Lemma,Lemmas},Refname={Lemma,Lemmas},name={Lemma}]{lemma}
\declaretheorem[numberlike=thm,refname={Corollary,Corollaries},Refname={Corollary,Corollaries},name={Corollary}]{cor}
\declaretheorem[numberlike=thm,refname={Corollary,Corollaries},Refname={Corollary,Corollaries},name={Corollary}]{corollary}

\declaretheorem[numberlike=thm,refname={Fact,Facts},Refname={Fact,Facts},name={Fact}]{fact}

\declaretheorem[numberlike=thm,refname={Proposition,Propositions},Refname={Proposition,Propositions},name={Proposition}]{prop}
\declaretheorem[numberlike=thm,refname={Proposition,Propositions},Refname={Proposition,Propositions},name={Proposition}]{proposition}
\declaretheorem[numberlike=thm,refname={Observation,Observations},Refname={Observation,Observations}]{observation}

\declaretheorem[style=remark,numberlike=thm,refname={Claim,Claims},Refname={Claim,Claims}]{claim}

\crefname{algorithm}{Algorithm}{Algorithms}
\Crefname{algorithm}{Algorithm}{Algorithms}

\theoremstyle{definition}
\declaretheorem[numberlike=theorem]{definition}

\newcommand{\ot}{\tilde{O}}
\newcommand{\ignore}[1]{}

\newcommand{\pset}{\mathcal{P}}
\newcommand{\uset}{\mathcal{U}}

\newcommand{\bset}{\mathcal{B}}

\newcommand{\capacity}{c}
\newcommand{\Ot}{\tilde{O}}

\newcommand{\out}{\operatorname{out}}
\newcommand{\vol}{\operatorname{vol}}
\newcommand{\poly}{\operatorname{poly}}
\newcommand{\polylog}{\operatorname{polylog}}

\newcommand{\tw}{\mathrm{tw}}
\newcommand{\mincut}{\mathrm{mincut}}
\newcommand{\dep}{\mathsf{dep}}
\newcommand{\mlp}{\textsf{Multi-level Pruning}\xspace}
\newcommand{\cp}{\textsf{Cluster Decomposition}\xspace}
\newcommand{\ep}{\textsf{Expander Decomposition}\xspace}

\def\ShowComment{True} 

\ifdefined\ShowComment
\def\thatchaphol#1{\marginpar{$\leftarrow$\fbox{T}}\footnote{$\Rightarrow$~{\sf\textcolor{purple}{#1 --Thatchaphol}}}}
\def\zihan#1{\marginpar{$\leftarrow$\fbox{Z}}\footnote{$\Rightarrow$~{\sf\textcolor{orange}{#1 --Zihan}}}}

\def\harry#1{\marginpar{$\leftarrow$\fbox{H}}\footnote{$\Rightarrow$~{\sf\textcolor{blue}{#1 --Harry}}}}

\def\note#1{#1}

\else
\def\thatchaphol#1{}
\def\zihan#1{}
\def\note#1{} 

\fi

\title{The Expander Hierarchy\\and its Applications to Dynamic Graph Algorithms}

\author[1]{Gramoz Goranci}
\author[2]{Harald R\"{a}cke}
\author[3]{Thatchaphol Saranurak}
\author[4]{Zihan Tan}
\affil[1]{University of Toronto, Canada}
\affil[2]{TU Munich, Germany}
\affil[3]{Toyota Technological Institute at Chicago, USA}
\affil[4]{University of Chicago, USA}

\date{}

\begin{document}

        \begin{titlepage}   
                \maketitle
                \pagenumbering{gobble}   
                \begin{abstract}
	We introduce a notion for hierarchical graph clustering which we call the \emph{expander hierarchy} and show a fully dynamic algorithm for maintaining such a hierarchy on a graph with $n$ vertices undergoing edge insertions and deletions using $n^{o(1)}$ update time. 
	An expander hierarchy is a tree representation of graphs that faithfully captures the cut-flow structure and consequently our dynamic algorithm almost immediately implies several results including:
	\begin{enumerate}
		\item The first fully dynamic algorithm with $n^{o(1)}$ worst-case update time that allows querying $n^{o(1)}$-approximate conductance, $s$-$t$ maximum flows, and $s$-$t$ minimum cuts for any given $(s,t)$ in $O(\log^{1/6} n)$ time. Our results are deterministic and extend to multi-commodity cuts and flows. All previous fully dynamic (or even decremental) algorithms for any of these problems take $\Omega(n)$ update or query time. 
		The key idea behind these results is a fully dynamic algorithm for maintaining a \emph{tree flow sparsifier}, a notion introduced by R\"acke~[FOCS'02] for constructing competitive oblivious routing schemes.{\small \par}
		\item A deterministic fully dynamic connectivity algorithm with $n^{o(1)}$ worst-case update time. 
		This significantly simplifies the recent algorithm by Chuzhoy et al.~that uses the framework of Nanongkai, Saranurak, and Wulff-Nilsen {[}FOCS'17{]}.{\small \par}
		\item A deterministic fully dynamic treewidth decomposition algorithm on constant-degree graphs with $n^{o(1)}$ worst-case update time that maintains a treewidth decomposition of width $\text{tw}(G)\cdot n^{o(1)}$
		where $\text{tw}(G)$ denotes the treewidth of the current graph. This is the first non-trivial dynamic algorithm for this problem.{\small \par}
	\end{enumerate}
	Our technique is based on a new stronger notion of the expander decomposition, called the \emph{boundary-linked expander decomposition}. This decomposition is more robust against updates and better captures clustering structure of graphs.
	Given that the expander decomposition has proved extremely useful in many fields, including approximation, sketching, distributed, and dynamic algorithms, we expect that our new notion will find more future applications.
\end{abstract}

                \newpage   
                \setcounter{tocdepth}{2}  
                \tableofcontents    
        \end{titlepage}

        \newpage         
        \pagenumbering{arabic}

\section{Introduction}

Computation on \emph{trees }is usually significantly easier than on
\emph{general graphs}. Hence, one of the universal themes in graph
algorithms is to compute tree representations that faithfully preserve fundamental
properties of a given graph. Examples include spanning forests (preserving connectivity), shortest
path trees (preserving distances from a source), Gomory-Hu trees (preserving
pairwise minimum cuts), low stretch spanning trees and tree embedding
(preserving average distances between pairs of vertices), and treewidth decomposition (preserving
``tree-like'' structure). Among all known approaches for representing a graph with a tree,
the \emph{tree flow sparsifier} introduced by Räcke \cite{Racke02}
is astonishingly strong. Roughly speaking, it is a tree $T$ that
approximately preserves the values of \emph{all cuts} of a graph $G$ (see the formal
definition in \Cref{subsec:applications}). The existence of such
trees, which is far from obvious, already enables competitive oblivious
routing schemes with both theoretical \cite{Racke02,Racke08} and
practical impact \cite{ApplegateC03}. Its polynomial-time construction
\cite{HarrelsonHR03,BienkowskiKR03} also leads to polynomial-time
approximation algorithms for many fundamental problems including minimum
bisection, min-max partitioning, $k$-multicut, etc (see e.g.~\cite{AndreevGGMM09,BansalFKMNNS14,ChekuriKS13,Racke08,RackeS14}).
More recently, the almost-linear time construction was shown \cite{RackeST14}
and played a key role in obtaining the celebrated result of approximating
maximum flows in near-linear time \cite{Sherman13,KelnerLOS14,Peng16}.
Given that the construction for static graphs are now well understood,
we raise the challenging question of whether it is possible to maintain tree flow
sparsifiers in \emph{dynamic graphs} that undergo a sequence of edge
insertions and deletions without recomputing from scratch after each
update.

In this paper, we answer this question in affirmative by introducing a new notion for hierarchical graph clustering
which we call the \emph{expander hierarchy}. We state a precise definition
later in \Cref{subsec:def}. We show
that the expander hierarchy is a tree representation of a graph that
is strong enough to imply tree flow sparsifiers (and much more), 
and yet robust against updates in the sense that it admits fully dynamic
algorithms on an $n$-vertex graph maintaining the hierarchy in $n^{o(1)}$
update time.

The fact that tree flow sparsifiers can be maintained efficiently
immediately allows us to efficiently compute approximate solutions to a wide range
of flow/cut-based problems, including max flows, multi-commodity flows,
minimum cuts, multi cuts, multi-way cuts, and conductance. Specifically, for all these problems, this gives the first sub-linear time fully dynamic algorithms with $n^{o(1)}$ worst-case update and query time.
The power of the expander hierarchy is not limited to tree flow sparsifiers. 
It also gives an algorithm for the deterministic dynamic connectivity with $n^{o(1)}$ worst-case update time, that significantly simplifies the recent breakthrough result in \cite{ChuzhoyGLNPS19det,NanongkaiSW17} on this problem.
Moreover, it gives the first algorithm for maintaining an approximate treewidth decomposition, a central object in the field of fixed-parameter
tractable algorithms. We discuss these applications in detail in \Cref{subsec:applications}.


In short, we introduce the expander hierarchy as a clean combinatorial
object that is very robust against adversarial updates, yet strong
enough to imply many new results and simplify previous important development.
It is likely that future development on dynamic graph algorithms can
build on such a hierarchy.

\subsection{Our Results: The Dynamic Expander Hierarchy}
\label{subsec:def}

First, we recall definitions related to \emph{expanders}. Let $G=(V,E)$ be
an $n$-vertex $m$-edge unweighted graph. For any set $S,T\subseteq V$,
let $E_{G}(S,T)$ denote a set of edges between $S$ and $T$. The
\emph{volume} of $S$ is $\vol_{G}(S)=\sum_{u\in S}\deg_{G}(u)$ and
we write $\vol(G)=\vol_{G}(V)$. The \emph{conductance} of a cut $(S,V\setminus S)$
is $\Phi_G(S)=\frac{|E_{G}(S,V\setminus S)|}{\min\{\vol_{G}(S),\vol_{G}(V\setminus S)\}}$
and the conductance of $G$ is denoted by $\Phi_{G}=\min_{\emptyset\neq S\subset V}\Phi_{G}(S)$.
We say that $G$ is a \emph{$\phi$-expander} iff $\Phi_{G}\ge\phi$.

We need the following generalized notation for induced subgraphs in order to define our new decomposition.
Recall that $G[S]$ denotes the subgraph of $G$ induced by $S$. For any $w\ge0$,
let $G[S]^{w}$ be obtained from $G[S]$ by adding $\left\lceil w\right\rceil $
self-loops to each vertex $v\in S$ for every boundary edge
$(v,x)$, $x\notin S$. Note that $G[S]^{0}=G[S]$ and vertices in
$G[S]^{1}$ have the same degree as in the original graph $G$ (each
self-loop contributes $1$ to the degree of the node that it is incident
to).

\paragraph{Stronger Expander Decomposition.}
The core of this paper is to identify a stronger notion of the well-known \emph{expander decomposition} \cite{KannanVV04}
which states that, given any graph $G=(V,E)$ and any parameter $\phi>0$,
there is a partition $\uset=(U_{1},\dots,U_{k})$ of $V$ into \emph{clusters}, such that: 
\begin{enumerate}
\item $\sum_{i}|E(U_{i},V\setminus U_{i})|=\tilde{O}(\phi m)$.
\item For all $i$, $G[U_{i}]$ is a $\phi$-expander.
\end{enumerate}
Basically, the decomposition says that one can remove $\tilde{O}(\phi)$-fraction
of edges so that each connected component in the remaining graph is a $\phi$-expander. As
expanders have many algorithm-friendly properties such as, having low diameter,
small mixing time, etc., this decomposition has found numerous applications
across areas, including property testing \cite{GoldreichR98,KumarSS18}, approximation
algorithms \cite{KleinbergR96,Trevisan05}, fast graph algorithms
\cite{SpielmanT11-SecondJournal,Sherman13,KelnerLOS14}, distributed
algorithms \cite{Censor-HillelHK17,ChangPZ19,ChangS19,EdenFFKO19},
and dynamic algorithms \cite{NanongkaiS17,Wulff-Nilsen17,NanongkaiSW17}.

In this paper, we propose a stronger notion of the expander decomposition. For parameters
$\alpha,\phi >0$, we define an \emph{$(\alpha,\phi)$-boundary-linked
expander decomposition} of $G$ as a partition $\uset=(U_{1},\dots,U_{k})$
of $V$, together with $\phi_{1},\dots,\phi_{k}\ge\phi$, such that:
\begin{enumerate}
\item \label{enu:boundary}$\sum_{i}|E(U_{i},V\setminus U_{i})|=\tilde{O}(\phi m)$.
\item \label{enu:linked}For all $i$, $G[U_{i}]^{\alpha/\phi_{i}}$ is
a $\phi_{i}$-expander.
\item \label{enu:sparse cut}For all $i$, $|E(U_{i},V\setminus U_{i})|\le\tilde{O}(\phi_{i}\vol_G(U_{i}))$.
\end{enumerate}
Compared to the the previous definition in \cite{KannanVV04}, we strengthen Property \ref{enu:linked}
and additionally require Property \ref{enu:sparse cut}.
Before discussing the power of the new decomposition, we start with 
intuitive observations how it more faithfully captures the clustering structure of graphs.
It is instructive to
think of $\alpha=1/\polylog(n)$ and $\phi=1/n^{o(1)}\ll\alpha$. 

Intuitively, in a good graph clustering, vertices in a cluster are better connected to the inside of the cluster than to its outside.
Observe that the stronger form of Property \ref{enu:linked} implies that, for every vertex $v\in U_{i}$, $\deg_{G[U_{i}]}(v)\ge\alpha\cdot\deg_{G[V\setminus U_{i}]}(v)$. Without this strengthening, there could be a vertex
$v\in U_{i}$ where $\deg_{G[U_{i}]}(v)\ll\deg_{G[V\setminus U_{i}]}(v)$,
which is counter-intuitive.
{} Moreover, Property \ref{enu:sparse cut} additionally implies that
$\deg_{G[V-U_{i}]}(v)\le\tilde{O}(\phi_{i})\cdot\deg_{G[U_{i}]}(v)$
for most $v\in U_{i}$. That is, most vertices also have few connection to the outside of the cluster which again matches our intuitive understanding of a good graph
clustering.

We say that the decomposition $\uset$ has slack $s \ge 1$ if we relax Property \ref{enu:linked} as follows:
for all $i$, $G[U_{i}]^{\alpha/\phi_{i}}$ is a $(\phi_{i}/s)$-expander. We say that $\uset$ has no slack if $s=1$.

{} 

\paragraph{The Expander Hierarchy.}
Let $\uset$ be an $(\alpha,\phi)$-boundary-linked expander decomposition
of $G$. Suppose that we contract each cluster $U_{i}\in\uset$ into
a vertex where we keep parallel edges and removes self-loops. Let
$G_{\uset}$ denote the contracted graph. Observe that $\vol(G_{\uset})=\sum_{i}|E_{G}(U_{i},V\setminus U_{i})|=\tilde{O}(\phi m)$.
So $\vol(G_{\uset})\ll\vol(G)$ for small enough $\phi$. Repeating
the process of decomposition and contraction leads to the following
definition. A sequence of graphs $(G^{0},\dots,G^{t})$ is an \emph{$(\alpha,\phi)$-expander
	decomposition sequence} of $G$ if $G^{0}=G$, $G^{t}$ has no edges,
and for each $i\ge0$, there is an $(\alpha,\phi)$-boundary-linked
expander decomposition $\uset_{i}$ where $G^{i+1}=G_{\uset_{i}}^{i}$.

Now, observe that the sequence $(G^{0},\dots,G^{t})$ naturally corresponds
to a tree $T$ where the set of vertices at level $i$ of $T$ corresponds
to vertices of $G^{i}$ (i.e.~$V_{i}(T)=V(G^{i})$) and if a vertex
$u_{i}\in V(G^{i})$ is contracted to a super-vertex $u_{i+1}\in V(G^{i+1})$,
then we add an edge $(u_{i},u_{i+1})\in T$ with weight $\deg_{G^{i}}(u_{i})$.
We call this tree an \emph{$(\alpha,\phi)$-expander hierarchy}, which is
the central object of this paper. We say that $T$ has slack $s$ if each $\uset_i$ from the $(\alpha,\phi)$-expander decomposition sequence has slack at most $s$.

Our main result shows that an expander hierarchy is robust enough to be maintained
under edges updates in subpolynomial update time.

\begin{thm}
\label{thm:main}There is an algorithm that, given an $n$-vertex unweighted graph $G$ undergoing edge insertions and deletions, explicitly maintains, with high probability, a $(1/\polylog(n),2^{-O(\log^{3/4}n)})$-expander hierarchy of $G$ with depth $O(\log^{1/4}n)$ and slack $2^{-O(\log^{1/2}n)}$ in $2^{O(\log^{3/4}n)}$ amortized update time. The algorithm works against an \textbf{adaptive} adversary.
\end{thm}

The algorithm in Theorem~\ref{thm:main} can be both \emph{derandomized} and \emph{deamortized}
with essentially the same guarantee up to subpolynomial factors. We note however that the deamortized algorithm does not explicitly maintain the hierarchy, but supports queries of the following form: given a vertex $u$ of $G$, return 
 a leaf-to-root path of $u$ in the hierarchy
in $O(\log^{1/4}n)$ time.\footnote{This kind of guarantee is similar to the dynamic matching algorithm
by \cite{BernsteinFH19} with worst-case update time.} 

\subsection{Applications}
\label{subsec:applications}
The dynamic algorithm for maintaining an expander hierarchy in \Cref{thm:main}
and its derandomized and deamortized counterpart immediately imply a number
of applications in dynamic graph algorithms. 
\Cref{tab:application} shows that the high-level algorithm for each application can be described in only one or two sentences.

Below, we discuss the contribution of each application. 
We say a dynamic algorithm is \emph{fully dynamic} if it handles both
edge insertions and deletions. Otherwise, it is \emph{incremental}
or \emph{decremental}, meaning that it handles only insertions
or only deletions of the edges, respectively. 

\begin{table}
\begin{tabular}{|>{\raggedright}p{0.33\textwidth}|>{\raggedright}p{0.65\textwidth}|}
\hline 
\textbf{Applications}  & \textbf{How to obtain from the expander hierarchy $T$ }\tabularnewline
\hline 
\hline 
Tree flow sparsifier & Return $T$ itself.\tabularnewline
\hline 
Vertex cut sparsifier w.r.t.~a terminal set $C$  & Return the union of root-to-leaf paths of $T$ over all vertices $u\in C$.
Denoted it by $T_{C}$.\tabularnewline
\hline 
$s$-$t$ max flow, $s$-$t$ min cut, multi-commodity cut/flow with
a terminal set $C$  & Solve the problem on $T_{\{s,t\}}$ or $T_{C}$, i.e.~the vertex
cut sparsifier defined in the line above.\tabularnewline
\hline 
Conductance and sparsest cut  & Implement the $\mathsf{Top}$ $\mathsf{Tree}$ data structure on $T$.\tabularnewline
\hline 
Pairwise connectivity & Given $u$ and $v$, check if the roots of $u$ and of $v$ in $T$
are the same.\tabularnewline
\hline 
Treewidth decomposition on constant degree graphs  & Return $T$ itself. For each level-$i$ node $x\in V_{i}(T)$ that
corresponds to a vertex $u_{i}\in V(G^{i})$, the bag $B_{x}$ contains
the original endpoints in $G$ of the edges incident to $u_{i}$ in
$G^{i}$.\tabularnewline
\hline 
\end{tabular}

\caption{Applications of an expander hierarchy $T$ of depth $\protect\dep(T)$, that originates from the expander decomposition sequence $(G^{0},\dots,G^{\protect\dep(T)})$ of $G$. All problems are on unweighted graphs. \label{tab:application}}
\end{table}

\paragraph{Tree Flow Sparsifiers.}
Our first application is the \emph{first} non-trivial fully-dynamic
(or even decremental) algorithm for \emph{tree flow sparsifiers},
which will be used to obtain many other applications in the paper. Intuitively, a tree
flow sparsifier is a tree that approximately captures the flow/cut structure in a graph. As the formal definition of tree flow
sparsifiers is a bit hard to digest, here we define a simpler and an almost equivalent notion of \emph{tree cut sparsifiers. }A\emph{ tree
cut sparsifier $T$ }of a graph $G=(V,E)$ with quality $q$ is a
weighted rooted tree such that (i) the leaves of $T$ corresponds to the vertex set $V$ of $G$, and (2) for any pair $A,B$ of disjoint subsets of $V$, $\mincut_{T}(A,B)\le\mincut_{G}(A,B)\le q\cdot\mincut_{T}(A,B)$
where, for any graph $H$, $\mincut_{H}(A,B)$ denotes the value of a minimum cut separating $A$ from $B$ in $H$. 

Tree flow sparsifiers have been extensively studied in the static
setting \cite{BartalL99,Racke02,HarrelsonHR03,BienkowskiKR03,RackeS14,RackeST14}
and have found many applications in approximation algorithms \cite{AndreevGGMM09,BansalFKMNNS14,ChekuriKS13,Racke08,RackeS14}
and fast algorithms for computing max-flow \cite{Sherman13,KelnerLOS14}. Currently,
the fastest algorithm for computing a tree flow sparsifier takes $\tilde{O}(m)$ time to produce a sparsifier of quality $O(\log^{4}n)$
with high probability \cite{RackeST14,Peng16}. 
However, only little progress has been obtained in the dynamic setting. 
Very recently, Goranci,
Henzinger, and Saranurak \cite{GoranciHS19inc} show that, by calling
the static algorithm of \cite{RackeST14} in a blackbox manner, they
can obtain an \emph{incremental} algorithm\footnote{They also show the same trade-off for the weaker \emph{offline} fully
	dynamic setting where the whole sequence of updates and queries is
	given from the beginning.}  for maintaining tree flow
sparsifiers with $\log^{O(\ell)}n$ quality in $\tilde{O}(n^{1/\ell})$
worst-case update time, for any $\ell>1$, but their technique inherently could not handle edge deletions.
The major reason for this lack of progress in dynamic algorithms is the fact that all existing static constructions for tree flow sparsifiers work in a \emph{top-down} manner, which is difficult to dynamize.

We show that an expander hierarchy is itself a tree flow sparsifier. In particular, this  hierarchy implies the first \emph{static}
tree flow sparsifier based on a \emph{bottom-up} clustering algorithm and is arguably the conceptually simplest of all
known constructions. The key feature of this construction is that it can be maintained dynamically, as summarized in the following corollary.

\begin{cor}
\label{thm:tree-flow-sparsifier}There is a fully dynamic deterministic
algorithm on an unweighted $n$-vertex graph that explicitly maintains
a tree flow sparsifier with quality $n^{o(1)}$ in $n^{o(1)}$
amortized update time.
\end{cor}
There has been recent interest in designing dynamic algorithms for maintaining trees that preserve important features or graphs, e.g., distances. One example is the work on dynamic low-stretch spanning trees that achieves sub-polynomial stretch~\cite{ForsterG19,ChechikZ20}, while in the static there are constructions that give nearly logarithmic stretch~\cite{Abraham2012using}. Driven by this, our work can be thought as a first step in understanding dynamic algorithms for maintaining trees that preserve the cut/flow structure of graphs.

\paragraph{Flow/Cut-based Problems.}
Using the above theorem, we improve upon the previous results on a wide-range of dynamic cut and flow problems whose previous fully dynamic (and even
decremental) algorithms either require $\Omega(n)$ update time\footnote{This includes the incremental exact max flow algorithm by \cite{GuptaK18}
and the dynamic conductance algorithm by \cite{BrandNS19}.} or $\Omega(n)$ query time\footnote{This is by using dynamic graph sparsifiers \cite{AbrahamDKKP16} and
running static algorithms on top of the sparsifier.}, as summarized in the following corollary.
\begin{cor}
\label{cor:cut and flow}There is a fully dynamic deterministic algorithm
on an unweighted $n$-vertex graph with $n^{o(1)}$ update
time that can return an $n^{o(1)}$-approximation to queries of the following problems:
\begin{enumerate}[noitemsep,topsep=0pt]
\item \label{maxflow} $s$-$t$ maximum flow, $s$-$t$ minimum cut;
\item \label{cut}lowest conductance cut, sparsest cut; and
\item\label{multi} multi-commodity flow, multi-cut, multi-way cut, and vertex cut sparsifiers.
\end{enumerate}
For problems in \ref{maxflow} and \ref{cut}, the query time is $O(\log^{1/6}n)$, while for problems in \ref{multi}, it is $O(|C|\log^{1/6}n)$ where $C$ is the terminal set of the respective problem.
For problems in \ref{maxflow} and \ref{multi}, the update time is worst-case, while it is amortized for problems in \ref{cut}.
\end{cor}

Previous sub-linear time algorithms are known only for the incremental
setting. There are incremental $\log^{O(\ell)}n$-approximation algorithms
with $O(n^{1/\ell})$ update time for all the problems
in \Cref{cor:cut and flow} \cite{GoranciHS19inc}.

Although our approximation ratio of $n^{o(1)}$ is moderately high, we believe that our results might serve as an efficient building block to $(1+\epsilon)$-approximation dynamic max flow and minimum cut algorithms,
analogous to the previous development in the static setting: A fast
static $n^{o(1)}$-approximate (multi-commodity) max flow algorithm
was first shown by Madry \cite{Madry10}, and later, the $(1+\epsilon)$-approximate
algorithms were devised \cite{Sherman13,KelnerLOS14,Peng16,Sherman17}
by combining Madry's technique with the \emph{gradient-descent-based
method}. Although the gradient-descent-based method in the dynamic
setting is currently unexplored, we hope that our result will motivate
further investigation in this interesting direction.

\paragraph{Connectivity: Bypassing the NSW Framework.}
Very recently, Chuzhoy et al.~\cite{ChuzhoyGLNPS19det} combine their
new balanced cut algorithm with the framework of Nanongkai, Saranurak,
and Wulff-Nilsen (NSW) \cite{NanongkaiSW17}, and obtain a deterministic
dynamic connectivity algorithm with $n^{o(1)}$ worst-case update
time, answering a major open problem of the field. 

Here, we show a significantly simplified algorithm which completely bypasses 
the complicated framework by \cite{NanongkaiSW17}.
Our algorithm simply follows from the observation
that a graph $G$ is connected iff the top level of an expander hierarchy
$T$ of $G$ contains only one vertex. Also, two vertices $u$ and
$v$ are connected iff the roots of $u$ and $v$ in $T$ are the
same. Interestingly, our algorithm is the first algorithm for dynamic connectivity problem that does \emph{not} explicitly maintain a spanning forest.
\begin{cor}
\label{cor:conn}There is a fully dynamic deterministic algorithm
on an $n$-vertex graph $G$ that maintains connectivity of $G$ using
$n^{o(1)}$ worst-case update time and also supports pairwise connectivity
in $O(\log^{1/6}n)$ time.
\end{cor}

\paragraph{Treewidth.}
Computing a \emph{treewidth decomposition} with approximately minimum
width is a core problem in the area of fixed-parameter tractable algorithms
\cite{RobertsonS95b,Bodlaender96,Reed92,FeigeM06,Amir10,BodlaenderDDFLP16,FominLSPW18}.
We observe that, on constant degree graphs, an expander hierarchy
itself gives a treewidth decomposition. Hence, we obtain the first
non-trivial dynamic algorithm for this problem.\footnote{Dvorak, Kupec, and Tuma \cite{DvorakKT13} show a fully dynamic algorithm
for maintaining a \emph{treedepth decomposition}, which is closely
related to a treewidth decomposition. Let $\text{td}(G)$ denote the
\emph{treedepth} of a graph $G$. It is known that $\tw(G)\le\text{td}(G)\le O(\tw(G)\log n)$
\cite{BodlaenderGHK95}. However, the update time of the algorithm
\cite{DvorakKT13} is proportional to a tower of height $\text{td}(G)$,
which is super-linear when $\text{td}(G)=\omega(\log^{*}n)$. So this
algorithm might take super-linear time even when $\tw(G)=O(1)$.} Our result is summarized in the following corollary, where $\tw(G)$ is the treewidth of a graph $G$ (i.e.~the
minimum width over all tree decomposition of $G$). 
\begin{cor}
\label{cor:treewidth}There is a fully dynamic deterministic algorithm
on a constant degree $n$-vertex graph $G$ that explicitly maintains
a treewidth decomposition of width $\tw(G)\cdot n^{o(1)}$ in $n^{o(1)}$ amortized update time.
\end{cor}

\subsection{Comparison of Techniques}

\label{sec:related} 

The expander hierarchy is strictly stronger than the so-called \emph{low-diameter
	hierarchy} appeared in the algorithms for constructing low-stretch spanning trees in both the static \cite{AlonKPW95} and dynamic setting \cite{ForsterG19,ChechikZ20}.
The low-diameter hierarchy is similar to the expander hierarchy, except
that each cluster is only guaranteed to have low diameter. Structurally,
the expander hierarchy is strictly stronger since every $\phi$-expander
automatically has low diameter $\tilde{O}(1/\phi)$, but some low-diameter
graph has very bad expansion. This is why the low-diameter hierarchy
could not be applied to cut/flow-based problems. Algorithmically,
previous approaches for maintaining the low-diameter hierarchy \cite{ForsterG19,ChechikZ20}
inherently have amortized update time guarantee and assume an \emph{oblivious}
adversary, while our algorithm using the expander hierarchy can be
made worst-case and deterministic.

Expander decomposition has almost never been used in a hierarchical
manner. Many algorithms perform the decomposition only once \cite{Trevisan05,ChekuriKS05,ChekuriKS13}
and some recursively decompose the graph by removing all edges inside
clusters, instead of contracting each cluster (e.g.~\cite{SpielmanT11-SecondJournal,AndoniCKQWZ16,JambulapatiS18,ChangPZ19,EdenFFKO19}
for static algorithms and~\cite{bernstein2020fullydynamic,ChuzhoyK2019new}
for dynamic ones). An exception is the sensitivity connectivity oracle
by Patrascu and Thorup~\cite{PatrascuT07}, which decomposes a graph
into a certain hierarchy of expanders. Unfortunately, their hierarchy
is not robust and inherently can handle only a single batch
of updates, so it does not work in the standard dynamic setting.

The dynamic connectivity and minimum spanning forest algorithm by
Nanongkai et al.~\cite{NanongkaiSW17} repeatedly applies the expander
decomposition and has a bottom-up flavor as in ours, but their underlying
structure does \emph{not} actually yield a hierarchy. More specifically,
while each cluster in our expander hierarchy is contracted into a single vertex in the next level,
their cluster can only be ``compressed'' into a smaller set, which
might even be cut through in the next level. This leads to a much
more complicated structure and requires an ad hoc treatment. The similar
structure appears in the very recent work on dynamic $c$-edge connectivity
by Jin and Sun~\cite{JinS2020fully}. We bypass such complication
via the boundary-linked expander decomposition and obtain the simplified
dynamic connectivity algorithms and other applications. We expect
that our clean hierarchy will be easy to work with and lead to more
interesting applications in the future.

Very recently, the concurrent work of Chen et al.~\cite{ChenGHPS20vertexsparsifier}
shows how to dynamize several known constructions of vertex sparsifiers
for various problems. One of their applications is a fully dynamic
algorithm for $s$-$t$ maximum flow and minimum cuts. Their algorithm
works against an oblivious adversary, has $\Ot(n^{2/3})$ amortized
update time and, given a query, returns an $O(\log n(\log\log n)^{O(1)})$-approximation
in $\Ot(n^{2/3})$ time. They also extend the algorithm to work against
an adaptive adversary while supporting updates and queries in $\tilde{O}(m^{3/4})$
time. Comparing with Item~\ref{maxflow} from \Cref{cor:cut and
	flow}, our algorithm has $n^{o(1)}$ worst-case update time and $O(\log^{1/6}n)$
query time and is deterministic, but our approximation factor is worse.

\section{Technical Overview}

\label{sec:overview}This overview is divided into two parts. In \Cref{sec:overview_capture} we show that an expander hierarchy is itself a tree flow sparsifier and faithfully captures the cut/flow structure of a graph. 
In \Cref{sec:overview_robust} we show how to maintain an expander hierarchy under dynamic edge updates. 
Below, we will write \emph{$(\alpha,\phi)$-decomposition}
as a shorthand for $(\alpha,\phi)$-boundary-linked expander decomposition.

\subsection{Tree Flow Sparsifiers}
\label{sec:overview_capture} 
We start by showing how to construct a tree flow sparsifier for an expander, a very special case. 
Along the way to generalize the idea to general graphs, we will see
how expander hierarchies arise naturally. We will also explain why
the approaches based on the standard expander decomposition or even
the slightly stronger decomposition from \cite{SaranurakW19expander}
fail.

Tree flow sparsifiers can be informally defined as follows. Let $G=(V,E)$
be an $n$-vertex $m$-edge graph. Let $D:V\times V\rightarrow\mathbb{R}_{\ge0}$
be a demand (for multi-commodity flow) between pairs of vertices in $G$.
We say that $D$ can be routed with congestion $\eta$ in $G$ if
there is a multi-commodity flow that routes $D$ with congestion $\eta$.
If $\eta=1$, we say that $D$ is \emph{routable} in $G$. A tree $T$ is
a \emph{tree flow sparsifier of $G$ with quality $q$}, iff (i) any
routable demand in $G$ is routable in $T$, and (ii) any routable
demand in $T$ can be routed with congestion $q$ in $G$. 

\paragraph{Special Case: Expanders.}
Suppose $G$ is a $\phi$-expander. We can construct a very simple
tree flow sparsifier $T$ with quality $q_{1}=O(\log (m)/\phi)$ as
follows. Let $T$ be a star with a root $r$ that connects to each vertex $v\in V$ with an edge $(r,v)$ of capacity $\deg_{G}(v)$.
Observe that any routable demand in $G$ is routable in $T$ because
of the way we set the capacities of edges in $T$. On the other hand,
if a demand $D$ is routable in $T$, then the total demand on each
vertex $v$ is at most $\deg(v)$. But it is well-known from the multi-commodity
max-flow/min-cut theorem \cite{LeightonR99} that, on $\phi$-expanders,
any demand $D$ with such a property can be routed with congestion $O(\log (m)/\phi)$.

\paragraph{Intermediate Case: Two Levels of Expanders.}
Next, we suppose that $G$ satisfies the following: there is a partition
$\uset=(U_{1},\dots U_{k})$ of $V$ such that $G[U_{i}]^{1}$ is
a $\phi$-expander for each $U_{i}\in\uset$ and the contracted graph
$G_{\uset}$ is also a $\phi$-expander. We write $V(G_{\uset})=\{u_{1},\dots,u_{k}\}$.
We can naturally construct a tree $T$ corresponding to the partition
$\uset$ as follows. The set of level-$0$ vertices of $T$ is $V_{0}(T)=V(G)$.
The set of level-$1$ vertices is $V_{1}(T)=V(G_{\uset})$. The level-$2$
contains only the root $r$ of $T$. For each pair $v\in V_0(T)$ and $u_i\in V_1(T)$ such that $v\in U_{i}$, we add an edge $(v,u_{i})$ with capacity $\deg_{G}(v)$. And each $u_{i}\in V_1(T)$, we add an edge $(u_{i},r)$ with capacity $\deg_{G_{\uset}}(u_{i})=|E(U_{i},V\setminus U_{i})|$.
We claim that $T$ has quality $q_{2}=O((\log(m)/\phi)^{2})$. 

By the choice of edge capacity in $T$, any routable demand in $G$
is routable in $T$. For the other direction, suppose that a demand
$D$ is routable in $T$, then we will show a multi-commodity flow
that route $D$ in $G$ with congestion $O((\log(m)/\phi)^{2})$. The
idea is to first consider the \emph{projected demand} $D_{\uset}$
where $D_{\uset}(u_{i},u_{j}):=\sum_{x\in U_{i},y\in U_{j}}D(x,y)$.
Using the argument from the first case, there is a multi-commodity
flow $F_{\uset}$ that routes $D_{\uset}$ in $G_{\uset}$ with congestion
$q_{1}$. Our goal is to extend this flow $F_{\uset}$ to another
flow that routes $D$ in $G$ with congestion $O((\log(m)/\phi)^{2}).$

For each vertex $u_{i}\in V(G_{\uset})$, consider the flow paths
in $F_{\uset}$ going through $u_{i}$ in $G_{\uset}$. Observe that
these paths corresponds to a demand $D_{U_{i}}$ between \emph{boundary
	edges} of $U_{i}$ (i.e.~$E(U_{i},V\setminus U_{i})$). Our main
task now is to route $D_{U_{i}}$ within $G[U_{i}]$. Once we have
obtain the flow $F_{i}$ within $G[U_{i}]$ that routes $D_{U_{i}}$
for all $U_{i}$, this would extend $F_{\uset}$ to a flow in $G$
as desired, and we are essentially done (some details are omitted here).

As $G[U_{i}]^{1}$ is a $\phi$-expander, the max-flow/min-cut theorem
implies that any demand $D'$ between the boundary edges of $U_{i}$
can be routed \emph{within} $G[U_{i}]$ with congestion $O(\log(m)/\phi)$
as long as the total demand of $D'$ on each edge is at most $1$.
However, the total demand of $D_{i}$ on each boundary edge of $U_{i}$
can be as large as $q_{1}$ (since the flow $F_{\uset}$ causes congestion $q_{1}$).
Therefore, $D_{i}$ can be routed in $G[U_{i}]$ with congestion $q_{1}\cdot O(\log(m)/\phi)=O((\log(m)/\phi)^{2})$.
As this holds for all $U_{i}$, the tree $T$ has quality $O((\log(m)/\phi)^{2})$.
Note that we crucially exploit the conductance bound on $G[U_{i}]^{1}$
and not just on $G[U_{i}]$. From the above discussion, the standard expander decomposition cannot give the partition $\uset$ as we need. 

An important observation is that, this quality can be improved if
we are further promised that, for each $U_{i}\in\uset$, $G[U_{i}]^{w}$
is a $\phi$-expander, for some $w>1$. This promise implies that we
could route a demand $D'$ between the boundary edges of $U_{i}$
within $G[U_{i}]$ with congestion $O(\log(m)/\phi)$ as long as the
total demand of $D'$ on each edge is at most $w$ (instead of $1$),
so the demand $D_{i}$ can be routed in $G[U_{i}]$ with congestion
$q_{1}\cdot O(\frac{\log(m)/\phi}{w})$. Therefore, if $\uset$ is
an $(\alpha,\phi)$-decomposition which guarantees that $G[U_{i}]^{\alpha/\phi}$
is a $\phi$-expander, the quality of $T$ will be $q_{1}\cdot O(\frac{\log(m)}{\alpha})$.\footnote{Actually, we have a guarantee that $G[U_{i}]^{\alpha/\phi_{i}}$ is
	a $\phi_{i}$-expander for some $\phi_{i}\ge\phi$. This implies the
	same bound of $q_{1}\cdot\frac{O(\log m/\phi_{i})}{w}=q_{1}\cdot O(\frac{\log(m)}{\alpha}).$} That is, we lose only a factor of $O(\log(m)/\alpha)$ per level,
instead of $O(\log(m)/\phi)$. 

\paragraph{General Case.}
Now, we are ready to consider an arbitrary graph $G$. Let $(G^{0},\dots,G^{t})$
be such that (i) $G^{0}=G$; (ii) $E(G^{t})=\emptyset$; and (iii) $G^{i+1}=G_{\uset_{i}}^{i}$
for some $(\alpha,\phi)$-decomposition $\uset_{i}$ of $G^i$. Observe that
if we define a tree from $(G^{0},\dots,G^{t})$ using the same idea
as above, we would exactly obtain an $(\alpha,\phi)$-expander hierarchy
$T$ of $G$. %
{} Let $t$ be the depth of $T$. We can argue inductively that $T$
is a tree flow sparsifier of $G$ with quality $O(\frac{\log m}{\phi})\cdot O(\frac{\log m}{\alpha})^{t-1}$. 

Note that that $t=O(\log_{1/\phi}m)$ by Property \ref{enu:boundary}
of the $(\alpha,\phi)$-decomposition. From \Cref{thm:decomp}, we can compute an
$(\alpha,\phi)$-expander hierarchy where $\alpha=1/\polylog(n)$
and $\phi=2^{-O(\log^{1/2}m)}$, this implies that $T$ has quality
$O(\frac{\log m}{\phi})\cdot O(\frac{\log m}{\alpha})^{t-1} = 2^{O(\log^{1/2}m\log\log m)}=n^{o(1)}$. Note that we need $\phi\ll\alpha$
to obtain the quality of $n^{o(1)}$. For example, if $\alpha=\phi$,
the quality we obtain would be $(\frac{\log m}{\phi})^{t}=\Omega(m)$.
This is the reason why we cannot use the expander decomposition algorithm
by \cite{SaranurakW19expander}, because their algorithm only returns
a (weaker version of) $(\phi,\phi)$-decomposition.

In the dynamic setting, our algorithm from \Cref{thm:main} maintains
an $(\alpha,\phi)$-expander hierarchy $T$ that has small slack $s$
where $\alpha=1/\polylog(n)$, $\phi=1/2^{O(\log^{3/4}m)}$ and $s=2^{O(\log^{1/2}m)}$.
Following the same analysis, the final quality of $T$ degrades slightly
to $O(\frac{s\log m}{\alpha})^{t-1}\cdot O(\frac{s\log m}{\phi})=2^{O(\log^{3/4}m)}=n^{o(1)}$.

\subsection{Robustness Against Updates}

\label{sec:overview_robust}

In this section, we show how an $(\alpha,\phi)$-expander hierarchy
$T$ with small slack $s$ can be maintained in $n^{o(1)}$ update
time, where $\alpha=1/\polylog(n)$, $\phi=2^{-O(\log^{3/4}n)}$,
and $s=2^{O(\log^{1/2}n)}$. Recall from above that $T$ is a tree
flow sparsifier with quality $2^{O(\log^{3/4}m)}$. 

Our goal is to illustrate that, because of the right kind of guarantees from the $(\alpha,\phi)$-decomposition, the algorithm for maintaining the expander hierarchy can be obtained quite naturally, especially for people familiar with standard techniques in dynamic algorithms. 

\paragraph{Reduction to One Level.}
An $(\alpha,\phi)$-expander hierarchy $T$ corresponds to an $(\alpha,\phi)$-expander
decomposition sequence $(G^{0},\dots,G^{t})$. In particular, for
each $i\ge0$, $G^{i+1}=G_{\uset_{i}}^{i}$ is obtained from $G^{i}$
by contracting each cluster of an $(\alpha,\phi)$-decomposition $\uset_{i}$.
Therefore, the problem of maintaining an expander hierarchy reduces
to maintaining an $(\alpha,\phi)$-decomposition $\uset$ and $G_{\uset}$
on a dynamic graph $G$. There are two important measures:
\begin{itemize}
	\item \textbf{Update Time:} The time for computing the updated $\uset$
	and $G_{\uset}$.
	\item \textbf{Recourse:} The number of edge updates\emph{ }to $G_{\uset}$. 
\end{itemize}
Suppose that there is an algorithm with $\tau$ (amortized) update
time and $\rho$ (amortized) recourse. This would imply an algorithm
for maintaining an $(\alpha,\phi)$-expander decomposition sequence
$(G^{0},\dots,G^{t})$ with $O(\rho^{t}\cdot\tau)$ (amortized) update
time, because the number of updates can be multiplied by $\rho$ per
level. We note that the depth $t=O(\log_{1/\phi}m)$ so we need $\rho=(1/\phi)^{o(1)}$
and $\tau=n^{o(1)}$ to conclude that the final update time is $n^{o(1)}$. 

\paragraph{Two Key Tools.}
From now, we focus on a dynamic graph $G$ and how to maintain an
$(\alpha,\phi)$-decomposition $\uset$ of $G$ with small slack.
To do this, we need two algorithmic tools. First, \Cref{thm:decomp} gives
a static algorithm for computing an $(\alpha,\phi)$-decomposition
$\uset$ of a graph $G$ with \emph{no slack} in time $\tilde{O}(m/\phi)$.
This algorithm strengthens the previous expander decomposition by Saranurak and Wang~\cite{SaranurakW19expander}.

Our second tool is the new expander pruning algorithm from \Cref{thm:dynamic_expander_pruning} with
the following guarantee. Suppose $G[U]^{w}$ is a $\phi$-expander
where $w\le1/(10\phi)$. Suppose there is a sequence of $k\le\phi\vol(U)/2000$
edge updates to $U$ (i.e.~these edges have at least one endpoint
in $U$). Then, the algorithm maintains a small set $P\subseteq U$
such that $G[U\setminus P]^{w}$ is still a $(\phi/38)$-expander in total
time $\tilde{O}(k/\phi^{2})$. More precisely, after the $i$-th update
to $U$, we have $\vol_{G}(P)=O(i/\phi)$ and $|E_{G}(P,U\setminus P)|=O(i)$.
\Cref{thm:dynamic_expander_pruning} generalizes the previous expander pruning algorithm by \cite{SaranurakW19expander}
that works only when $w=1$.

\paragraph{A Simple Algorithm with Too Large Recourse.}
Both tools above suggest the following simple approach. First, we compute
an $(\alpha,\phi)$-decomposition $\uset=\{U_{1},\dots,U_{k}\}$ of
$G$ with no slack where each $G[U_{i}]^{\alpha/\phi_{i}}$ is a $\phi_{i}$-expander.
Next, given an edge update to $U_{i}$, we maintain a pruned set $P_{i}\subseteq U_{i}$
so that $G[U_{i}\setminus P_{i}]^{\alpha/\phi_{i}}$ is a $(\phi_{i}/38)$-expander.
We update $\uset$ by adding the singleton sets $\{\{u\}\}_{u\in P_{i}}$
and replacing $U_{i}$ by $U_{i}\setminus P_{i}$. 

For simplicity, we assume that all edge updates
have at least one endpoint in the cluster $U_{1}\in\uset$. Furthermore,
assume that there are less than $k_{1}\ll\phi^{2}\vol_G(U_{1})$ updates
to $U_{1}$. With this assumption, the updated $\uset$ is an $(\alpha,\phi)$-decomposition
of the updated $G$ with slack $38$. To see this, observe that the number of new
inter-cluster edges is at most $\vol_G(P_{1})$, and so the total number
of inter-cluster edges becomes $\tilde{O}(\phi m)+\vol_G(P_{1})=\tilde{O}(\phi m)$
satisfying Property \ref{enu:boundary}. Let $U_{1}'=U_{1}\setminus P_{1}$.
We have that $G[U'_{1}]^{\alpha/\phi_{1}}$ is a $(\phi_{i}/38)$-expander
by the expander pruning algorithm. This satisfies Property \ref{enu:linked}.
For Property \ref{enu:sparse cut}, observe that $|E_G(U'_{1},V\setminus U'_{1})|\le|E_G(U{}_{1},V\setminus U{}_{1})|+|E_{G}(P_{1},U_{1}\setminus P_{1})|=\tilde{O}(\phi_{i}\vol_G(U_{1}))$
because $|E_{G}(P_{1},U_{1}\setminus P_{1})|\le O(k_{1})$. Note that
all singleton clusters $\{\{u\}\}_{u\in P_{1}}$ satisfies Properties
\ref{enu:boundary} and \ref{enu:sparse cut} vacuously. Lastly,
the recourse on $G_{\uset}$ is $O(\vol_G(P_{1}))=O(k_{1}/\phi_{1})$.

Therefore, we have that the amortized recourse and update time are
$O(1/\phi)$ and $\tilde{O}(1/\phi^{2})$ respectively. While both
of them seem small for maintaining a one-level decomposition, the
recourse is in fact too large if the expander hierarchy has many levels. 
After composing the algorithm for $t$ levels
using the reduction in the beginning of this section, the update time
is at least $\Omega(1/\phi)^{t}=m^{\Omega(1)}$, which is too large
for us.

Previous dynamic algorithms with hierarchical structure have faced
the same issue. This is why the dynamic low-stretch spanning trees
algorithm of \cite{ForsterG19} has $O(\sqrt{n})$ update time. Chechik
and Zhang \cite{ChechikZ20} fixed this issue and improved the update
time of \cite{ForsterG19} to $n^{o(1)}$. However, they only require
each cluster $U\in\uset$ to have a small diameter of $\tilde{O}(1/\phi)$,
which is a much weaker guarantee than being a $\Omega(\phi)$-expander.
Unfortunately, their technique is specific to this weaker guarantee
(and is also inherently amortized). In the dynamic minimum spanning
forests algorithm by Nanongkai et al.~\cite{NanongkaiSW17}, they
require each cluster to be an expander like us, and can only guarantee
$\Omega(1/\phi)$ recourse per update. As we mentioned in \Cref{sec:related},
they fix the issue using an ad-hoc and complicated tool tailored to (minimum) spanning forests. Below, we will see how the $(\alpha,\phi)$-decomposition
allows us to bound the recourse in a simple way.

\paragraph{One-batch Updates.}
First, let us simplify the situation even more by assuming that all
$k_{1}$ updates are simultaneously given to $U_{1}$ in \emph{one
	batch}. We will also need a slight generalization of $(\alpha,\phi)$-decomposition
defined on a subset of vertices here, instead of the whole graph.
For a set $P\subseteq V$, an \emph{$(\alpha,\phi)$-decomposition
	of $P$ in $G$ }is a partition $\uset'=\{U'_{1},\dots,U'_{k}\}$
of $P$ that satisfies the properties of the $(\alpha,\phi)$-decomposition,
except that Property \ref{enu:boundary} is now $\sum_{i=1}^{k}|E(U'_{i},V\setminus U'_{i})|\le O(|E(P,V\setminus P)|)+\tilde{O}(\phi\vol_G(P))$.
Note that the term $\tilde{O}(\phi\vol_G(P))=\tilde{O}(\phi m)$ as
before when $P=V$ and the term $O(|E(P,V\setminus P)|)$ is unavoidable. 

Now, we describe the algorithm. Given the batch of $k_{1}$ updates,
we compute the pruned set $P_{1}\subseteq U_{1}$. Then, we compute
an $(\alpha,\phi)$-decomposition of $P_{1}$ in $G$ and obtain a
partition $\uset'=\{U'_{1},\dots,U'_{k'}\}$ of $P_{1}$. Finally,
we replace $U_{1}$ in $\uset$ with $\{U_{1}\setminus P_{1},U'_{1},\dots,U'_{k'}\}$.
It follows from the description that the updated $\uset$ is an $(\alpha,\phi)$-decomposition
of $G$ with slack at most $38$.

The key step is to bound the total recourse, which is at most 
\[
\sum_{i=1}^{k}|E(U'_{i},V\setminus U'_{i})|\le O(|E(P_{1},V\setminus P_{1})|)+\tilde{O}(\phi\vol_G(P_{1})).
\]
Recall that, by the pruning algorithm, $\vol_G(P_{1})\le O(k_{1}/\phi)$
and $|E_G(P_{1},U\setminus P_{1})|\le O(k_{1})$. So it remains to bound
$|E_G(P_{1},V\setminus U_{1})|$. This is where we exploit Property
\ref{enu:linked} of the $(\alpha,\phi)$-decomposition. Before any
update, observe that 
\[
|E_G(P_{1},U_{1}\setminus P_{1})|\ge\phi_{1}\vol_{G^{\alpha/\phi_{1}}[U_{1}]}(P_{1})\ge\phi_{1}\cdot\frac{\alpha}{\phi_{1}}|E_G(P_{1},V\setminus U_{1})|
\]
where the first inequality is because $G[U_{1}]^{\alpha/\phi_{1}}$
was a $\phi_{1}$-expander and we assume $\vol_{G[U_{1}]^{\alpha/\phi_{1}}}(P_{1})\le\vol_{G[U_{1}]^{\alpha/\phi_{1}}}(U_{1}\setminus P_{1})$
(as $k_{1}$ is small enough), and the second inequality is by the
definition of $G[U_{1}]^{\alpha/\phi_{1}}$. So we have $|E_G(P_{1},V\setminus U_{1})|\le O(k_{1}/\alpha)$
after the updates, because $|E_G(P_{1},U\setminus P_{1})|\le O(k_{1})$
and there are $k_{1}$ updates. This implies that the total recourse
is $O(k_{1}/\alpha)+\tilde{O}(k_{1})$, which is $\tilde{O}(1/\alpha)$
amortized. 

\Cref{thm:decomp} shows that the decomposition $\uset'$ of $P_{1}$ can be computed
in $\tilde{O}(|E_G(P_{1},V\setminus P_{1})|/\phi^{2}+\vol_G(P)/\phi)=\tilde{O}(k_{1}/(\alpha\phi^{2}))$
time. Also, the total time for pruning is $\tilde{O}(\vol_G(P)/\phi)=\tilde{O}(k_{1}/\phi^{2})$.
So the amortized update time is $\tilde{O}(1/(\alpha\phi^{2}))$.
Plugging these bounds into the reduction, the update time for maintaining
an $(\alpha,\phi)$-expander hierarchy is $\tilde{O}(1/\alpha)^{t}\cdot\tilde{O}(1/(\alpha\phi^{2}))=n^{o(1)}$.

\paragraph{Removing Assumptions.}
In fact, in the analysis above, we did not require a strong upper
bound on $k_{1}\ll\phi^{2}\vol_G(U_{1})$, but we did require $k_{1}\le\phi_{1}\vol_G(U_{1})/2000$
because the expander pruning algorithm can handle that many updates
and we also need $\vol_{G[U_{1}]^{\alpha/\phi_{1}}}(P_{1})\le\vol_{G[U_{1}]^{\alpha/\phi_{1}}}(U_{1}\setminus P_{1})$.
This requirement can be removed as follows. If $k_{1}>\phi_{1}\vol_G(U_{1})/2000$,
then we just ``reset'' the cluster $U_{1}$ by computing an $(\alpha,\phi)$-decomposition
$\uset'$ of $U_{1}$ in $G$. Then, we remove $U_{1}$ from $\uset$
and add the new clusters in $\uset'$ to $\uset$. The key point is
again to bound the recourse which is $\sum_{i=1}^{k}|E_G(U'_{i},V\setminus U'_{i})|\le O(|E_G(U_{1},V\setminus U_{1})|)+\tilde{O}(\phi\vol_G(U_{1}))=\tilde{O}(k_{1})$.
This is where we exploit Property \ref{enu:sparse cut} of the $(\alpha,\phi)$-decomposition,
which says $|E_G(U_{1},V\setminus U_{1})|=\tilde{O}(\phi_{1}\vol_G(U_{1}))=\tilde{O}(k_{1})$.
So the amortized recourse is $\tilde{O}(1)$ in this case.

To remove the assumption that all updates have an endpoint in $U_{1}$,
we simply perform the same algorithm on each cluster $U_{i}\in\uset$.
If the number of updates is larger than $\phi\vol(G)$, we just compute
the $(\alpha,\phi)$-decomposition of the updated graph in $\tilde{O}(\vol(G)/\phi)$
time so that Property \ref{enu:boundary} of the $(\alpha,\phi)$-decomposition
is satisfied. In this case, the amortized recourse and update time
is $\frac{O(\phi\vol(G))}{\phi\vol(G)}=\tilde{O}(1)$ and $\frac{\tilde{O}(\vol(G)/\phi)}{\phi\vol(G)}=\tilde{O}(1/\phi^{2})$
respectively. 

From all cases above, we conclude that, given an $(\alpha,\phi)$-decomposition
$\uset$ of $G$ with no slack and one batch of updates, we obtain
an $(\alpha,\phi)$-decomposition $\uset$ of $G$ with slack $38$
with $\rho=\tilde{O}(1/\alpha)$ amortized recourse and $\tau=\tilde{O}(1/(\alpha\phi^{2}))$
amortized update time. This implies an algorithm for an $(\alpha,\phi)$-expander
hierarchy with slack $38$ and $O(\rho^{t}\tau)=2^{O(\log^{3/4}m)}$
amortized update time.

\paragraph{Sequence of Updates.}
Lastly, we remove the final assumption that the updates are given
in one batch. The main issue arises from the expander pruning algorithms.
Recall that, to bound the recourse, it was enough to bound $O(|E_G(P,V\setminus P)|)+\tilde{O}(\phi\vol_G(P))$
where $P$ is the pruned set of some cluster. However, given a sequence
of updates to \Cref{thm:dynamic_expander_pruning}, the set $P$ is a dynamic set that changes through
time. Although we can bound $O(|E_G(P,V\setminus P)|)+\tilde{O}(\phi\vol_G(P))$
at any point of time, the total recourse throughout the algorithm can become much larger. To
fix this, we use the known trick from \cite{NanongkaiSW17} (Section
5.2.1) so that the resulting pruned set changes in a much more controlled
way. 

At a very high level, let $\psi=2^{O(\log^{1/2}m)}$ and $h=\log_{\psi}m=O(\log^{1/2}m)$.
Our algorithm, called $\mlp$ from \Cref{sec:dynamic_expander_decomposition}, will partition $P=P_{h-1}\cup\dots\cup P_{0}$
into $h$ parts such that, for each $i$, $\vol_G(P_{i})\le\vol_G(P_{i+1})/\psi$
and $P_{i}$ can change only every $\psi^{i}$ updates. In words,
the bigger the part, the less often it changes. At the end, $\mlp$
has the amortized recourse $\rho=O(2^{O(\log^{1/2}m)}/\alpha)$ and
update time $\tau=O(2^{O(\log^{1/2}m)}/\alpha\phi^{2})$. This bound
still implies a dynamic $(\alpha,\phi)$-expander hierarchy with update
time $O(\rho^{t}\tau)=2^{O(\log^{3/4}m)}$. However, with this technique,
the slack of the maintained $(\alpha,\phi)$-decomposition become
$s=38^{h}=2^{O(\log^{1/2}m)}$, which gives \Cref{thm:main}.

\paragraph{Derandomization \& Deamortization.}
The only randomized component in this paper is the cut-matching game
\cite{khandekar2009graph,RackeST14} that is used in our static algorithm
for computing an $(\alpha,\phi)$-decomposition. By plugging in the
new deterministic balanced cut algorithm by Chuzhoy et al.~\cite{ChuzhoyGLNPS19det}
into our framework, this immediately derandomizes the whole algorithm.

We can also make our update time to be worst-case using the standard
``building in the background'' technique (although this technique
prevents us from explicitly maintaining the expander hierarchy). The
reason we are allowed to do this is as follows. The only component
that is inherently amortized is the expander pruning algorithm from
\Cref{thm:dynamic_expander_pruning}. However,
\Cref{thm:dynamic_expander_pruning} is only called by $\mlp$, and we can apply
the ``building
in the background'' technique to each level of the algorithm, so
that the input to \Cref{thm:dynamic_expander_pruning} is always in one-batch (not a sequence of updates) and so the running time is worst-case. For other parts of the
algorithms, it is clear when the algorithm needs to spend a lot of
time to reset or re-preprocess the graph, so we can apply ``building
in the background'' in a straight-forward manner.


        \section{Preliminaries}
\label{sec:prelim}
By default, all logarithms are to the base of $2$. 
Normally we use $n$ to denote the number of nodes of a graph, and use $m$ to denote the number of edges of a graph. Even when we allow parallel edges and self-loops, we will assume in this paper that $m = \poly(n)$. We use $\Ot(\cdot)$ to hide $\polylog(n)$ factors.
 
\paragraph{General Notation.}
Let $G=(V,E)$ be an unweighted graph. 
For a vertex $v\in V$, we denote $\deg_G(v)$ as the number of edges incident to $v$ in $G$. 
For two subsets $A, B\subseteq V$ of vertices, we denote by $E_G(A, B)$ the set of
edges with one endpoint in $A$ and the other endpoint in $B$. 
For a subset $S\subseteq V$, we denote by $E_G(S)$ the subset of edges of $E$ with both endpoints in $S$. 

\makeatletter
\DeclareDocumentCommand{\induced}{ o m o}{%
        \IfNoValueTF{#1}%
            {G}%
            {\@ifmtarg{#1}{G}{#1}}%
        [#2]%
        \IfNoValueTF{#3}%
            {}%
            {\@ifmtarg{#3}{}{^{#3}}}%
}
\makeatother
\def\border{\operatorname{border}}
\def\cut{\operatorname{cut}}
\def\out{\operatorname{out}}

To reduce notational clutter we sometimes use the following shorthand notation
for the cardinality of certain edge-sets:
the cardinality of edges incident to $S\subseteq V$ is denoted with
$\out_G(S):=|E_G(S,V\setminus S)|$;
the \emph{border of $S\subseteq U$ w.r.t.\ $U$} is denoted with
$\border_{G,U}(S):=|E_G(S,V\setminus U)|$; the
\emph{cut of $S\subseteq U$ w.r.t.\ $U$} is denoted with
$\cut_{G,U}(S) := |E_G(S,U\setminus S)|$. We drop the subscript $G$ if the
graph is clear from the context and we write, e.g., $\out_G(v)$ instead
of $\out_G(\{v\})$, i.e., we drop the brackets if the respective set contains
just a single vertex.

For an unweighted graph $G$ and a cluster $S\subseteq V$ we use
$\induced{S}[w]$ to denote the subgraph of $G$ induced by the vertex set $S$
where we add $\lceil w\rceil$ self-loops to a vertex $v\in S$ for every
boundary edge $(v,x)$, $x\notin S$ that is incident to $v$ in $G$. Note that
$\induced{S}$ is just the standard notion of an induced subgraph and that in
the graph $\induced{U}[1]$ the degree of all vertices is the same as in the
original graph $G$ (each self-loop contributes $1$ to the degree of the node
that it is incident to).


Let $T$ be a tree and denote $r$ as its root. We denote by $L(T)$ the set of
leaves of $T$. For each $i\ge 0$, we say that a node $v\in V(T)$ is at the
$i$th level of $T$ if the length of the unique path connecting $v$ to $r$ in
$T$ is $i$. So the root $r$ is at the $0$th level of $T$, and all its children
are at the $1$st level of $T$, and so on. For each $i\ge 0$, we let $V_i(T)$ be
the set of all nodes that lie on the $i$th level of the tree $T$.

\paragraph{Conductance and Expander.}
For a weighted graph $G$ and a subset $S\subseteq V$ of its vertices, we define the \emph{volume} of $S$ in $G$ to be $\vol_G(S)=\sum_{v\in S}\deg_G(v)$. 
We refer to a bi-partition $(S,\overline{S})$ of $V$ by a \emph{cut} of $G$ if both $S$ and $\overline{S}$ are not $\emptyset$, and we define the capacity of the cut to be $|E(S,\overline{S})|$.
The \emph{conductance} of a cut $(S,\overline{S})$ in $G$ is defined to be $\Phi_G(S)=\frac{|E(S,\overline{S})|}{\min\{\vol_G(S),\vol_G(\overline{S})\}}$. 
The conductance of a graph $G$ is defined to be $\Phi_G=\min_{S\subsetneq V,S\ne\emptyset}\Phi_G(S)$. For a real number $\phi>0$, we say that $G$ is a $\phi$-expander if $\Phi_G\ge \phi$.
We will omit the subsript $G$ in the notations above if the graph is clear from the context.

\begin{definition}[Near Expander]
\label{defn:nearlyexp}
Given an unweighted graph $G=(V,E)$ and a subset of vertices $A\subseteq V$, we say that $A$ is a {\em near $\phi$-expander} in $G$ for some real number $\phi>0$ if for all $S\subseteq A$ such that $\vol(S)\leq \vol(A)/2$, we have $|E(S,V\setminus S)| \geq \phi\cdot\vol(S)$.
\end{definition}

\paragraph{Contracted Graph.}
Given an unweighted graph $G=(V,E)$ and a partition $\mathcal{U}=(U_1,\ldots,U_r)$ of its vertices, such that the subgraph $G[U_i]$ is connected for each $1\le i\le r$, we define the graph $G_{\cal{U}}$ to be the \emph{contracted graph} of $G$ by contracting each cluster $U_i$ to a single vertex, while keeping the parallel edges that connect vertices from the same pair of subsets in $\mathcal{U}$. 

\paragraph{(Single-commodity) Flow Notation.}
A flow problem $\Pi = (\Delta, T, c)$ on a graph $G = (V,E)$ consists of (i) a
source function $\Delta: V\to\mathbb{R}^{\ge 0}$, (ii) a sink capacity function
$T: V\to \mathbb{R}^{\ge 0}$, and (iii) an edge capacity function
$\capacity: E\to \mathbb{R}^{\ge 0}$. Specifically, for each node $v\in V$, we
denote $\Delta(v)$ to be the amount of mass that is placed on $v$, and we
denote $T(v)$ to be the capacity of $v$ as a sink. For an edge $e$, the
capacity $\capacity(e)$ limits how much flow can be routed along $e$ in both
directions.

Given a \emph{single-commodity} flow $f$ on $G$, we define its
\emph{edge-formulation} by a function $f: V\times V\to \mathbb{R}$, such that
for any pair $(u,v)$ of nodes with $(u,v)\in E$, $f(u,v)$ equals the total
amount of flow sent from $u$ to $v$ along the edge $(u,v)$ minus the total
amount of flow sent from $v$ to $u$ along the edge $(u,v)$. Note that for all
pairs $(u,v)$ such that $(u,v)\in E$, $f(u,v)=-f(v,u)$, and $f(u,v)=0$ for all
pairs $(u,v)$ with $(u,v)\notin E$. We will also refer to an edge-formulation
$f$ by a flow.
Given a flow problem $\Pi=(\Delta, T, \capacity)$ and a flow $f$ on $G$, for
each node $v\in V$, we define $f_\Delta(v)=\Delta(v)+\sum_{u}f(u,v)$ to be the amount
of mass ending at $v$ after routing the flow $f$ from the initial source
function $\Delta$. We say that $f$ is a \emph{feasible flow} of $\Pi$ if
$|f(u,v)|\le \capacity(u,v)$ for each edge $(u,v)\in E$,
$\sum_{u}f(v,u)\le \Delta(v)$ for each $v\in V$, and $0 \leq f_\Delta(v)\le T(v)$ for each
$v\in V$. 


%

The following subroutine is implicit in \cite{SaranurakW19expander}. They use it as a key subroutine for implementing expander trimming and pruning.
\begin{lemma}[Incremental Flow]  \label{lem: incrementalFlow}
	Given an $m$-edge graph $G=(V,E)$, and a flow problem $(\Delta,T,c)$ on $G$ where (i) each edge has integral capacity $1 \leq c_e \leq c_{\max}$, and (ii) each vertex $v$ can absorb $T(v) = \deg(v)$ mass of flow, there is a deterministic algorithm that maintains an incremental, initially empty set $P \subseteq V$ (i.e., vertices can only join $P$ through time) under a sequence of source-injecting operations of the following form: given $v \in V$, increase $\Delta(v)$. 
	
	At any time, as long as $\sum_{v \in V} \Delta(v) \leq \vol(V)/3$, the algorithm guarantees that  
	\begin{enumerate}
		\item the flow problem $(\Delta',T',c')$ on $G[V \setminus P]^1$ is feasible, where $\Delta'(v) = \Delta(v) + c(E(\{v\},P))$ for all $v \in V$, and $T'$, $c'$ are $T$, $c$ restricted to $V \setminus P$, respectively, and
		\item $\vol(P) \leq 2\sum_{v \in V} \Delta(v)$ and $|E(P, V \setminus P)| \leq \frac{2 \sum_{v \in V} \Delta(v)}{\min_e \{c_e\}}$.
	\end{enumerate}
	The total update time is $O(c_{\max} \sum_{v \in V} \Delta(v) \log m)$.
\end{lemma}
Let us give some intuition about this subroutine. We are given a graph $G$ that undergoes a sequence of ``injecting'' mass operations, after some time the total mass will not be routable (i.e.~the flow problem is not feasible) and get stuck, the above subroutine will maintain a growing set $P$ such that, the mass in the remaining part $G[V\setminus P]^1$ is routable. 
Moreover, this remains feasible even if we inject additional mass through the cut edges $E(P,V\setminus P)$ at full capacity.


        
         \def\ivol{\operatorname{ivol}}
 
\def\gkrv{\gamma_{\textsc{cmp}}}
\def\indGraph #1#2#3{%
  #1[#2]^{#3}
}
\def\U{\mathcal{U}}
\def\expansion{\operatorname{expansion}}

 \def\fI{f_I}
  \def\fB{f_B}


\section{Boundary-Linked Expander Decomposition and Hierarchy}
In this section we formally introduce the notion of a boundary-linked expander
decomposition, which is the main concept of this paper.

\begin{definition}(Boundary-Linkedness)
For a graph $G=(V,E)$ and parameters $\alpha, \phi\in(0,1)$
we say that a cluster $U\subseteq V$ is $(\alpha,\phi)$-boundary-linked
in $G$ if the graph $\induced{U}[\alpha/\phi]$ is a $\phi$-expander.
\end{definition}


Intuitively, the conductance $\Phi_{\induced{U}[1]}$ (i.e., when we choose
$\alpha=\phi$) measures how well the edges of the cluster $U$ (including the
boundary edges $\Gamma_G(U)$) are connected \emph{inside} the cluster.
$\Phi_{\induced{U}[1]}\ge\phi$ means that we can solve an all-to-all
multicommodity flow problem between the edges of $E_G(U,V)$ (i.e.~edges incident to $U$)
inside $\induced{U}$ with congestion at most $\ot(1/\phi)$.\footnote{In an
  all-to-all multicommodity flow problem \emph{between a subset of edges $E'$} in
  $G[U]$, there is a weight $w(v):= \sum_{e\in E'}|\{v\}\cap e|$ assigned to
  every vertex $v\in U$. Then the demand between two vertices $u,v\in U$ is
  $w(u)w(v)/w(U)$.
}
Boundary linkedness with a parameter $\alpha\gg\phi$, means that the boundary
edges themselves have higher connectivity. We can solve an all-to-all
mutlicommodity flow problem between boundary-edges with congestion
$\ot(1/\alpha)$ inside $U$.

Next, we define the notion of a boundary-linked expander decomposition and
that of an expander hierarchy. These are the central definitions in this paper.

\begin{definition}[Boundary-Linked Expander Decomposition]
\label{def:decomp}
Let $G=(V,E)$ be a graph and
$\alpha,\phi\in(0,1)$ be parameters. Let $U\subseteq V$ be a cluster in $G$.

An \emph{$(\alpha,\phi)$-boundary-linked expander decomposition of $U$ in $G$} with slack $s\ge1$
consists of a partition $\uset=\{U_{1},\dots,U_{k}\}$ of $U$ together with a
conductance-bound $\phi_i\ge\phi$ for every $i\in1,\dots,k$ such that the
following holds:
\begin{enumerate}
\item $\sum_{i=1}^{k}\out_G(U_{i})\le O(\out_G(U))+\Ot(\phi\vol_G(U))$.\label{ed:globalboundary}
\item For all $i$: $\induced{U_i}[\alpha/\phi_i]$ is a $(\phi_i/s)$-expander.\label{ed:expanding}
\item For all $i$: $\out_G(U_i)\le \Ot(\phi_i\vol_G(U_i))$.\label{ed:localboundary}
\end{enumerate}
%
%
\end{definition}

When we have an expander-decomposition with slack $1$, we will usually not
mention the slack and just call it an $(\alpha,\phi)$-boundary linked expander
decomposition. The notion of slack will not be important for our static
constructions but only becomes important for maintaining boundary-linked
expander decompositions dynamically. Instead of writing
\enquote{$(\alpha,\phi)$-boundary-linked expander decomposition}, we sometimes
write \enquote{$(\alpha,\phi)$-expander decomposition} or just
\enquote{$(\alpha,\phi)$-ED}.
If $U=V$, then we say that $\uset$ is an $(\alpha,\phi)$-ED of $G$.

\begin{definition}[Expander Decomposition Sequence] Let $G=(V,E)$ be a graph
with $m$ edges and $\alpha,\phi\in(0,1)$ be parameters. We say that a sequence
of graphs $(G^{0},G^{1},\dots,G^{t})$ is an \emph{$(\alpha,\phi)$-expander
  decomposition sequence of $G$ with slack $s$} or $(\alpha,\phi)$-ED-sequence of $G$ if (1)
$G^{0}=G$, (2) $G^{t}$ has no edge, and (3) $G^{i+1}=G_{\uset^{i}}^{i}$ is the
contracted graph of $G^{i}$ w.r.t.\ to $\uset^{i}$ where $\uset^{i}$ is an
$(\alpha,\phi)$-ED of $G^{i}$ with slack $s$.
\end{definition}

\begin{definition}[Expander Hierarchy]	     
An $(\alpha,\phi)$-ED sequence $(G^{0},G^{1},\dots,G^{t})$ naturally
corresponds to a tree $T$ where (1) the set of nodes at level $i$ of $T$ is
$V_{i}(T)=V(G^{i})$ and (2) a node $u_{i}\in V_{i}(T)$ has a parent
$u_{i+1}\in V_{i+1}(T)$ if $u_{i}\in V(G^{i})$ is contracted into the
super-vertex $u_{i+1}\in V(G^{i+1})$. The edge $(u_i,u_{i+1})$ is assigned a
capacity of $\deg_{G^i}(u_i)$. We call $T$ an \emph{$(\alpha,\phi)$-expander
  hierarchy} or $(\alpha,\phi)$-EH (with slack $s$).
\end{definition}

\noindent
The next theorem is the main result in this section. Throughout this section,
we define $\gkrv=O(\log^2m)$. This is a value derived from the approximation
guarantee for sparsest cut of the cut-matching-game \cite{khandekar2009graph}
on an $m$-edge graph.
	
\begin{thm}\label{thm:decomp}
There is a randomized algorithm that, given a graph $G=(V,E)$, a cluster
$U\subseteq V$ with $\vol_{G}(U)=m$ and $\out_G(U)= b$, and
parameters $\alpha,\phi$, with $\alpha\le{1}/(4\gkrv\log_2 m)$ computes an
$(\alpha,\phi)$-ED of $U$ in $\tilde{O}(b/\phi^2+m/\phi)$ time with high
probability. In particular
\begin{enumerate}
\item $\sum_{i=1}^{k}\out_G(U_{i})\le 4\out_G(U)+O(\log^3m\cdot\phi\vol_G(U))$.
\item For all $i$: $\induced{U_i}[\alpha/\phi_i]$ is a $\phi_i$-expander.
\item For all $i$: $\out_G(U_i)\le O(\log^6m\cdot\phi_i\vol_G(U_i))$.
\end{enumerate}
\end{thm}

	As an $(\alpha,\phi)$-ED-sequence and its corresponding $(\alpha,\phi)$-expander
	hierarchy can be naturally computed bottom up given the above algorithm,
	we immediately get the following corollary.
	\begin{cor}\label{cor:hierarchy}
		There is a randomized algorithm that, given a graph $G$ with $m$
		edges and parameters $\alpha,\phi$, with
                $\alpha\le{1}/(2\gkrv\log_2 m)$ computes an $(\alpha,\phi)$-expander
		decomposition sequence of $G$ and its corresponding expander hierarchy
		in $\tilde{O}(m/\phi)$ time with high probability. 
	\end{cor}

\noindent
The remainder of this section is devoted to proving \Cref{thm:decomp}.

\subsection{The Key Subroutine}

The major building block for our algorithm is the following sub-routine that
when applied to a cluster either (1) finds a sparse balanced cut, or (2) finds
a sparse unbalanced cut such that the large side of the cut forms a cluster
with good expansion. The sub-routine uses the cut matching game due to
Khandekar, Rao, and Vazirani\cite{khandekar2009graph} and adds a pruning step
(\cite{SaranurakW19expander}) for the case that the cut-matching step returns
a very unbalanced cut. The Pruning step is the same as
in~\cite{SaranurakW19expander} but here we give a different analysis that shows
a stronger guarantee.

\begin{lemma}[Cut-Matching + Trimming]
\label{thm:cut-matching-trimming}
Given an unweighted graph $G=(V,E)$ with $m$ edges and parameters $\phi,w$ with
$w<{1}/(8\phi)$, a cut-matching+trimming step runs in time $O(m\log
m/\phi)$ and must end in one of the following two cases:
\begin{enumerate}
\item We find a cut $(A,\bar{A})$ of $G$
with $cut_G(A,\bar{A})\le \gkrv\cdot\phi\min\{\vol_G(A),\vol_G(\bar{A})\}$,
and $\vol_G(A)$, $\vol_G(\bar{A})$
are both $\Omega(m/\log^2 m)$, i.e., we find a relatively balanced low
conductance cut.
\item We find a cut $(A,\bar{A})$,
with $cut_G(A,\bar{A})\le \gkrv\cdot\phi\min\{\vol_G(A),\vol_G(\bar{A})\}$, and
$\vol_G(\bar{A})=m/10$. Moreover, we conclude that
$\induced{A}^w$ is a $\phi$-expander. This conclusion may be wrong with probability
$o(m^{-10})$. 
\end{enumerate}
\end{lemma}


\noindent
In the remainder of this section we prove the above theorem.


\def\gorig{\gamma_{\textsc{krv}}}

\bigskip
\noindent
We  use a standard adaptation by
Saranurak and Wang~\cite{SaranurakW19expander} of the {\em cut-matching
  framework}, which was originally proposed by Khandekar, Rao and
Vazirani~\cite{khandekar2009graph}. The following \emph{cut-matching step} was proved
in~\cite{SaranurakW19expander}.
\begin{lemma}[Adapted Statement of Theorem 2.2 in~\cite{SaranurakW19expander}]
\label{thm:cut-matching}
Given an unweighted graph $G=(V,E)$ with $m$ edges and a parameter $\phi>0$, the
\emph{cut-matching step} takes $O(m\log m/\phi)$ time and must end with one of
three cases:
\begin{enumerate}
\item We conclude that $G$ has conductance $\Phi_G\geq8\phi$. This conclusion
is wrong with probability $o(m^{-10})$.
\item We find a cut $(A,\bar{A})$ of $G$ with conductance
$\Phi_G(A)\le\gorig\phi$, and $\vol_G(A),\vol_G(\bar{A})$
are both at least $m/(100\gorig)$, i.e., we find a relatively balanced low
conductance cut.
\item We find a cut $(A,\bar{A})$, such that
$\Phi_G(A)\le\gorig\phi$ and
$\vol_G(\bar{A})\le m/(100\gorig)$. Moreover, we conclude that
$A$ is a near $8\phi$-expander. This conclusion may be wrong with probability
$o(m^{-10})$.

\end{enumerate}
Here, $\gorig = O(\log^2 m)$. 
\end{lemma}

Recall the definition of near expanders from \Cref{defn:nearlyexp}.
We remark that this is the only building block that is randomized in
our algorithms. Once we derandomize it, all our algorithms become
deterministic. In fact, in a recent paper~\cite{ChuzhoyGLNPS19det}, a
deterministic counterpart of the cut-matching step was constructed. We will use
their results and roughly show how to make our algorithms deterministic in
Section~\ref{sec:deterministic-EH}.

  In order to obtain \Cref{thm:cut-matching-trimming}, we proceed as
  follows. We run the cut-matching step from \Cref{thm:cut-matching} on the graph $G$. If we are in Case~1
  of \Cref{thm:cut-matching} we obtain a valid set $A$ for Case~2 in
  \Cref{thm:cut-matching-trimming}, where $\bar{A}$ is the empty set. If
  we are in Case~2 we obtain valids set $A,\bar{A}$ for Case~1 in
  \Cref{thm:cut-matching-trimming}. If we are in Case~3 we perform a
  \emph{trimming operation} on the set $A$ to obtain a set $A'$. 
   We will need to prove that the set $A'$ fulfills all
  properties required for \Cref{thm:cut-matching-trimming}.
  
  The trimming operation (stated below in \Cref{lem:trimming}) is algorithmically exactly the same as in \cite{SaranurakW19expander}. 
  The only difference is in the analysis; we open their black-box and state the guarantee about flow explicitly. 
  Then, we give a new analysis and conclude a stronger statement than the one in \cite{SaranurakW19expander}.
  More precisely, we show that $G[A']^w$ is a $\phi$-expander while they only show that $G[A']^1$ is a $\phi$-expander.


\begin{lemma}[Trimming]\label{lem:trimming}
We can compute a pruned set $P \subset A$ in time $O(\log m|E_G(A,\bar{A})|/\phi^2)$ with the following properties:
\begin{enumerate}
\item $\vol_G(P)\le\frac{4}{\phi}|E_G(A,\bar{A})|$\label{pruning:A}
\item $|E_G(A',\bar{A'})|\le 2|E_G(A,\bar{A})|$\label{pruning:B}
\end{enumerate}
where $A'=A\setminus P$. In addition the following flow problem is feasible
in $\induced{A'}$. 
\begin{itemize}
\item $\Delta(v)=\frac{2}{\phi}|E_G(\{v\},V\setminus A')|$ 
\item $T(v)=\vol_{G}(v)$ 
\item $c(e)=2/\phi$ for every edge in $G[A']$. 
\end{itemize}
\end{lemma}
\begin{proof}
We run the algorithm from 
Lemma~\ref{lem: incrementalFlow} on $\induced{A}[1]$ with $c(e)=2/\phi$  for every edge.
Then we increase $\Delta(v)$ by $2/\phi$ for every edge in $E_G(A,\bar{A})$.

The resulting pruned set $P$ fulfills the properties. Property~\ref{pruning:A}
follows as
\begin{equation*}
\vol_{\induced{A}[1]}(P)=\vol_{G}(P)\le 2\textstyle{\sum_v}\Delta(v)=\tfrac{4}{\phi}|E_G(A,\bar{A})|\enspace.\qedhere
\end{equation*}
\end{proof}


We have to argue that $A'=A\setminus P$ fulfills all requirements of set $A$ in
Case~2 of \Cref{thm:cut-matching-trimming}.

\begin{itemize}
\item $\Phi_G(A')\le\gkrv\phi$.\hfill\\
The conductance of the cut is $|E_G(A',\bar{A'})|/\vol_G(\bar{A'})$. We have
\begin{equation*}
\begin{split}
\vol_G(\bar{A'})
&\ge\vol_G(\bar{A})
\ge\tfrac{1}{\gorig\phi}|E_G(A,\bar{A})|
\ge\tfrac{1}{2\gorig\phi}|E_G(A',\bar{A'})|\enspace.
\end{split}
\end{equation*}
Hence, setting $\gkrv=2\gorig=O(\log^2 n)$ is sufficient.
\item
$\vol_G(\bar{A'})\le m/10$.\\
\begin{equation*}
\begin{split}
\vol_G(\bar{A'})
&\le \vol_G(P) + \vol_G(\bar{A})
\le \tfrac{4}{\phi}|E_G(A,\bar{A})| + \vol_G(\bar{A})
\le \tfrac{4}{\phi}\gorig\phi\vol_G(\bar{A}) + \vol_G(\bar{A})\\
&\le 5\gkrv\vol_G(\bar{A}) = m/10\enspace.
\end{split}
\end{equation*}
\end{itemize}

The final property is given by the following lemma.
  \begin{lemma}
  If $w\le1/(8\phi)$ then $\induced{A'}[w]$ is a $\phi$-expander.
  \end{lemma}
  \begin{proof}
  Directly from the guarantee of the cut-matching step from \Cref{thm:cut-matching}, we get that 
  \begin{equation}
  |E_{G}(A,\bar{A})|\le\gorig\phi\vol_{G}(\bar{A})\le\phi m/16\le\phi\vol_{G}(A)/16. \label{eq:small boundary}
  \end{equation}
Now, consider a subset $S\subseteq{A'}$ such that $\vol_{\induced{A'}[w]}(S) \le
  \vol_{\induced{A'}[w]}(A')/2$. We first show a helpful claim:
     \begin{claim}
     $\vol_G(S) \le \frac{2}{3}\vol_G(A')$.
     \end{claim}
     \begin{proof}
       By the guarantee of the trimming operation from \Cref{lem:trimming} and \Cref{eq:small boundary},
       $\vol_G(P)\le\frac{4}{\phi}\cdot|E_G(A,\bar{A})| \le
       \frac{4}{\phi}\cdot\frac{\phi}{16}\vol_{G}(A)\le\vol_{G}(A)/4$.
       So $\vol_{G}(A')=\vol_G(A)-\vol_G(P)\ge3\vol_{G}(A)/4$.
       Again by \Cref{lem:trimming} and \Cref{eq:small boundary}, we have $|E_G(A',\bar{A'})|\le 2|E_G(A,\bar{A})|\le 2\cdot\frac{\phi}{16}\vol_{G}(A)\le\frac{\phi}{6}\vol_G(A')$.
       We get
       \begin{equation*}\begin{split}
         \vol_G(S)&
         \le\vol_{\induced{A'}[w]}(S)
         \le\vol_{\induced{A'}[w]}(A')/2
           =\tfrac{1}{2}(\vol_{G}(A')+(w-1)|E_{G}(A',\bar{A'})|)\\
        & \le\tfrac{1}{2}(\vol_{G}(A')+\tfrac{1}{8\phi}\cdot\tfrac{\phi}{6}\vol_{G}(A'))
         \le\tfrac{2}{3}\vol_{G}(A')\enspace.\\
        \end{split}
        \end{equation*}
        The equality holds because edges between $A'$ and
        $\bar{A'}=V\setminus A'$ are turned into self-loops of weight $w$ in
        $\induced{A'}[w]$ while having weight $1$ in $G$. Hence, the degree of a vertex in $A'$ incident to
        such an edge increases by $w-1$. The following inequality uses $w-1\le
        1/(8\phi)$ and our previous bound on $|E_{G}(A',\bar{A'})|$.
        \end{proof}

        \def\H{\induced{A'}[w]}

        \noindent        
Recall that $\border_{A'}(S):=|E_G(S,V\setminus A')|$ and
$\cut_{A'}(S):=|E_G(S,A'\setminus S)|$.       
We have to show that $\cut_{A'}(S)\ge \phi\cdot\vol_{\H}(S)$.
From $\vol_G(S)\le\frac{2}{3}\vol_G(A')$ we get
$\vol_G(A'\setminus S)=\vol_G(A')-\vol_G(S)\ge\tfrac{3}{2}\vol_G(S)-\vol_G(S)\ge\vol_G(S)/2$. The fact that $A$ is a near
${8\phi}$-expander in $G$
gives that
\begin{equation}
\label{eqn:helper}
\begin{split}
\cut_{A'}(S)+\border_{A'}(S)
&\ge 8\phi\cdot\min\{\vol_G(S),\vol_G(A\setminus S)\}\\
&\ge 8\phi\cdot\min\{\vol_G(S),\vol_G(A'\setminus S)\}\\
&\ge 4\phi\cdot\vol_G(S)\enspace.
\end{split}
\end{equation}
By the feasibility of the flow problem for $\induced{A'}$ we obtain
\begin{equation*}
\begin{split}
\tfrac{2}{\phi}\cdot\border_{A'}(S)\le\Delta(S)
&\le T(S)+\tfrac{2}{\phi}\cut_{A'}(S)\\
&= \vol_G(S)+\tfrac{2}{\phi}\cut_{A'}(S)\\
&\le \tfrac{1}{4\phi}\border_{A'}(S)+(\tfrac{1}{4\phi}+\tfrac{2}{\phi})\cut_{A'}(S)\enspace,
\end{split}
\end{equation*}
which yields $\border_{A'}(S)\le 9/7\cdot\cut_{A'}(S)\le 2\cut_{A'}(S)$.
Here, the first inequality is due to the fact that the flow problem
injects $2/\phi$ units of flow for every border edge.
The second inequality follows because
the total flow that can be absorbed at the vertices of $S$ is at most $T(S)$
and the flow that can be send to $A'\setminus S$ is at most $\tfrac{2}{\phi}\cut_{A'}(S)$ as each edge has capacity $2/\phi$.
The final step uses Equation~\ref{eqn:helper}.
        
Finally, we obtain
\begin{equation*}
\begin{split}
\tfrac{1}{\phi}\cut_{A'}(S)
&\ge\tfrac{1}{4\phi}(\cut_{A'}(S)+2\cut_{A'}(S))+\tfrac{1}{4\phi}\cut_{A'}(S)\\
&\ge\tfrac{1}{4\phi}(\cut_{A'}(S)+\border_{A'}(S))+\tfrac{1}{8\phi}\border_{A'}(S)\\
&\ge\vol_G(S)+w\border_{A'}(S)
\ge\vol_{\H}(S)\enspace,
\end{split}
\end{equation*}
as desired.
\end{proof}

\paragraph{Running time.} The running time of the cut-matching step from \Cref{thm:cut-matching} is $O(m\log m/\phi)$. The
running time of the trimming step from \Cref{lem:trimming} is $O(\log m |E_G(A,\bar{A})|/\phi^2) = O(m\log m/\phi)$ by \Cref{eq:small boundary}. Hence, the total running time is $O(m\log m/\phi)$.

\subsection{The Algorithm}
The algorithm maintains an expansion parameter $\varphi$ and a partitioning
$\U$ that initially just contains the set $U$ (i.e., the cluster $U\subseteq V$
on which we startet the algorithm) as an active
cluster. Recall that $\vol_G(U)=m$.
Then the algorithm proceeds in rounds, where a round is
an iteration of the outer while loop. During a round the algorithm tries to
certify for all active clusters $U_i$ in $\U$ that
$\induced{U_i}^{\alpha/\varphi}$ is $\varphi$-expanding. For this it uses the
cut-matching+trimming algorithm from \Cref{thm:cut-matching-trimming} with
parameter $\varphi$ on the graph
$\induced{U_i}^{\alpha/\varphi}$. From \Cref{thm:cut-matching-trimming}, there
are two possible outcomes:

\smallskip
\noindent
\textbf{Case 1.} The framework finds a sparse fairly balanced cut
$(A, \bar{A})$. Then the algorithm just replaces $U_i$ by $A$ and $\bar{A}$
in $\U$.

\begin{algorithm}[t]%
\setstretch{1.1}%
\SetKwInOut{Input}{input}\SetKwInOut{Output}{output}%
\SetKwData{false}{false}\SetKwFunction{expansion}{expands}%
\SetKwData{true}{true}%

\Input{graph $G=(V,E)$, cluster $U\subseteq V$, parameters $\alpha,\phi$}%
\Output{partition $\U=(U_1,\dots,U_k)$ of $U$, expansion bounds $\phi_1,\dots,\phi_k$}%
\BlankLine%

define $\U$ to contain only $U$ as an active cluster\;
\While{$\exists$ active sub-cluster in $\U$}{%
    $\varphi\leftarrow\max\Big\{\frac{1}{8\gkrv\log_2^2m}\cdot{\sum_{\text{act.\ $i$}}\out(U_i)}/{\sum_{\text{act.\ $i$}}\vol(U_i)}~,~\phi\Big\}$\;
    \lFor{ $U_i\in \U$}{ $\expansion(U_i,\varphi)\gets \false$}
    \While{$\exists$ active cluster $U_i$ with $\expansion(U_i,\varphi)=\false$}{%
        apply cut-matching + trimming from \Cref{thm:cut-matching-trimming} to $\induced{U_i}[\alpha/\varphi]$\;
        \textbf{case~1}: replace $U_i$ by active sets $A$ and $\bar{A}$ in $\U$\;
        \textbf{case~2}: $\expansion(A,\varphi)\gets \true$\tcp*{$A$ is $\varphi$-expanding, w.h.p.}
        \hphantom{\textbf{case~2}:} replace $U_i$ by active sets $A$ and $\bar{A}$ in $\U$\;
    }
    \For{every active set $U_i\in \U$}{%
      \If(\tcp*[f]{check Property~\ref{ed:localboundary}}){$\out(U_i)\le 80\gkrv\log^4m\cdot\varphi\vol(U_i)$}{%
         $\phi_i\leftarrow\varphi$\tcp*{set expansion bound for $U_i$}
        deactivate $U_i$\tcp*{$U_i$ fulfills Property~\ref{ed:expanding} (w.h.p.)~and \ref{ed:localboundary}}    
      }
    }
}
\caption{An algorithm to compute an $(\alpha,\phi)$-expander decomposition.}%
\label{algo}
\end{algorithm}

\smallskip
\noindent
\textbf{Case 2.} The framework finds an unbalanced cut $(A, \bar{A})$ and concludes that the larger
part $A$ forms a $\varphi$-expander. Then the algorithm
replaces $U_i$ by $A$ and $\bar{A}$ and remembers the conclusion that $A$ is
expanding, i.e., the algorithm will not work on $A$ again during a round.

\medskip
\noindent
After the algorithm has determined that w.h.p.\ all active clusters in $\U$ are
$\varphi$-expanding it checks for every cluster whether
Property~\ref{ed:localboundary} from \Cref{def:decomp} of the boundary-linked expander-decomposition holds. If this is
the case for a cluster $U_i$ the algorithm sets $\phi_i$ to the current value
of $\varphi$ and deactivates the cluster.


The algorithm then proceeds to the next round (possibly increasing $\varphi$)
and continues until no active clusters are left. Algorithm~\ref{algo} gives an
overview of the algorithm.

\subsection{The Analysis}
In the following we assume that all conclusions by the algorithm that are
correct with high probability are indeed correct.

It is clear that when the algorithm terminates all clusters fulfill
Property~\ref{ed:expanding} and Property~\ref{ed:localboundary}, i.e., we only
have to prove that the partition $\U$ fulfills
Property~\ref{ed:globalboundary} and that the algorithm indeed terminates.

\def\cluster{\operatorname{cluster}}

  Let for a subset $X$ $\ivol_G(X):=\sum_{x\in X}\out{\{x\}}$ denote the
  \emph{internal volume} of the set, i.e., the part of the volume that is due
  to the edges for which both endpoints are in $X$. In the following the
  notation $\cut_U(A)$, $\border_U(A)$, and $\out(A)$ are always w.r.t.\ the
  graph $G$. Further we use 
  $Z:=\log_2(\vol(U))$ as a shorthand notation. For $x\in V$
  we use $\cluster(x)$ to denote the cluster from the partion $\U$ that $x$ is
  contained in. If $x\notin U$ then this evaluates to $\texttt{undefined}$.


  In order to derive a bound on $\sum_{i=1}^k\out(U_i)$ we proceed as
  follows. We distribute an initial charge to the edges incident to vertices in
  $U$. Whenever we \emph{cut} edges, i.e., we partition a subset $U_i$ into
  $A$ and $\bar{A}$ we redistribute charge to the edges in the cut. In the
  end we compare the charge on edges leaving sub-clusters to the initial charge
  and thereby obain a bound on $\sum_i\out(U_i)$. In addition we will
  redistribute charge whenever we adjust the value of $\varphi$ in the
  beginning of a round. However, importantly we will never increase the total
  charge, hence, in the end we can derive a bound on the number of cut-edges by
  comparing the charge on a cut-edge to the total \emph{initial charge}.

  We call one iteration of the outer while-loop a round of the algorithm. 
  Let $R$ denote an upper bound on the number of rounds. Later we will show
  that $R\le\log_2m$. For any round $r \le R$, we maintain the following
  invariant concerning the distribution of charge on the edges that have at
  least one end-point from the set $U$:

  \begin{quote}
  \smallskip\noindent
  \emph{border edges}\hfill\\
  An edge $(u,v)$ for which not both endpoints are in the same sub-cluster of $\U$ is
  assigned a charge of $\fB(r)(Z+\log_2(\ivol(\cluster(x))))$ for each end-point
  $x\in\{u,v\}\cap U$. We call a charge \emph{active} if it comes from an
  endpoint within an active cluster. This means that an edge could be assigned
  active as well as inactive charge.

  \smallskip\noindent
  \emph{internal edges}\hfill\\
  An edge $(u,v)$ for which both endpoints are in the same sub-cluster $U_i$ of $\U$ is
  assigned an active charge of $\fI(r)\log_2(\ivol(U_i))$ if this cluster is
  active. Otherwise, it is assigned a charge of $0$.
  \end{quote}
  

  \noindent
  We refer to the charge on border edges as \emph{border charge} and to the
  charge on internal edges as \emph{internal charge}.
  The factors $\fI(r)$ and $\fB(r)$ in the above definition depend on the round and are chosen as follows:
  \begin{align*}
    \fB(r) &= 2R-r   \\
    \fI(r) &= 4\gkrv Z\fB(r)\varphi(r).
  \end{align*}
  Recall the parameter $\gkrv=O(\log^2 m)$ from the cut-matching + trimming algorithm in \Cref{thm:cut-matching-trimming}. 
  When we call the algorithm from \Cref{thm:cut-matching-trimming} with conductance parameter $\varphi$ then the non-empty cuts returned in
  have conductance at most $\gkrv\varphi$. Below, let $\varphi(r)$ denote the value of $\varphi$ during round $r$.
  For technical reasons we also introduce a round $r=0$, which is the beginning
  of the algorithm. We set $\varphi(0)=\phi$. This gives that the total
  \emph{initial charge} is
  \begin{equation*}
  \begin{split}
  \textit{initial-charge}
  &=\fB(0)\cdot b\cdot(Z+\log_2(\ivol(U)))+4\gkrv Z\fB(0)\varphi(0)\cdot\ivol(U)\log_2(\ivol(U))\\
  &\le 4RZ\cdot b+8\gkrv RZ^2\cdot \phi m.
  \end{split}
  \end{equation*}

  \noindent
  The following claim gives Property~\ref{ed:globalboundary} provided that we
  can establish the above charge distribution without generating new charge
  during the algorithm.

\begin{claim}
Suppose that no charge is generated during the algorithm. Then
at the end of the algorithm $\sum_i\out(U_i)\le 4b+O(\log^3m)\cdot\phi
m$, i.e., Property~\ref{ed:globalboundary} holds.
\end{claim}
\begin{proof}
Observe that in the end every inter-cluster edge will have a charge of
at least $\fB(R)Z\ge RZ$. Therefore
$\sum\nolimits_i \out(U_i)
   \le\tfrac{1}{RZ}\textit{initial-charge}
   \le4b+8\gkrv Z\cdot\phi m
   =4b+O(\log^3m)\cdot\phi m\enspace.\hfill
$
\end{proof}
Observe that the number of rounds performed by our algorithm is not important
for the above claim. This is only important for the running time analysis.

\paragraph{No charge increase during a round.}
The main task of the analysis is to establish the charging scheme and to show that we do not
generate charge throughout the algorithm. We first show that we do not generate
charge \emph{during} a round.
  
  Suppose \Cref{thm:cut-matching-trimming} finds a cut of conductance at most
  $\gkrv\cdot\varphi(r)$ within the graph $H:=\induced{U_i}^{\alpha/\varphi(r)}$.
 This means we
  have a set $S\subseteq U_i$ with
  \begin{equation*}
  \cut_{U_i}(S) < \gkrv\cdot\varphi(r)\cdot\min\{\vol_H(S),\vol_H(U_i\setminus S)\}\enspace.
  \end{equation*}
  W.l.o.g.\ assume that $\ivol(S)\le\ivol(U_i\setminus S)$. Then
  
  \begin{equation*}
  \cut_{U_i}(S) <  \gkrv\cdot\varphi(r)\vol_H(S)
               = \gkrv\cdot\varphi(r)\cdot(\ivol(S)+\tfrac{\alpha}{\varphi(r)}\border_{U_i}(S))\enspace.
  \end{equation*}
  By performing the cut, every edge that contributes to $\border_{U_i}(S)$
  reduces its required charge by at least $\fB(r)$ because one of its endpoints
  reduces the internal volume of its cluster by a factor of $2$. A similar
  argument holds for the edges with both endpoints in $S$, which reduce their
  required charge by at least $\fI(r)$. This means we obtain a charge of at
  least
  \begin{equation}\label{availablecharge}
  \fB(r)\cdot\border_{U_i}(S)+\fI(r)\cdot\ivol(S)
  \end{equation}
  that we can redistribute to the edges in the cut so that these fulfill their
  increased charge requirement. The new charge for the cut edges (i.e., edges in
  $\Gamma_G(S,U_i\setminus S)$) is at most
  \begin{equation*}
  \fB(r)(Z+\log_2(\ivol(S)))+\fB(r)(Z+\log_2(\ivol(U_i\setminus S)))\le 4\fB(r) Z,
  \end{equation*}
  where $Z = \log_2(\vol(U))$. 
  This means the new required charge is
  \begin{equation*}\begin{split}
  4\fB(r) Z\cdot\cut_{U_i}(S)
      &< 4\gkrv\fB(r) Z\cdot(\varphi(r)\ivol(S)+\alpha\border_{U_i}(S))\\
      &\le 4\gkrv\alpha Z\fB(r)\border_{U_i}(S) + 4\gkrv\varphi(r) Z\fB(r)\ivol(S)\\
      &\stackrel{!}{\le}
\fB(r)\cdot\border_{U_i}(S)+\fI(r)\cdot\ivol(S)\enspace,
  \end{split}
  \end{equation*}
where we want to ensure the last inequality so that the new charge on cut-edges
is at most the charge that we have for redistribution according to Equation~\ref{availablecharge}.
We ensure the last inequality by requiring that
\begin{equation*}
4\gkrv\alpha Z\le 1
\end{equation*}
as a precondition of the theorem and by setting
\begin{equation*}
\fI(r):=4\gkrv\fB(r) Z\varphi(r).
\end{equation*}
This shows that we can redistribute enough charge to border edges and the total
charge does not increase.

\def\iact{i\in I}

\paragraph{No charge increase between rounds.} Let $I$ denote the index set of
active clusters at the start of round $r$. At the
beginning of a round all border edges decrease their charge as $\fB(r)$
decreases. If we choose $\varphi(r)=\phi$ then the charge on internal edges
does not increase because in the previous round we had $\varphi(r-1)\ge\phi$
and the charge on an internal edge is increasing with $\varphi$. Hence, we only
need to consider the case if $\varphi(r)$ is chosen larger than $\phi$ and
hence
$$\varphi(r)=\frac{1}{8\gkrv
  RZ}\cdot{\sum_{\iact}\out(U_i)}/{\sum_{\iact}\vol(U_i)}.$$

For this case we show that the decrease in charge on active border edges is
more than the increase in charge on internal edges. This is sufficient as only
active internal edges increase their charge. Every edge in the boundary
of an active cluster $U_i$ decreases its charge by at least
$(\fB(r-1)-\fB(r))Z\ge Z$. This means the border charge decreases by at least
$Z\sum_{\iact}\out(U_i)$.

What is the total internal charge?
Every edge inside an active cluster $U_i$ is assigned a charge of
$\fI(r)\log_2(\ivol(U_i))$. Recall that we set internal charge of inactive cluster to be zero. Therefore, the total \emph{internal charge} is
\begin{equation}
\label{eqn:internal-charge}
\begin{split}
\textit{internal-charge}&=\fI(r)\sum_{\iact}\ivol(U_i)\log_2(\ivol(U_i))\\
                        &\le4\gkrv\fB(r) Z^2\cdot\varphi(r)\cdot\sum_{\iact}\ivol(U_i)\\
                        &\le8\gkrv R Z^2\cdot\varphi(r)\cdot\sum_{\iact}\vol(U_i)\\ 
                        &=Z\sum_{\iact}\out(U_i)\\
\end{split}
\end{equation}
where the last step follows by the choice of $\varphi(r)$.
This means the reduction of charge on border edges is even lower bounded by the
\emph{total} internal charge (not just the increase of internal charge). Hence,
the overall charge is not increasing.


\paragraph{Bound on the number of rounds.}
In order to keep the total number of rounds small we guarantee that the active
volume, i.e., $\sum_{\iact}\vol(U_i)$ decreases by a constant factor between
two rounds. In order to guarantee this we first show that the choice of
$\varphi(r)$ fulfills the following inequality.
\begin{equation}\label{eqn:totalcharge}
\textit{active-charge}\le 40\gkrv R^2Z^2\cdot\varphi(r)\cdot\sum_{\iact}\vol(U_i)\enspace.
\end{equation}
The active charge consists of the total internal charge and the active border
charge. Equation~\ref{eqn:internal-charge} gives that
\begin{equation*}
\textit{internal-charge} \le8\gkrv R Z^2\cdot\varphi(r)\cdot\sum_{\iact}\vol(U_i)\enspace.
\end{equation*}

\noindent
The active border charge is
\begin{equation}
\label{eqn:border-charge}
\begin{split}
\textit{active-border-charge} &=   \fB(r)\sum_{\iact}\out(U_i)(Z+\log_2(\ivol(U_i)))\\
                       &\le 2\fB(r)Z\sum_{\iact}\out(U_i)\le 4RZ\sum_{\iact}\out(U_i)\\
                       &\le 32\gkrv R^2Z^2\varphi(r)\cdot \sum_{\iact}\vol(U_i)
\end{split}
\end{equation}
where the last inequality follows as the algorithm chooses
$\varphi\ge
\frac{1}{8\gkrv RZ}\cdot{\sum_{\iact}\out(U_i)}/{\sum_{\iact}\vol(U_i)}$.
Combining both inequalities gives Inequality~\ref{eqn:totalcharge}.




\begin{claim}
The term $\sum_{\iact}\vol(U_i)$ decreases by a factor of at least
$1/2$ between two rounds of the algorithm. This gives that $R\le\log_2m$.
\end{claim}
\begin{proof}
The active charge on a boundary edge is at least $RZ$. Since we do not generate
charge during a round and we do not redistribute inactive charge we obtain that at the end of the round
\begin{equation*}
\begin{split}
\sum_{\iact'}\out(U_i)
&\le \frac{1}{RZ}\textit{active-charge}'
\le \frac{1}{RZ}\textit{active-charge}
\le 40\gkrv RZ\varphi(r)\sum_{\iact}\vol(U_i)\\
&= 40\gkrv RZ\varphi(r)\sum_{\iact}\vol(U_i)\enspace,
\end{split}
\end{equation*}
where $I'$ denotes the set of active cluster after the first inner while-loop
(i.e., before we start deactivating clusters). We use $\textit{active-charge}'$
to denote the active charge at this time. The last equation holds because the
active volume does not change during the first while-loop.

Now a simple averaging argument gives that the volume in ``bad'' clusters (i.e.~active
clusters that have $80\gkrv\varphi(r) RZ\vol(U_i)<\out(U_i)$)
is at most half of the active volume. These are the clusters that make it to
the next round. Hence, the claim follows.
\end{proof}

\paragraph{Running time.}
We first derive a bound on the running time of a single round. When we apply
the cut-matching+trimming algorithm from \Cref{thm:cut-matching-trimming} to a
subgraph $G[U_i]^w$ we can charge the running time to the edges in
$G[U_i]^w$. 
We charge $O(\frac{1}{\varphi}\log m)= O(\frac{1}{\phi}\log m)$ to every edge.
Whenever we charge an edge $e$ at least one cluster $U_i$ that contains an
endpoint of $e$ changed. We either set
$\expansion(U_i,\varphi) \gets \operatorname{true}$ for this cluster (and,
hence, stop processing this cluster for this round) or $\vol_G(U_i)$ decreases
by a $(1-1/\log^2m)$ factor. This implies that an edge can be charged at most
$O(\log^3 m)$ times during a round.

It remains to derive a bound on the total number of edges in active clusters.
Note that we cannot simply use $m$ as an upper bound because the algorithm acts
on sub-cluster $G[U]^w$, i.e., graphs where $w$ self-loops are added for each
border-edge. 

The total number of border edges during a round is at most
\begin{equation*}
\textit{initial-charge}/RZ\le  O(b+\gkrv Z\phi m)
\end{equation*}
because every border-edge has charge at least $RZ$.

For each such border edge we add $w=\lceil\alpha/\varphi\rceil\le 1/(\gorig
Z\phi)$ self-loops (where we use $\phi\le\alpha\le 1/(\gkrv Z)$).
Therefore there are at most $\ot(b/\phi+m)$ edges in all graphs
$\induced{U_i}[\alpha/\phi]$. Hence, the cost of a single round is only
$O(\log^4m(b/\phi^2+m/\phi))$. Since, the number of rounds is logarithmic the running
time follows.

\section{Tree Flow Sparsifier From Expander Hierarchy}
\label{sec:flow sparsifier}

In this section, we show that an expander hierarchy of a graph $G$ is itself a
tree flow sparsifier of $G$. Usually the concept of a flow sparsifier is
defined for weighted graph. In order to simplify the notation and keep it
consisten throughout the paper our definitions and proofs only consider
unweighted (multi-)graphs. However the extension to weighted graphs is
straightforward.


\paragraph{Multicommodity Flow.}
Given an unweighted (multi-)graph $G=(V,E)$, let $\cal{P}$ be the set of all paths in
$G$. A \emph{multicommodity flow} (that is also referred to as a \emph{flow})
$F$ is an assignment of non-negative values $F_P$ to all paths $P\in \cal{P}$.
Each path in $\cal{P}$ has one of its endpoints being the \emph{originating
  vertex}, and the other endpoints being the \emph{terminating vertex}. When we
assign the value $F_P$ to the path $P$, we are sending $F_P$ unit of flow from
its originating vertex to its terminating vertex. For two vertices $v,v'\in V$,
we denote by $\pset_{v,v'}\subseteq\mathcal{P}$ the set of paths that
originate at $v$ and terminate at $v'$, and we say that the amount of flow
that $F$ sends from $v$ to $v'$ is $\sum_{P\in \pset_{v,v'}}F_P$. The
\emph{congestion} of the flow $F$ is defined to be
$\text{cong}_G(F)=\max_{e\in E}\{F(e)\}$, where $F(e)$ is the
total amount of flow sent along the edge $e$. Given a flow $F$ on $G$ and two
subsets $A,B\subseteq V(G)$ of vertices, we define $F(A,B)$ to be the total
amount of flow of $F$ that is sent along an edge $e\in E(A,B)$ from its
endpoint in $A$ to its endpoint in $B$. Note that, however, for two vertices
$v,v'\in V$ such that $(v,v')\in E$, $F(\{v\},\{v'\})$ can be smaller than the
amount of flow that $F$ sends from $v$ to $v'$.


\paragraph{Cut and Flow Sparsifiers.}
Given a (multi-)graph $G=(V,E)$ and a subset $S\subseteq V$ of
vertices, let $H=(V',E')$ be a (multi-)graph with
$S\subseteq V(H)$. We say that the graph $H$ is a \emph{cut sparsifier of
  quality $q$} for $G$ with respect to $S$, if for each partition $(A,B)$ of
$S$ such that both $A$ and $B$ are not empty, we have
$\mincut_{H}(A,B)\le\mincut_{G}(A,B)\le q\cdot\mincut_{H}(A,B)$, where
$\mincut_{H}(A,B)$ ($\mincut_{G}(A,B)$, resp.) is the capacity of a minimum cut
that separates the subsets $A$ and $B$ of vertices in $H$ ($G$, resp.).
If $H$ is a tree, then $H$ is called a \emph{tree cut sparsifier}.

Given a (multi-)graph $G=(V,E)$ and a subset $S\subseteq V$ of
vertices, a set $D$ of \emph{demands} on $S$ is a function
$D:S\times S\to \mathbb{R}_{\ge 0}$, that specifies, for each pair $u,v\in V$
of vertices, a demand $D(u,v)$. We say that the set $D$ of demands is
\emph{$\gamma$-restricted}, iff for each vertex $v\in S$,
$\sum_{u\in S}D(v,u)\le \gamma\out(v)$ and
$\sum_{u\in S}D(u,v)\le \gamma\out(v)$, i.e., the demand
entering or leaving $v$ is at most $\gamma$ times the number of edges leaving
$v$. We call it $\gamma$-boundary
restricted (w.r.t.\ $S$) if
$\sum_{u\in S}D(v,u)\le \gamma\border_S(v)$ and
$\sum_{u\in S}D(u,v)\le \gamma\border_S(v)$.
Given a subset $S\subseteq V$ and a set $D$ of demands on $S$, a
\emph{routing} of $D$ in $G$ is a flow $F$ on $G$, where for each pair
$u,v\in S$, the amount of flow that $F$ sends from $u$ to $v$ is $D(u,v)$. We
define the congestion $\eta(G,D)$ of a set $D$ of demands in $G$ to be the
minimum congestion of a flow $F$ that is a routing of $D$ in $G$. We say that a
set $D$ of demands is \emph{routable} on $G$ if $\eta(G,D)\le 1$.

We say that a graph $H$ is a \emph{flow sparsifier of quality $q$} for $G$ with
respect to $S$, if $S\subseteq V(H)$, and for any set $D$ of demands on $S$,
$\eta(H,D)\le \eta(G,D)\le q\cdot\eta(H,D)$. A flow sparsifier $H$ of $G$
w.r.t.\ subset $V(G)$ is just called a flow sparsifier for $G$. If $H$ is a
tree we call $H$ a \emph{tree flow sparsifier}.

We will use the following lemma, which is a direct consequence of approximate
max-flow mincut ratios for multicommodity flows. The proof appears in
Appendix~\ref{apd:restricted}.

\begin{lemma}
	\label{lem:restricted_demand}
Given a graph $G$ together with a subset $S\subseteq
V$ that is $(\alpha,\phi)$-linked in $G$. Then the following
two statements hold.
\begin{itemize}
\item 
We can route a $\gamma$-restricted set of demands $D$ on $S$ with congestion
$O(\frac{\gamma}{\phi}\log m)$ inside $G[S]$. 
\item 
We can route a $\gamma$-boundary restricted set of demands $D$ on $S$ with congestion
$O(\frac{\gamma}{\alpha}\log m)$ inside $G[S]$. 
\end{itemize}
\end{lemma}



The main theorem of this section is to show the following theorem that 
a $(\alpha,\phi)$-EH of a graph is automatically a tree flow sparsifier. It is
well known that flow sparsifiers are a stronger notion than cut-sparsifiers and
that the quality of the two version may differ by a logarithmic factor.
\begin{theorem}
\label{thm:EH_gives_R_tree}
The $(\alpha,\phi)$-EH of an undirected, connected graph $G$ with $m$ edges
forms a tree flow sparsifier for $G$ with quality
$O(s\log m)^t\cdot O(\max\{\frac{1}{\alpha},\frac{1}{\phi}\}/\alpha^{t-1})$,
where $t$ denotes the depth and $s$ the slack of the hierarchy.
\end{theorem}

If we set $\phi=2^{-\sqrt{\log m}}$ and so $t=O(\sqrt{\log m})$, then together with
\Cref{cor:hierarchy}, we immediately obtain the following corollary:

\begin{corollary}
There is an algorithm, that, given any unweighted $m$-edge graph $G$, with high
probability, computes a tree flow sparsifier for graph $G$ of quality
$O(\log n)^{O(\sqrt{\log n})}$ in time $m^{1+o(1)}$.
\end{corollary}

Observe that stronger results than the above theorem are known
because~\cite{RackeST14} gives a polylogarithmic guarantee on the quality with a
running time of $O(m\polylog m)$. However, our approach here is simpler and
because we are able to efficiently maintain an expander hierarchy we also
obtain a result for dynamic graphs.
The main tool for proving \Cref{thm:EH_gives_R_tree} is the following lemma that shows how to
construct a flow sparsifier for a graph $G$ if one is given a flow sparsifier
for some contraction $G_\U$ of $G$.

\def\GU{G_{\mathcal{U}}}
\def\DU{D_{\mathcal{U}}}
\def\HU{H_{\mathcal{U}}}
\begin{lemma}\label{lem:helper}
Let $G$ be an unweighted graph and
$\mathcal{U}=(U_1,\ldots,U_r)$ be an $(\alpha,\phi)$-ED of $G$ with slack $s$. Given a flow sparsifier $\HU$ for the contracted graph
$\GU$ we can construct a flow sparsifier $H$ for $G$ as follows:
\begin{enumerate}[itemsep=0pt]
\item Add vertices of $V(G)$ to $V(H_{\mathcal{U}})$.
\item Connect a newly added vertex $v\in U_i$ to the vertex $u_i\in
V(H_{\mathcal{U}})$ with $\out_G(v)$ parallel edges.%
\label{steptwo}%
\end{enumerate}
The quality $q_H$ of the resulting flow sparsifier $H$ is
$O((\frac{q}{\alpha}+\frac{1}{\phi})s\log m)$, where $q$ denotes the quality of
the flow sparsifier $\HU$.
\end{lemma}
\begin{proof}
For a given set $D$ of demands on $V(G)$ we use $\DU$ to denote 
the \emph{projection of $D$} to $V(\GU)$, i.e., for two nodes
$u_i,u_j\in V(\GU)$ we define $\DU(u_i,u_j):=\sum_{x\in U_i,y\in U_j}D(x,y)$.

\def\Dsource{D_s}
\def\Dtarget{D_t}
We first show that for all demands $D$ we have $\eta(H,D)\le\eta(G,D)$. Fix
some demand $D$ and assume w.l.o.g.\ that the congestion $\eta(G,D)$ for routing
$D$ in $G$ is $1$ (this can be obtained by scaling).

For routing between two vertices $x\in U_i$ and
$v\in U_j$ from $H$ we split their demand into three parts: $x\rightarrow u_i$,
$u_i\rightarrow u_j$, and $u_j\rightarrow y$. Doing this for all demand-pairs
gives us three sets of demands: 
the \emph{source demand} $D_s$ defined by 
$D_s(x,u_i):=\sum_{y\in V(G)}D(x,y)$ (where $x\in U_i$),
the \emph{projected demand $\DU$}
and the \emph{target demand}
$D_t(u_j,y):=\sum_{x\in V(G)}D(x,y)$ (where $y\in U_j$). 
We route theses demands in $H$ as follows.

\begin{itemize}
\item
The source and target demand can be routed in $H$ via the edges
that were added in Step~\ref{steptwo}. The total traffic that
is generated on the edge $(x,u_i)$ is the total demand that leaves
or enters vertex $x$ in $D$. However, the latter is at most $\out_G(x)$
as otherwise the demand could not be routed in $G$ with congestion $1$. Hence,
the congestion caused by this step in $H$ is at most $1$.

\item
The projection demand $\DU$ can be routed only along edges belonging to $\HU$.
Clearly, this demand can be routed in $\GU$ with congestion at most $1$, and,
hence, it can also be routed in $\HU$ with congestion at most $1$ as $\HU$ is
a flow sparsifier for $\GU$.
\end{itemize}
Observe, that the edges used for routing in the above two steps are disjoint.
Hence, we can concurrently route demands $D_s, D_t$, and $\DU$ with congestion 
$1$, and, hence, we can also route $D$ with this congestion.

Now, we show that $\eta(G,D)\le q_H\cdot\eta(H,D)$. Fix some demand $D$ and
assume w.l.o.g.\ that $\eta(H,D)=1$. From this it follows that we can route the
projected demand $\DU$ in $\HU$ with congestion $1$. Since, $\HU$ is a flow
sparsifier for $G_\U$ (with quality $q$) this implies that we can also route
$\DU$ in $\GU$ with congestion $q$. 

In the following we describe how to extend a routing $F$ for the demand $\DU$
in $\GU$ to a routing of $D$ in $G$.
%
%
%
In a first step we map the non-empty flow paths of $F$ to $G$ (note that the edges of
the contracted multigraph $\GU$ also exist in $G$; we simply map the flow from
edges in $\GU$ to the corresponding edge in $G$).
Consider such a flow path
$u_i=u_{s_1},u_{s_2},\dots,u_{s_k}=u_j$. In $G$ its edges connect subsets
$U_{s_1},U_{s_2},\dots,U_{s_k}$ but they do not form paths. For example we
could have an edge $e=(x_1,x_2)$ followed by an edge $e'=(x_2',x_3)$, with
$x_1\in U_{s_1}$, $x_2,x_2'\in U_{s_2}$, and $x_3\in U_{s_3}$. In order to
obtain \emph{paths} in $G$ we have to connect $x_2$ to $x_2'$. Performing this
reconnection step for all routing paths from $F$ induces a flow problem for
every cluster $U_i$. The total demand (incoming and outgoing) for a vertex
$x\in U_i$ in this flow problem is the total value of all flow-paths that $x$
participates in. But this can be at most $q\border_{U_i}(v)$ as each of
these flow paths uses an edge incident to $x$ that leaves $U_i$ and the
congestion is at most $q$. 
Since $U_i$ is $(\alpha/s,\phi/s)$-linked and this set of demands is $q$-boundary-restricted, by \Cref{lem:restricted_demand},
we can route such a set of demands in $G[U_i]$ with congestion $O(\frac{q}{\alpha}s\log m)$. 
As all clusters $U_i$ are vertex-disjoint, performing all reconnections results in 
congestion $O(\frac{q}{\alpha}s\log m)$.

We also map the empty flow paths of $F$ to empty paths in $G$ as follows.
A $u_i-u_i$ path in $F$ is mapped to $x-x$ paths in $G$ with $x\in U_i$ such
that the total flow that starts at a vertex $x$ (including empty paths) is
exactly $\sum_{y\in V(G)}D(x,y)$. This can be obtained because
$\sum_{y\in V(G)}\DU(u_i,y)=\sum_{x\in U_i}\sum_{y\in V(G)}D(x,y)$ and because $F$ routes demands
$\DU$.

Let $D'$ denote the set of demands routed by the flow system that we have constructed so far.
Observe that $D'$ has the same projection as our demand $D$, i.e., $\DU'=\DU$. The 
following claim shows that one can extend a routing for $D'$ to a routing for
$D$ with a small increase in congestion.

\begin{claim}
Suppose we are given $\gamma$-restricted demands $D$ and $D'$ that fulfill
$\DU'=\DU$ and assume that $D'$ can be routed with congestion $C'$. Then we can route
$D$ with congestion at most
$O(C'+s\gamma/\phi\cdot\log m)$.
\end{claim}
\begin{proof}
%
%
%
%
%

Since, the projection of demands $D$ and $D'$ are equal we know that
$\sum_{(x,y)\in U_i\times U_j}D(x,y)=\sum_{(x,y)\in U_i\times U_j}D'(x,y)$.
We successively route $D$ using the flow-paths of the routing for $D'$. For
every pair $(x,y)$ that we want to connect in $D$ we find portals
$(x',y')\in U_i\times U_j$ that are connected in $D'$. Then we add flow paths
from $x$ to $x'$ and from $y'$ to $y$. Formally, we use the following algorithm
to compute a demand $R''$ such that $D'$ together with $R''$ can route $D$. 

{%
\intextsep0pt
\begin{algorithm}
$R\leftarrow D$\;
$R'\leftarrow D'$\;
\While{$\exists x\in U_i, y\in U_j$ with $R(x ,y)>0$}{
	choose $x'\in U_i, y'\in U_j$ with $R'(x',y')>0$              \tcp*{choose pair of portals}
	decrease $R(x,y)$ and $R'(x',y')$ by $\epsilon$            \tcp*{route flow $\epsilon$ via pair $(x,y)$}
	increase $R''(x,x')$ and $R''(y',y)$ by $\epsilon$\;
        \tcp{R'' stores demand for connecting to portals}
}
\end{algorithm}
}

\noindent
The demands in $R''$ are just between vertex pairs inside clusters $U_i$,
$i\in\{1,\dots, s\}$. The total demand that can enter or leave a vertex $v$ (in
$R''$) is at most $2\gamma\out_G(v)$, because each such demand either occurs
in $R''$ because $v$ is used as an original source/target for demand in $D$ or
as a portal (i.e., as a source/target of a demand in $D'$). Since, both $D$ and
$D'$ are $\gamma$-restricted we get that $R''$ is $2\gamma$-restricted.
Therefore, we can route $R''$ with congestion $O(s\gamma/\phi\cdot\log m)$ by
\Cref{lem:restricted_demand} using the fact that each $U_i$ is $(\alpha/s,\phi/s)$-linked.
\end{proof}


Using the fact that demands $D$ and $D'$ are $O(1)$-restricted we can route $D$
with congestion at most $q_H:=O((\frac{sq}{\alpha}+\frac{s}{\phi})\log m)$. This
gives the bound on the quality of the sparsifier $H$.
\end{proof}

\paragraph{Proof of Theorem~\ref{thm:EH_gives_R_tree}:}
Let $(G^{0},G^{1},\dots,G^{t})$ be some $(\alpha,\phi)$-expander decomposition
sequence with slack $s$ and let $T$ denote the associated
$(\alpha,\phi)$-expander-hierarchy. Recall that $G^{i+1}$ is the contraction of
$G^i$ w.r.t.\ some $(\alpha/s,\phi/s)$-linked partition $\mathcal{U}_i$ of
$G_i$, i.e., $G^{i+1}=G_{\mathcal{U}_i}^i$. $G^t$ corresponds to the root of
the tree and consists of just a single vertex\footnote{For simplicity we assume that
$G$ is connected; the proof easily generalizes to graphs with several connected
components.}, while $G^0$ is identical to $G$. Let $T_{\ge i}$ denote the
subgraph of $T$ that just contains vertices that have at least distance $i$ to
the leaf-level, i.e., $T_{\ge t}$ is just the single root vertex and
$T_{\ge 0}=T$. Note that the leaf vertices in $T_{\ge i}$ are the vertices on
level $i$, which correspond to the nodes in $G^i$.

Let $c$ denote the hidden constant in Lemma~\ref{lem:helper} and define
$a:=\frac{cs}{\alpha}\log m$ and $b:=\frac{cs}{\phi}\log m$. This means
$q_H\le aq+b$ in Lemma~\ref{lem:helper}. We show by induction that $T_{\ge i}$
is a sparsifier for $G^i$ with quality $b\frac{a^{t-i}-1}{a-1}+a^{t-i}$. This
clearly holds for $i=t$ as then both graphs are identical (just a single
vertex) and, hence, $T_{\ge t}$ is a sparsifier of quality $1$. Now assume that
the statement holds for $i+1>0$. We prove it for $i$. We want to show that
$T_{\ge i}$ is a sparsifier for $G^{i}$. We know that $T_{\ge i+1}$ is a
sparsifier for $G^{i+1}=G_{\mathcal{U}_{i}}^i$; in addition $T_{\ge i}$ is obtained
from $T_{\ge i+1}$ by adding vertices of $V(G^i)$ and attaching each
vertex $v\in V(G^i)$ to the leaf vertex in $T_{\ge i+1}$ that corresponds to 
the cluster in $\mathcal{U}_i$ that contains $v$. This means we can apply
Lemma~\ref{lem:helper} and obtain that $T_{\ge i}$ is a sparsifier for
$G^i$ of quality
\begin{equation*} 
a\cdot\bigg(b\frac{a^{t-i-1}-1}{a-1}+a^{t-i-1}\bigg)+b
=
b\frac{a^{t-i}-1}{a-1}+a^{t-i}\enspace.
\end{equation*}

Hence, for $i=0$ we obtain that $T=T_{\ge 0}$ is a sparsifier for 
$G^0=G$. The quality is 
$b\frac{a^t-1}{a-1}+a^t=O(c^ts^t\log^t
m\max\{\frac{1}{\alpha},\frac{1}{\phi}\}/\alpha^{t-1})$. This finishes
the proof of Theorem~\ref{thm:EH_gives_R_tree}.
\qed

\section{Fully Dynamic Expander Pruning}
\label{sec:dynamic_expander_pruning}

In this section we prove the following theorem, which generalizes Theorem 1.3 in~\cite{SaranurakW19expander}.


\begin{theorem}[Fully Dynamic Expander Pruning]
\label{thm:dynamic_expander_pruning} Let $0 \leq \alpha, \phi \leq 1$ and $\alpha/\phi \leq w \leq 3/(5\phi)$. There is a deterministic algorithm that given a graph $G = (V,E)$, a cluster $U \subseteq V$ such that $G[U]^{w}$ is an $\phi$-expander, and an online sequence of $k \le \phi \vol_{G[U]^{w}}(U)/120$ edge updates, where each update is an edge insertion or deletion for which at least one of the endpoints is contained in $U$, maintains a pruned set $P \subseteq U$ of vertices such that the following property holds. For each $1\le i\le k$,  let $G_i=(V,E_i)$ be the graph after the $i$th update, and denote by $P_i$ the set $P$ after the $i$-th update. We have 
	\begin{enumerate}
		\item $P_0=\emptyset$, and $P_i\subseteq P_{i+1}$.
		\item $\vol_{G[U]^w}(P_i)\le 32i/\phi$, and $|E_{G}(P_i, U \setminus P_i )| \leq 16 i$.
		\item $|E_G(P_i, V \setminus U)| \leq 16i/\alpha.$
		\item The graph $G_i[U \setminus P_i]^{w}$ is an $(\phi/38)$-expander. 
	\end{enumerate}
	Moreover, the total running time for updating $P_1,\ldots,P_k$ is $O(k \log m/\phi^2)$.
\end{theorem}

While the proof of Theorem~\ref{thm:dynamic_expander_pruning} is similar to the proof of Theorem 1.3 in~\cite{SaranurakW19expander}, there are subtle differences as we need to work with a cluster in a graph and not the whole graph, and more importantly, we need to show that a stronger notion of a graph defined on a cluster remains an expander.

Our algorithmic construction behind Theorem~\ref{thm:dynamic_expander_pruning} uses Incremental Flow algorithm from Lemma~\ref{lem: incrementalFlow} as a subroutine. Concretely, let $G=(V,E)$ be a graph and let $U \subseteq V$  be a cluster that is $G[U]^{w}$ is an $\phi$-expander. Let $\Pi = (\Delta,T,c)$ be a flow problem defined on $G[U]^{w}$ with $\Delta(v) = 0$ for all $v \in U$, $T(v) = \deg_{G[U]^{w}}(v)$ for all $v \in U$ and $c(e) = 2/\phi$ for all $e \in E(G[U]^{w})$. We give $G[U]^{w}$ and $\Pi$ as inputs to the Incremental Flow algorithm of Lemma~\ref{lem: incrementalFlow}. 

We next show how to handle updates in $G$. Consider the insertion or deletion of an edge $e=(u,v)$ in $G$ for which at least one of the endpoints is contained in $U$. For each endpoint $w \in \{u,v\}$ of $e$ such that $w \in U$, we add $8/\phi$ unit of source mass at $w$, i.e., we set $\Delta(w) = \Delta(w) + 8/\phi$, and pass these source injecting operations to Incremental Flow. This completes the description of an iteration and the algorithm.

We next verify that the above algorithm satisfies the properties of Theorem~\ref{thm:dynamic_expander_pruning}. To prove the first property, note that from the Incremental Flow, it is clear that the maintained incremental set $P$ satisfies $P_0=\emptyset$, and $P_i\subseteq P_{i+1}$ for all $1\le i\le k$ and thus $P \subseteq U$ serves as a pruned set in Theorem~\ref{thm:dynamic_expander_pruning}. Next, observing that (i) $\sum_{v \in V} \Delta(v) \leq 16i/\phi$ after $i$ edge updates and (ii) $\min_e \{c_e\} = 2/\phi$, and using the second guarantee of Incremental Flow in Lemma~\ref{lem: incrementalFlow}, we get that $\vol_{G[U]^w}(P_i) \leq 32 i / \phi$ and $|E_G(P_i, U \setminus P_i)| = |E_{G[U]^w}(P_i, U \setminus P_i)| \leq 16 i$, thus proving the second property of Theorem~\ref{thm:dynamic_expander_pruning}. 

The third property, i.e., the bound on the connectivity between the pruned set $P_i$ and $V \setminus U$, is proved in the lemma below. Throughout, recall that $k \leq \phi \vol_{G[U]^w} (U)/120$ from Theorem~\ref{thm:dynamic_expander_pruning}.

\begin{lemma}
Let $P_i$ be the pruned set. Then $|E_G(P_i, V \setminus U)| \leq 16i/\alpha$. 
\end{lemma}
\begin{proof}
By the second guarantee of Incremental Flow in Lemma~\ref{lem: incrementalFlow}, we have that 
\begin{equation} 
\label{eq:volPrunBound}
	\vol_{G[U]^w} (P_i) \leq 2 \cdot 16k/\phi \leq \vol_{G[U]^{w}}(U)/2,
\end{equation}
and
\begin{equation} 
\label{eq:boundaryPrunBound}
	|E_G(P_i, U \setminus P_i)| = |E_{G[U]^{w}}(P_i, U \setminus P_i)| \leq 2 \cdot 16/\phi \cdot \phi/2 = 16i.
\end{equation} 

As $G[U]^{w}$ is an $\phi$-expander, it follows that $\vol_{G[U]^{w}}(P_i) \leq 1/\phi \cdot |E_G(P_i, U \setminus P_i)|$, and thus $\vol_{G[U]^{w}}(P_i) \leq 16i/\phi$. Moreover, $\vol_{G[U]^{w}}(P_i) \geq w \cdot |E_G(P_i, V \setminus U|)$ by definition of $G[U]^{w}$. Combining these two bounds and since $w \geq \alpha/\phi$, it follows that $|E_G(P_i, V \setminus U)|  \leq 16i/(\phi w) \leq 16i/\alpha$. \qedhere
\end{proof}

It remains to show the fourth property, i.e., the graph $G_i[U \setminus P_i]^{w}$ is an $(\phi/38)$-expander for all $1 \leq i \leq k$.

\begin{lemma}
\label{lem:expansionPruning}
	Let $\alpha/\phi \leq w \leq 3/(5\phi)$. The graph $G_i[U \setminus P_i]^{w}$ is an $(\phi/38)$-expander.
\end{lemma}

We will prove the above lemma through several steps. We start by bounding the total amount of mass injected in any subset of the cluster $U \setminus P_i$. To this end, let $A := U \setminus P_i$ be the cluster after pruning
the set $P_{i}$. The Incremental Flow subroutine guarantees that the flow
problem $(\Delta',T',c')$ is feasible on $G{[U]^w}[A]^1$\footnote{To explain the notation, let $H = G{[U]^w}$. We have $G{[U]^w}[A]^1 = H[A]^1$.} where $\Delta'(v)=\Delta(v)+2/\phi\cdot|\{e\in E_G(P_{i},A)\mid v\in e\}|$
(it is crucial to note here that the flow problem is defined on $G{[U]^w}[A]^1$ and not on $G_i{[U]^w}[A]^1$), and $T'(v)=\deg_{G[U]^{w}}(v)$ for all $v\in A$ and $c'(e)=2/\phi$ for all $e\in E(G{[U]^w}[A]^1)$. Let $\Delta'(S) := \sum_{u\in S}\Delta'(u)$ be the total amount of source mass in $S$. We next prove a proposition, which will be instrumental in proving Lemma~\ref{lem:expansionPruning}.

\begin{prop}
	\label{prop:flow in out}
	For any set $S\subseteq A$, $\Delta'(S)\le\vol_{G[U]^w}(S)+\frac{2}{\phi}|E_{G}(S,A\setminus S)|$.
\end{prop}

\begin{proof}
	Consider a feasible flow $f$ for the flow problem $(\Delta',T',c')$ defined on $G{[U]^w}[A]^1$. Recall that $f(v)=\Delta'(v)+\sum_{u}f(u,v)$ and $f(u,v) = - f(v,u)$. It follows that 
\begin{align*}
	\Delta'(S) & =\sum_{v\in S} \Big[f(v)+\sum_{u}f(v,u)\Big]\\
	& \le\sum_{v\in S}T'(v)+\sum_{e\in E_{G[U]^{w}}(S,A\setminus S)}c(e)\\
	& =\vol_{G[U]^w}(S) +\tfrac{2}{\phi}|E_{G}(S,A \setminus S)|. \qedhere
	\end{align*}
\end{proof}

In order to leverage the $\phi$-expansion of the graph $G[U]^{w}$, the following lemma shows how to relate the volume of a subset defined on $G_i[A]^{w}$ with the volume of that subset defined on $G[U]^w$. 

\global\long\def\cluster#1#2#3{#1[#2]^{#3}}
\begin{lem}
\label{lem:stable prune} Let $S\subset A\subset U$. If $\vol_{\cluster{G_i}Aw}(S)\le\frac{1}{2}\vol_{\cluster{G_i}Aw}(A)$,
then $\vol_{\cluster GUw}(S)\le\frac{3}{5}\vol_{\cluster GUw}(A)$. 
\end{lem}

\begin{proof}
From Equation~(\ref{eq:volPrunBound}), note that $\vol_{\cluster GUw}(A) = \vol_{\cluster GUw}(U) - \vol_{\cluster GUw}(P_i)\ge\vol_{\cluster GUw}(U)/2$, which in turn implies that  $k\le\phi\vol_{\cluster GUw}(A)/60$.
We also have that $|E_{G_i}(P_i,A)|\le|E_{G}(P_i,A)|+k \le 17k$ by Equation~(\ref{eq:boundaryPrunBound}). It follows that
\begin{align*}
\vol_{\cluster GUw}(S) & \le\vol_{\cluster{G_i}Uw}(S)+wk && (|\vol_{\cluster GUw}(S) -  \vol_{\cluster{G_i}Uw}(S)| \leq wk )\\
 & \le\vol_{\cluster{G_i}Aw}(S)+wk && (\text{since } A \subset U )\\
 & \le\tfrac{1}{2}\vol_{\cluster{G_i}Aw}(A)+wk && (\text{by assumption of the lemma}) \\
 & \le\tfrac{1}{2}(\vol_{\cluster{G_i}Uw}(A)+w|E_{G_i}(P_i,A)|)+wk && (\text{by definition of } \cluster{G_i}Aw) \\
 & \le\tfrac{1}{2}(\vol_{\cluster GUw}(A)+wk+17wk)+wk &&  (|\vol_{\cluster{G_i}Uw}(A) - \vol_{\cluster GUw}(A)| \leq wk)\\
 & \le\tfrac{1}{2}\vol_{\cluster GUw}(A)+10wk &&  \\
 & \le\tfrac{1}{2}\vol_{\cluster GUw}(A)+\tfrac{1}{10}\vol_{\cluster GUw}(A) && (k\le\phi\vol_{\cluster GUw}(A)/60   \text{ and } w \leq 3/(5\phi))\\
 & \le\tfrac{3}{5}\vol_{\cluster GUw}(A) && \qedhere
\end{align*}
\end{proof}

Before proceeding to the proof of Lemma~\ref{lem:expansionPruning}, we introduce some useful notation. Fix an arbitrary subset $S\subseteq A$. Let $a:=|E_{G}(S,P_{i})|$ and  $c:=|E_{G}(S,A\setminus S)|$ be the number boundary edges of $S$ that cross different parts in the original graph $G$.
Similarly, let $a':=|E_{G_{i}}(S,P_{i})|$ and $c':=|E_{G_{i}}(S,A\setminus S)|$ be the boundary edges of $S$ that cross different parts in the current graph $G_i$. Let $v_a:=\vol_{G[A]^w}(S) , v_u:=\vol_{G[U]^w}(S)$ and $v'_a:=\vol_{G_i[A]^w}(S) , v'_u:=\vol_{G_i[U]^w}(S)$. Note that by Proposition~\ref{prop:flow in out}, we have that $\Delta'(S)\le v_u + \frac{2}{\phi}c$.

\begin{proof}[Proof of Lemma~\ref{lem:expansionPruning}]
Recall that $A = U \setminus P_i$ and consider any $S \subset A$ such that $\vol_{G_i[A]^w}(S) \leq \vol_{G_i[A]^w}(A)/2$. To prove that $G_i[A]^w$ is an $(\phi/38)$-expander, we need to show that $|E_{G_i}(S,A \setminus S)| \geq (\phi/38) \cdot \vol_{G_i[A]^w}(S)$, i.e., $c' \geq (\phi/38) \cdot v'_a$.

To this end, we first show a useful relation using the $\phi$-expansion of $G[U]^{w}$. As $\vol_{G_i[A]^w}(S) \leq \vol_{G_i[A]^w}(A)/2$ holds, by Lemma~\ref{lem:stable prune} we have that $\vol_{\cluster GUw}(S)\le\frac{3}{5}\vol_{\cluster GUw}(A)$. From the latter we get $\vol_{\cluster GUw} (A \setminus S) = \vol_{\cluster GUw} (A) - \vol_{\cluster GUw}(S) \geq \frac{5}{3} \vol_{\cluster GUw}(S) - \vol_{\cluster GUw}(S) \geq \frac{2}{3}\vol_{\cluster GUw}(S)$.

 Therefore,
\begin{align*}
	 |E_G(S, P_i)| + |E_G(S, A \setminus S)| & = | E_{G[U]^{w}}(S, U \setminus S)| \\
	& \geq \phi \cdot \min\{\vol_{\cluster GUw}(S), \vol_{\cluster GUw}(U \setminus S)\} \\
	& \geq \phi \cdot \min\{\vol_{\cluster GUw}(S), \vol_{\cluster GUw}(A \setminus S)\} \\
	& \geq \tfrac{2}{3} \phi \cdot \vol_{\cluster GUw}(S),
\end{align*} 
or $\frac{3}{2\phi}a + \frac{3}{2\phi}c \geq v_u$. The latter, together with Proposition~\ref{prop:flow in out}, implies that the total amount of source mass in $S$ is bounded by 
\begin{equation}
\label{eq:flowOut}
\Delta'(S) \leq v_u + \tfrac{2}{\phi}c \leq \tfrac{3}{2\phi}a + \tfrac{7}{2\phi}c.
\end{equation}

Now, by construction, recall that our algorithm increases the source mass of the endpoints in $U$ from the inserted and deleted edges by $8/\phi$. Moreover, by the first property of Lemma~\ref{lem: incrementalFlow}, the flow problem $(\Delta',T',c')$ on $G[U]^w[A]^1$ is feasible. These together imply that $\Delta'(S) \geq \frac{2}{\phi} a + \frac{8}{\phi}|v_u'-v_u|$, $\Delta'(S) \geq \frac{2}{\phi} a + \frac{8}{\phi}|a'-a|$, and $\Delta'(S) \geq \frac{2}{\phi} a + \frac{8}{\phi}|c'-c|$. We claim that
\begin{equation}
\label{eq: absoluteValueBound}
	|v_u'-v_u|, |a'-a|, |c'-c| \leq c/2,
\end{equation}
for otherwise $\Delta'(S) \geq \frac{2}{\phi} a + \frac{8}{\phi} (c/2) = \frac{2}{\phi}a + \frac{4}{\phi}c$, which contradicts Equation~(\ref{eq:flowOut}). Since $\frac{2}{\phi} a \leq \Delta'(S) \leq \frac{1}{\phi} a + \frac{3}{\phi} c$, we get that $a \leq 3c$. It follows that
\begin{align*}
	v'_a = \vol_{G_i[A]^{w}}(S) & \leq \vol_{G_i[U]^{w}}(S) + w \cdot |E_{G_{i}}(S,P_{i})| &&  (\text{by definition of } G_i[A]^{w})\\
	& = v'_u + w a' && \\
	& \leq v_u + \tfrac{c}{2} + w\left(a + \tfrac{c}{2}\right) &&  (\text{Equation~(\ref{eq: absoluteValueBound})})\\
	& \leq \left( \tfrac{3}{2\phi} a + \tfrac{7}{2\phi}c \right) + \tfrac{1}{\phi}\tfrac{c}{2} + \tfrac{3}{5\phi}\left(a + \tfrac{c}{2}\right) && (v_u \leq \tfrac{3}{2\phi}a + \tfrac{7}{2\phi}c \text{ and } w \leq 3/(5\phi)) \\
	& \leq \tfrac{19}{\phi}  c && (a \leq 7c) \\
	& \leq \tfrac{38}{\phi} c', && (\text{Equation~(\ref{eq: absoluteValueBound})})
\end{align*}
what we wanted to show.
\end{proof}

Finally, we analyse the running time. Note that over the course of the algorithm, there are at most $k = \phi \vol_{G[U]^w} (U)/120$ iterations and thus the total amount of mass $\sum_{v \in V} \Delta(v)$ injected in the graph $G[U]^w$ is at most $16/\phi \cdot \phi \vol_{G[U]^w}(U)/120 \leq \vol_{G[U]^w}(U)/3$. The latter implies that the condition on the total mass of Lemma~\ref{lem: incrementalFlow} is met and by the same lemma we get that the running time is bounded by $O(c_{\max} \sum_{v \in V} \Delta(v) \log m) = O(k \log m /\phi^2)$. This completes the proof of Theorem~\ref{thm:dynamic_expander_pruning}.

\def\tO{\tilde{O}}

\def\height{{\hbar}}

\def\paralpha{\alpha}
\def\parphi{\phi'}

\def\paralphaS{\alpha_s}
\def\parphiS{\phi'_s}
\def\paralphaSp{\alpha_{s+1}}
\def\parphiSp{\phi'_{s+1}}
\def\cA{\rho}

\def\barP{P}
\def\P{\tilde{P}}
\def\poly{\operatorname{poly}}

\def\hatn{\bar{n}}
\def\hatm{\bar{m}}

\section{Fully Dynamic Expander Hierarchy}
\label{sec:dynamic_expander_decomposition}

In this section we deal with an undirected unweighted (multi-)graph $G$ that
undergoes a sequence of fully adaptive vertex and edge updates with the
restriction that only isolated vertices may be deleted.

We assume that at any time $G$ contains at most $\hatn$ vertices and at most
$\hatm=\poly(\hatn)$ edges. Further, we fix the following parameters throughout
this section: $\phi=2^{-\Theta(\log^{3/4}\hatn)}$, $\psi=2^{\Theta(\log^{1/2}\hatn)}$,
$\alpha=1/\poly(\log\hatn)$ and we let $h=\log_{\psi}(\hatm)=\Theta(\log ^{1/2}\hatn)$
and $\rho=38^h\psi/\alpha$. The main result of this section is the following theorem.

%
%
%
%
%
%
%
%
%

\begin{theorem}
\label{thm:dynamic_R_tree}
There is a randomized algorithm, for maintaining 
an $(\alpha,\phi)$-expander hierarchy of $G$ with slack $2^{\log^{1/2}(\hatn)}$
in amortized
update time $2^{O(\log^{3/4}(\hatn))}$. 
\end{theorem}


\noindent
The above theorem is based on the following theorem that shows that one can
efficiently maintain an expander-decomposition.

\begin{theorem}
\label{thm:dyn-expander-new}
There is a randomized algorithm that maintains an $(\alpha, \phi)$-expander decomposition
$\cal U$ of $G$  with slack $38^h$ together with its contracted graph $G_\U$ with the following properties:
\begin{itemize}[label=--,itemsep=0ex]
\item update time: $\tilde{O}(\psi\cdot 38^{2h}/\phi^2)$
\item amortized recourse (number of updates to $G_\U$):
$\tO(\rho)=\tO(38^{h}\cdot\psi/\alpha)$\enspace.
\end{itemize}
\end{theorem}

\noindent
With the help of Theorem~\ref{thm:dyn-expander-new} we obtain
Theorem~\ref{thm:dynamic_R_tree} almost immediately.
\begin{proof}[Proof of Theorem~\ref{thm:dynamic_R_tree}]
We maintain an $(\alpha,\phi)$-expander decompositions sequence
$(G^0,\dots,G^t)$ with slack $38^h$. For this we use algorithms
$\tilde A_1, \dots, \tilde A_{t_{\max}}$, where $t_{\max}=2^{O(\log^{1/4}\hatn)}$ is an
upper bound on the depth of the sequence for our choice of $\phi$. The
algorithm $\tilde A_i$ 
observes the updates for graph $G^{i-1}$, maintains an
$(\alpha,\phi)$-expander decomposition $\U^{i-1}$ with slack $38^h$ on this
graph and generates updates for the contracted graph $G^i:=G_{\U^{i-1}}$. The
graph $G^0$ corresponds to the input graph $G$.
The depth $t$ of the maintained expander-hierarchy is determined by
the first graph $G^t$ in this sequence that does not contain any edges.

Because of the bounded recourse the number of updates that have to be performed
for a graph $G^i$ in this sequence is at most $\tO(\rho)^i k$, where $k$ is the
length of the update sequence for $G=G^0$. This results in a total update time of
$\tO(k\sum_i\rho^i \psi 38^{2h}/\phi^2)=k2^{\tO(\log^{3/4}\hatn)}$.
\end{proof}

\subsection{Fully Dynamic Expander Decomposition}
In this section we prove Theorem~\ref{thm:dyn-expander-new}.
The theorem follows from the following main lemma.

\begin{restatable}[Main Lemma]{lemma}{main}
\label{lem:main}\label{lem:fully_dynamic_expander_decomposition}
Suppose a graph $G$ initially contains $m$ edges and undergoes a
sequence of at most $O(\phi m/\rho$) adaptive updates
such that $V(G)\le \hatn$ and $E(G)\le\hatm$ always hold. Then there
exists an algorithm that maintains an $(\alpha, \phi)$-expander decomposition
$\cal U$ with slack $38^h$ and its contracted graph $G_\U$ with the following properties:
\begin{enumerate}[itemsep=0ex]
\item update time: $\tilde{O}(\psi\cdot 38^{2h}/\phi^2)$
\item preprocessing time: $\tilde{O}(m/\phi)$
\item initial volume of $G_\U$ (after preprocessing): $\tO(\phi m)$
\item amortized recourse (number of updates to $G_\U$): $O(\rho)=O(38^{h}\cdot\psi/\alpha)$.
\end{enumerate}
\end{restatable}

\begin{proof}[Proof of Theorem~\ref{thm:dyn-expander-new}]
We simply restart the algorithm from the above lemma whenever an update appears
that would exceed the update limit. This means
we have to perform this after $Z=\Theta(\phi m/\rho)+1$ updates.

We have to analyze how this increases
the update time and the recourse. First observe that before a restart the number of edges can be at
most $m+Z$. Thus, the restart requires preprocessing time $\tO((m+Z)/\phi)$.
Amortizing this against the $Z$ updates  increases the amortized update time
by $\tO(\frac{m+Z}{\phi Z})=\tO(m/(\phi Z)+1/\phi)=\tO(\rho/\phi^2)=\tO(\psi\cdot
38^{2h}/\phi^2)$, where the last step follows because $\alpha=O(38^h)$.

The amortized recourse increases as follows. Observe that 
before the restart the total number
of edges in $G_\U$ is at most $\tO(\phi m+Z\rho)$, because we only experienced
$Z$ updates and the amortized recourse is $O(\rho)$. We delete all these edges.
Then we perform a preprocessing step. Since we have at most $m+Z$ edges in $G$, Property~3 from the above lemma
guarantees that this step inserts at most $\tO(\phi(m+Z))$ edges.
Overall this increases the amortized recourse by
$\tO((\phi m+Z\rho+\phi m +\phi Z)/Z)=\tO(\phi m/Z+\rho)=\tO(\rho)$.

This means the restarts only increase the recourse to $\tO(\rho)$.
\end{proof}

In the remainder of this section we define the details for the ED-process,
i.e., the algorithm from Lemma~\ref{lem:main}. 
\subsection*{\mlp}

In order to define the details of the ED-process we first
define a different process called \mlp. A variant of this process
will serve as a sub-routine in the ED-process.

The input for the \mlp process is a cluster $U$ that is
$(\paralpha,\parphi)$-linked for parameters $\paralpha,\parphi$ that are known to the
process. Then the process receives up to $N\le\parphi\vol_G(U)/\cA$ many updates
for $G$ that are relevant for $U$, i.e., updates of edges for which at least
one endpoint is in $U$. We will refer to $N$ as the \emph{update limit}.
The process maintains a collection of pruned sets
$\barP^1,\ldots,\barP^\height$ such that $U\setminus \bigcup_s\barP^s$
is $(\paralpha/38^{\hbar},\parphi/38^{\hbar})$-linked in $G$. Here
$\height=\lceil\log_\psi(N)\rceil\le h$.


The pruned sets are generated by a hierarchy of algorithms $A_{\hbar},\dots,A_1$. The
algorithm $A_s$ maintains a set $\P^s$ and from time to time it changes
$\barP^s$ to the current value of $\P^s$. In this respect $\barP^s$ is a
\enquote{snapshot} of $\P^s$ from an earlier time step. In the following
$\P^s_t$ and $\barP^s_t$ denote the sets $\P^s$ and $\barP^s$ right after the
$t$-th update.


The precise relationship between $\P^s$ and $\barP^s$ is as follows. For
constructing/maintaining its sets the level $s$ algorithm $A_s$ partitions the
update sequence into \emph{batches} of length
\begin{equation*}
\ell_s:= \left\{\begin{array}{ll}N&\text{if $s=\hbar$}\\ \psi^s&\text{otherwise.}\end{array}\right.
\end{equation*}
each of which
is partitioned into sub-batches of length $\ell_{s-1}$ ($\ell_0=1$). The $i$-th batch on
level $s$ contains updates number $(i-1)\ell_s+1,\dots,i\ell_s$. The $j$-th
sub-batch of the $i$-th batch contains updates
$(i-1)\ell_s+(j-1)\ell_{s-1}+1,\dots,(i-1)\ell_{s}+j\ell_{s-1}$. As in general
$N\neq\psi^{\hbar}$ we allow the last batch for an
algorithm to be incomplete and contain less than $\ell_s$ updates.

The algorithm $A_s$ takes a \enquote{snapshot} of $\P^s$ at the start of every
sub-batch. This means we define $\barP^s_t:=\P^s_{\lfloor
  t/\ell_{s-1}\rfloor\ell_{s-1}}$ if $t$ does not start a new batch; otherwise
$\barP^s_t:=\emptyset$ as $\P^s_t$ is reset at the start of a batch.


How is a set $\P^s_t$ constructed? The construction of the set $\P^s_t$ on level $s$
depends on the sets $\barP^{s'}_t$, $s'>s$. Let $Q^s_t:=\bigcup_{s'>s}\barP^{s'}_t$
and observe that this set does not change during a batch for algorithm $A_s$.
At the beginning of a batch $A_s$ initializes $\P^s:=\emptyset$ (since
  this is also the start of a sub-batch it means also $\barP^s=\emptyset$ at
  this point). Then it
simulates a run of the algorithm for fully dynamic expander pruning
(Theorem~\ref{thm:dynamic_expander_pruning}) on subset $U\setminus Q_t^s$
for the $\ell_s$ updates of the batch. For this run it uses parameters
$\paralphaS:=\paralpha/38^{\height-s}$ and $\parphiS:=\parphi/38^{\height-s}$
and $w:=\paralpha/\parphi$.

In order for the simulation to be valid we have to make sure that the
preconditions of Theorem~\ref{thm:dynamic_expander_pruning} are met.
In particular we require that $U\setminus Q_t^s$ is
$(\paralphaS,\parphiS)$-linked and that
the number of updates in a batch is at most the update limit of the expander
pruning algorithm in Theorem~\ref{thm:dynamic_expander_pruning}.

\begin{claim}[Correctness]
\label{cla:correctness}
For $s\in \{0,\dots,\hbar\}$ the following properties hold.
\begin{enumerate}
\item $U\setminus Q_t^s$ is $(\paralphaS,\parphiS)$-linked; \label{correctA}
\item $\ell_{s}\le\parphiS\vol_G(U\setminus Q^s_t)/120\le\parphiS\vol(G[U\setminus Q^s_t]^w)/120$;   \label{correctB}
\end{enumerate}
\end{claim}
\begin{proof}
We prove the lemma via induction. For the base case $s=\hbar$ the
set $Q_t^{\hbar}$ is empty. Then the above properties directly follow from the
precondition of the input cluster $U$ and the fact that
$N\le\parphi\vol_G(U)/\cA$.

Now, suppose that the statement holds for $s+1$. We prove it for $s$.
From the fact that the statement holds for $s+1$ we are guaranteed that the
simulation of the dynamic expander pruning that is performed by algorithm
$A_{s+1}$ is valid.

Part~1 follows because $U\setminus Q_t^s$ is the
unpruned part that results from the execution of 
Theorem~\ref{thm:dynamic_expander_pruning} by $A_{s+1}$. This theorem guarantees
that $G[U\setminus Q_t^s]^w$ is a $(\phi_{s+1}/38)$-expander with $w=\paralpha_{s+1}/\parphi_{s+1}$. But this also means
that $G[U\setminus Q_t^s]^w$ is a $\phi_{s}$-expander with $w=\paralphaS/\parphiS$.

Theorem~\ref{thm:dynamic_expander_pruning} also gives the following property
for $\barP_t^{s+1}$:
\begin{equation*}
\vol_G(\barP_t^{s+1})\le\tfrac{32}{\parphiSp}\ell_{s+1}\le\tfrac{32}{\parphiSp}\parphiSp\vol_G(U\setminus
Q^{s+1}_t)/120
\le \tfrac{1}{2}\vol_G(U\setminus Q_t^{s+1})\enspace,
\end{equation*}
where Step~1 is due to Theorem~\ref{thm:dynamic_expander_pruning} and Step~2 is
due to induction hypothesis.
Hence,
\begin{equation*}
\begin{split}
\vol_G(U\setminus Q_t^s)
&=\vol_G(U\setminus Q_t^{s+1})-\vol_G(\barP_t^{s+1})
\ge \tfrac{1}{2}\vol_G(U\setminus Q_t^{s+1})\enspace.
\end{split}
\end{equation*}
We can use this relationship to obtain a bound on $\ell_s$, which gives the
second part of the claim. We differentiate two cases. If $s=\hbar-1$ we get
\begin{equation*}
\begin{split}
\ell_s
&\le N\le \parphi\vol_G(U)/\cA =  38\parphi_s\vol_G(U\setminus
Q_t^{s+1})/\cA
\le 76\phi_s'\vol_G(U\setminus
Q_t^{s})/\cA\\
&\le \parphiS\vol_G(U\setminus
Q_t^{s})/120\enspace,\\
\end{split}
\end{equation*}
where the equality uses the fact that $U\setminus Q_t^{s+1}=U$ for $s+1=\hbar$.
If $s<\hbar-1$ we have
\begin{equation*}
\begin{split}
\ell_s
&=\ell_{s+1}/\psi
\le\parphi_{s+1}\vol_G(U\setminus
Q_t^{s+1})/(120\psi)
\le38\parphi_{s}2\vol_G(U\setminus
Q_t^{s})/(120\psi)\\
&\le\parphiS\vol_G(U\setminus Q_t^s)/120\enspace,
\end{split}
\end{equation*}
for sufficiently large $n$ as $\psi=\omega(1)$. This gives Part~2 of the claim. 
\end{proof}

\noindent
Claim~\ref{cla:correctness}
guarantees that a pruned set $\barP^s_t$ is generated by a
valid run of the expander pruning algorithm from
Theorem~\ref{thm:dynamic_expander_pruning} on cluster $U\setminus Q_t^s$
with parameters $\parphiS, \paralphaS$.
Therefore it fulfills the following properties guaranteed by this
theorem.

\begin{claim}\label{cla:help}
A set $\barP_t^s$ fulfills the following properties.
\begin{enumerate}[itemsep=0pt]
\item $\vol_G(\barP^s_t)\le 32\ell_s/\parphiS=O(\psi^s/\parphiS)$\hfill(from Property~2a in Theorem~\ref{thm:dynamic_expander_pruning})\label{hpro:A}
\item $|E_G(\barP^s_t,U\setminus Q_t^s\setminus\barP_s^t)|\le 16\ell_s=O(\psi^s)$\hfill({from Property~2b in Theorem~\ref{thm:dynamic_expander_pruning}})\label{hpro:B}
\item $|E_G(\barP^s_t,V\setminus (U\setminus Q_t^s))|\le 16\ell_s/\paralphaS=O(\psi^s/\paralphaS)$\hfill({from
  Property~3 in Theorem~\ref{thm:dynamic_expander_pruning}})\label{hpro:C}
\item $\out_G(\barP^s_t)\le 32\ell_s/\paralphaS=O(\psi^s/\paralphaS)$ \label{hpro:D}
\end{enumerate}
\end{claim}
\begin{proof}
The first three properties are directed consequences of
Theorem~\ref{thm:dynamic_expander_pruning}. The last one follows from
Property~2 and Property~3 because
$\out_G(\barP_s^t)
=|E_G(\barP_s^t,V\setminus (U\setminus Q_t^s))|
+|E_G(\barP_s^t,U\setminus Q_t^s\setminus\barP_s^t)|$.
\end{proof}

\def\cut{\operatorname{cut}}

\begin{claim}
\label{cla:cluster-properties}
At any time, the cluster $U\setminus\bigcup_s\barP^s_t=U\setminus Q_t^0$ maintained by the multilevel
pruning process fulfills the following properties:
\begin{enumerate}[itemsep=0pt]
\item $U\setminus Q_t^0$ is $(\paralpha/38^h,\parphi/38^h)$-linked in $G$ \label{cluster-proA}
\item $\vol(U\setminus Q_t^0)\ge\tfrac{1}{2}\vol_G(U)$.\label{cluster-proB}
\item $\cut_U(Q_t^0)\le 48N$.\label{cluster-proC}
\end{enumerate}
\end{claim}
\begin{proof}
Part~1 directly follows by applying the above Claim~\ref{cla:correctness} for
$s=0$ and using $\hbar\le h$.
For the remaining parts first observe that
$\sum_{s=1}^\hbar\ell_s=\sum_{s=1}^{\hbar-1}\psi^s+N\le 3N$. For Part~2 we estimate
$\vol_G(Q_t^0)$ by
\begin{equation*}
\vol_G(Q_t^0)\le{\textstyle \sum_s}\vol_G(\barP^s_t)\le {\textstyle
  \sum_s}32\ell_s/\parphi_s\le32\cdot 38^h/\parphi\cdot 3N
\le \vol_G(U)/10\enspace,
\end{equation*}
where the second step uses Property~1 from Claim~\ref{cla:help}, and the last
step uses $N\le\phi'\vol_G(U)/\cA$. This implies $\vol_G(U\setminus Q_t^0)\ge \vol_G(U)/2$.
Part~3 follows because
\begin{equation*}
\begin{split}
\cut_U(Q_t^0)
&=|E_G(Q_t^0,U\setminus Q_t^0)|
=|E_G({\textstyle\bigcup_s}\barP_t^s,U\setminus Q_t^0)|
=\textstyle{\sum_s}|E_G(\barP_t^s,U\setminus Q_t^0)|\\
&\le\textstyle{\sum_s}|E_G(\barP_t^s,U\setminus Q_t^{s+1})|
\le\textstyle{\sum_s}16\ell_s
\le 48N\enspace,
\end{split}
\end{equation*}
where the second inequality is due to Claim~\ref{cla:help} (Part~\ref{hpro:B}).
\end{proof}

\subsection*{\ep and \cp}
The \ep process (ED-process) for maintaining the expander decomposition is a process that
uses a variant of the \mlp process as a sub-routine. We refer to
this variant as a \cp process (CD-process). The ED-process
gets as input a cluster $U$ in a graph $G$, parameters $\paralpha, \parphi$ and
a sequence of updates relevant for $U$. 
It first computes an $(\paralpha,\parphi)$-linked expander decomposition $\U$ of
$U$ using Theorem~\ref{thm:decomp}. For each $U_i\in\U$ with expansion
parameter $\phi_i$ it then starts a CD-process on $U_i$ with parameters
$\paralpha$ and $\phi_i$.

A CD-process is a \mlp process with a slight tweak. Whenever, the
\mlp process (with parameters $\paralpha,\parphi$) as described in
the previous section changes a set $\barP^s$, the CD-process starts an
ED-process on this set (with parameters $\paralpha,\parphi$). 
%

There is one further complication in the definition of an ED-process, which
concerns the update limits of the CD-processes. An ED-process handles
CD-processes for several clusters. It may happen that one of the CD-processes
on some cluster $U_i$ reaches its update limit $N$---we say the CD-process
\emph{expires}. In this case if another update for the cluster appears
the ED-process does the following: it uses
Theorem~\ref{thm:decomp} on the cluster $U_i$ with parameters $\paralpha, \parphi$
and starts a new CD-process on each generated sub-cluster $U_{ij}$ (with parameters
$\paralpha, \phi_j$). We call this step a \emph{restart} of cluster $U_i$. 

There is one subtle issue about the above definition. The CD-process is
recursive. The non-recursive case happens when $\vol_G(U)< \rho/\phi'$. Then
the CD-process has an update limit $N=\lfloor\phi'\vol_G(U)/\rho\rfloor=0$.
This means any update triggers a restart of the CD-process, which results in
computing an expander decomposition for $U$ from scratch.

\begin{observation}
The parameter $\parphi$ passed to a CD-process or an ED-process on any level of
the recursion is at least $\phi$, where $\phi$ is the parameter for the root ED-process.
\end{observation}

The following claim means that an ED-process automatically fulfills
Property~\ref{ed:expanding} and Property~\ref{ed:localboundary} of an
$(\alpha,\phi)$-boundary-linked expander decomposition with slack $38^h$.

\begin{claim}
\label{cla:edproperties}
An ED-process with parameters $\alpha,\phi$ maintains a cluster-partition $\U$ s.t.\
\begin{itemize}[label=--,itemsep=0pt]
\item A cluster $U_i\in \U$ is $(\alpha/38^{h},\phi_i/38^{h})$-linked, with $\phi_i\ge\phi$
\item A cluster $U_i\in \U$ fulfills $\out_G(U_i)\le\tO(\phi_i\vol_G(U_i))$.
\end{itemize}
\end{claim}
\begin{proof}
The clusters that are maintained by the ED-process are the sets
$U\setminus Q_t^0$ that are maintained by the various CD-processes on various
levels of the recursion ($U$ being the set on which the process was started).

A CD-process is always started with parameters
$(\alpha,\varphi)$ for a set $U$ that results from the expander-decomposition
algorithm of Theorem~\ref{thm:decomp} (run with parameters $\alpha$ and
$\phi'\ge\phi$). This set is $(\alpha,\varphi)$-linked for $\varphi\ge\phi'$
according to this theorem. Then the first part is a direct consequence of
Claim~\ref{cla:cluster-properties} (Part~\ref{cluster-proA}).

For the second part again observe that a CD-process is started with parameters
$(\alpha,\varphi)$ for a set $U$ that results from the expander-decomposition
algorithm of Theorem~\ref{thm:decomp}. This gives that
$\out_G(U)\le\tO(\varphi\vol_G(U))$ at the start of the CD-process.

The update limit guarantees that at most $O(\varphi\vol_G(U))$ updates are performed.
This means while the CD-process is active it fulfills
$\out_G(U)\le\tO(\varphi\vol_G(U))$.
Claim~\ref{cla:cluster-properties} (Part~\ref{cluster-proC}) guarantees that
the additional edges that are added by the pruning process are
$O(N)=\tO(\varphi\vol_G(U))$, as well. This means
$\out_G(U\setminus Q^0_t)=\tO(\vol_G(U))$. Finally,
Claim~\ref{cla:cluster-properties} (Part~\ref{cluster-proB})
gives that $\vol_G(U\setminus Q^0_t)\ge\vol_G(U)/2$.
\end{proof}

\paragraph{Amortized Update Time of the CD-process.}
In the following we analyze the amortized update time of a CD-process. The cost
for a CD-process on some cluster $U$ with parameters $\paralpha,\parphi$
consists of the following parts:
\begin{enumerate}[A.] 
\item \label{costA} The cost for each algorithm $A_s$, which consists of
      \begin{enumerate}[1.]
          \item the cost for simulating the expander pruning algorithm from
                Theorem~\ref{thm:dynamic_expander_pruning};
          \item the cost for starting an ED-process on each generated set $\barP_t^s$.
      \end{enumerate}
\item \label{costD}The cost for maintaining pointers from vertices in $G$ to the
corresponding cluster vertex in $G_\U$.
\item \label{costB}The cost for performing a restart on the cluster $U$ when the CD-process expires.
\item \label{costC}The recursive cost incurred by CD-processes on lower levels of the
      recursion.
\end{enumerate}

\noindent
\textit{Part~\ref{costA}~}
The amortized cost for $A_s$ to simulate a step of the expander pruning algorithm is
$\tilde{O}((\parphiS)^2)=O((38^{\hbar-s}/\parphi)^2)$ due to Theorem~\ref{thm:dynamic_expander_pruning}.
Summing over all $s$ gives that the amortized simulation cost (i.e.\ cost A.1) over all
algorithms $A_s$ is only $O(38^{2\hbar}/\parphi^2)$.
%
%

The cost for starting an ED-process on each generated set $\barP_t^s$ is
dominated by the running time for an expander-decomposition algorithm from Theorem~\ref{thm:decomp}.
This has running time $\tO(\out_G(\barP_t^s)/\parphi^2+\vol_G(\barP_t^s)/\parphi)$.
From Part~\ref{hpro:B} and Part~\ref{hpro:C} of Claim~\ref{cla:help} we get
$\out_G(\barP_t^s)=O(\psi^s/\paralphaS)$ and $\vol_G(\barP_t^s)=O(\psi^s/\parphiS)$.
This gives a total running time of

\begin{equation*}
\tO\Big(\psi^s(\tfrac{1}{\paralphaS\parphi^2}+\tfrac{1}{\parphiS\parphi})\Big)
=\tO\Big(\psi^s38^{\hbar-s}(\tfrac{1}{\paralpha\parphi^2}+\tfrac{1}{\parphi^2})\Big)
=\tO\Big(\psi^s38^{\hbar-s}/(\alpha\parphi^2)\Big)\enspace.
\end{equation*}

We
have to perform this operation at the end of every sub-batch, i.e., we can amortize the cost
against the $\ell_{s-1}=\psi^{s-1}$ updates in the sub-batch. Therefore, the
amortized cost of algorithm $A_s$ to start an ED-process is
$\tO(\psi38^{\hbar-s}/(\alpha\parphi^2))$. Summing this over all $s$ gives
an amortized cost for Part~A.2 of $\tO(\psi38^\hbar/(\alpha\parphi^2))$.

Combining this with the simulation cost gives that the amortized cost for
Part~A is
$\tO(\psi38^{\hbar}/\parphi^2\cdot(38^{\hbar}+1/\alpha))
=\tO(\psi38^{2\hbar}/\parphi^2).
$

\medskip
\noindent
\textit{Part~\ref{costD}~} Whenever we initialize a CD-process we also
initialize a pointer for every vertex in $U$ to point to the cluster node in
$G_\U$ corresponding to the set $U$. Later this pointer is changed for the
nodes that are pruned as these then belong to different clusters. However, the
running time is $O(|U|)$ whenever we intialize a CD-process on some cluster
$U$. As a CD-process is always started on some cluster that results from the
static expander-pruning algorithm of Theorem~\ref{thm:decomp} we can amortize
the cost for maintaining poionters against the cost of the expander-pruning
algorithm

\medskip
\noindent
\textit{Part~\ref{costB}~}
The cost for a restart is dominated by running the expander-decomposition
algorithm from Theorem~\ref{thm:decomp} for the cluster $U$. This is 
$\tO(\out_G(U)/\parphi^2+\vol_G(U)/\parphi)$ due to Theorem~\ref{thm:decomp}.
Observe that a restart is only performed after $\lfloor\phi'\vol_{G'}(U)/\cA\rfloor+1$
updates (the first $\lfloor\phi'\vol_{G'}(U)/\cA\rfloor$ to reach the update
limit and one further update to trigger a restart).  Here,
$G'$ refers to the graph at the time that the CD-process on $U$ was started. 
We use the following claim.
\begin{claim}
\label{cla:superhelpful}
Let $G$ and $G'$ denote the current graph and the graph at the
time that the CD-process was started, respectively. Then 
\begin{itemize}[label=--,itemsep=0pt]
\item$\vol_G(U) = O(\vol_{G'}(U))$
\item$\out_G(U) = \tO(\phi'\vol_{G'}(U))$
\end{itemize}
\end{claim}
\begin{proof}
We have
$\vol_G(U)
\le \vol_{G'}(U)+2\parphi\vol_{G'}(U)/\rho=O(\vol_{G'}(U))$, 
because an update can increase the volume by at most 2. Further, we have
$\out_G(U)
\le \out_{G'}(U)+2\parphi\vol_{G'}(U)/\rho
=\tO(\phi'\vol_{G'}(U))$, 
because Property~\ref{ed:localboundary} of an expander
decomposition gives $\out_G'(U)\le\tO(\parphi\vol_{G'}(U))$.
\end{proof}

\noindent
Using these inequalities we get that the cost for the restart is at most
$\tO(\out_G(U)/\parphi^2+\vol_G(U)/\parphi)
=\tO(\vol_{G'}(U)/\phi')$. 
We can amortize this cost against $N+1$ 
updates. This gives an amortized
cost of $\tO(\rho/\parphi^2)$ for the restarts.

\medskip
\noindent
\textit{Part~\ref{costC}~}
The following claim shows that incorporating the cost for the recursion only
increases the cost by a logarithmic factor.
\begin{claim}
An update that is relevant for a CD-process {\cal C} on cluster $U$ can be relevant
for at most $O(\log(\vol_G(U)))$ sub CD-processes of {\cal C}.
\end{claim}
\begin{proof}
At any time the set of subsets on which a CD-process is running forms a laminar
family. In addition suppose we have a CD-process running on subset $U$. A sub
CD-process will  be startet on some subset of a set $\barP_t^s$ maintained by
the CD-process (recall that the CD-process starts an ED-process on $\barP_t^s$,
which in turn partitions the set and then starts a CD-process on each part of
the partition).

According to Claim~\ref{cla:help} we have
$\vol_G(\barP_t^s)\le 32\ell_s/\parphiS\le 32\vol_G(U\setminus
Q_t^s)/120\le\vol_G(U)/2$, where the second Step uses Part~\ref{correctB} of
Claim~\ref{cla:correctness}. This implies that the height of the recursion is
only $\log(\vol_G(U))$.

An edge upate is relevant for a CD-process on set $U$ if at least one end-point of the edge is
contained in $U$. This can happen for at most $2\log(\vol_G(U))$ CD-processes.
\end{proof}

\medskip
\noindent
We get the following lemma.

\begin{lemma}
\label{lem:cdcost}
The amortized time for performing an update operation with the CD-process is
$\tO(\psi38^{2h}/\phi^2)$.
\end{lemma}
\begin{proof}
The lemma simply follows by combining the costs from all parts and using
$\phi'\ge\phi$ and $\hbar\le h$.
\end{proof}

\begin{lemma}
\label{lem:edtime}
An ED-process with parameters $\alpha,\phi$ started
on a graph $G$ with $m$ edges has an
amortized update time of $\tO(\psi38^{2h}/\phi^2)$ and a pre-processing
time of $\tO(m/\phi)$.
\end{lemma}
\begin{proof}
The ED-process triggers an update for at most 2 CD-processes on the top level.
Each of these incurs amortized cost $\tO(\psi38^{2\hbar}/\parphi^2)$ with
$\phi'\ge\phi$ according to Lemma~\ref{lem:cdcost}. The preprocessing consists
of executing the algorithm of Theorem~\ref{thm:decomp}, which has a running time of
$\tO(m/\phi)$.
\end{proof}

\paragraph{Amortized Recourse of the ED-process.}
In the following we derive a bound on the recourse generated by the root ED-process
running on the graph $G$.
We have to analyze how many edge insertions or deletions are generated for the
contracted graph $G_\U$, where $\cal U$ is the decomposition maintained by
the ED-process. Since we care about the amortized recourse we can focus on edge
insertions and amortize the deletions against the insertions at a loss of a
factor of $2$.

The partition maintained by the ED-process changes whenever a CD-process in
the recursion hierarchy changes one of its sets $\barP_t^s$. Fix a CD-process
${\cal C}$
on some subset $U$ and assume there is an update relevant for $U$. 

Let $\bar{s}$ denote the unique level for which the current time-step ends the
current sub-batch of algorithm $A_{\bar{s}}$ without also ending the current
batch. All sets $\barP_t^{s}$ with $s<\bar{s}$ will be reset to $\emptyset$
because the batch ends and all sets $\barP_t^{s}$ with $s>\bar{s}$ will not
change. We first issue edge-deletions for all edges that are at the border of
some partition of the ED-process inside $\bigcup_{s\le\bar{s}}\barP^s_t$ (i.e.,
they contain exactly one vertex of a sub-partition) and also issue vertex
deletions for the corresponding vertices of $G_\U$ (we do not have to count
these deletions because of the amortization described above).

Then we run the expander decomposition algorithm on the new set $P_t^{\bar{s}}$. For
each edge in $G$ that afterwards is at the border of a subset in the partition of $P_t^{\bar{s}}$
we issue an edge-insertion. This gives that the total number of insertions is
$\sum_i\out_G(U_i)$, where $U_1,\dots,U_k$ are the subsets of the partition.
According to Theorem~\ref{thm:decomp} this is at most
\begin{equation*}
\begin{split}
\tO\big(\out_G(\barP^s_t)+\phi'\vol_G(P^s_t)\big)
\le \tO\big(\psi^s/\paralphaS+\phi'\cdot\psi^s/\parphiS\big)
\le \tO(38^{\hbar-s}\psi^s/\alpha)\enspace.
\end{split}
\end{equation*}
This means that the algorithm $A_s$ of CD-process ${\cal C}$ generates on average
$\tO(38^{\hbar-s}\psi/\alpha)$ updates for $G_\U$ per relevant update for
${\cal C}$. This holds because the above updates for $G_\U$ are only incurred
after $\ell_{s-1}=\psi^{s-1}$
relevant updates for ${\cal C}$. Summing this over all $s$ gives that the
amortized recourse for $\cal C$ due to algorithms $A_s$ is only
$\tO(38^{\hbar}\psi/\alpha)=\tO(\rho)$.

It remains to derive a bound on the recourse that is generated by $\cal C$ when
a restart is triggered. We delete all edges in $G_\U$ that are incident to a current
sub-cluster and we also delete the vertices corresponding to these
sub-clusters. Then we repartition $U$ and start a CD-process on each cluster.
We have to insert all edges that are at the border of a sub-set of the
partition, i.e., $\sum_i\out_G(U_i)$ many edges, where $U_1,\dots,U_k$ denote
the subsets in the partition. We have
\begin{equation*}
{\textstyle\sum_i}\out_G(U_i) \le \tO(\out_G(U)+\phi'\vol_G(U))
\le 
\tO(\phi'\vol_G(U))\enspace,
\end{equation*}
where the first inequality is due to Theorem~\ref{thm:decomp} and the second
due to Claim~\ref{cla:superhelpful}. Since, we can amortize these costs
over $N+1=\lfloor\phi'\vol_{G'}(U)/\cA\rfloor+1$ many relevant updates we obtain
that the amortized recourse due to restarts is only $\tO(\rho)$.

Since a single update is relevant for at most at most $O(\log(\vol(G)))$
CD-processes we obtain the following lemma.

\begin{lemma}
\label{lem:bounded-recourse}
The amortized recourse of the ED-process is $\tO(38^{\hbar}\psi/\alpha)=\tO(\rho)$.
\end{lemma}

\paragraph{Total boundary for the partition of the ED-process.}
So far we have only shown that the ED-process maintains a partition that
fulfills Property~\ref{ed:expanding} and Property~\ref{ed:localboundary} for a boundary-linked partition.
It remains to show that the total number of edges between subsets in the
partition fulfills Property~\ref{ed:globalboundary}.

\begin{claim}\label{cla:totalboundary}
An ED-process with parameters $\alpha,\phi$ on a graph $G$ with $m$ edges
that receives $Z=O(\phi m/\rho)$ updates
maintains a cluster-partition $\U$ such that
\begin{equation*}
{\textstyle\sum_{U_i\in\U}}\out_G(U_i)\le \tO(\phi m)\enspace.
\end{equation*}
\end{claim}

%
\begin{proof}
First the ED-process performs a boundary-linked expander decomposition and then
starts a CD-process on each cluster. At this point the number of edges between
sub-clusters (which equals the number of edges in $G_\U$) is at most $\tO(\phi m)$ according to 
Theorem~\ref{thm:decomp}. Because of the bounded recourse from Lemma~\ref{lem:bounded-recourse}
the total number of edges between vertices in $G_\U$ can be at most
$\tO(\phi m)+Z\cdot \tO(\rho)=\tO(\phi m)$, 
after $Z$ updates.
\end{proof}

\noindent
We are now ready to prove the main lemma.
\main*
\begin{proof}
The first property follows from Lemma~\ref{lem:edtime}. The second and third
property follow from the fact that in the pre-processing we perform the
expander decomposition algorithm from Theorem~\ref{thm:decomp}, which has a
running time of $\tO(m/\phi)$ and generates an expander decomposition that
fulfills $\sum_i\out(U_i)\le \tO(\phi m)$. The bound on the amortized recourse
follows from Lemma~\ref{lem:bounded-recourse}.

Claim~\ref{cla:edproperties} shows that the maintained partition fulfills
Property~\ref{ed:expanding} and Property~\ref{ed:localboundary} of an
$(\alpha,\phi)$-linked expander decomposition with slack $38^h$.
Finally, Claim~\ref{cla:totalboundary} shows that it also fulfills
Property~\ref{ed:globalboundary}.
\end{proof}


\section{Derandomization and Deamortization}
\label{sec:deterministic-EH}
In this section, we show that our algorithm in Theorem~\ref{thm:dynamic_R_tree}, which maintains an $(n^{o(1)},n^{o(1)})$-expander hierarchy of a dynamic graph on $n$ vertices in $n^{o(1)}$ time can be de-randomized and de-amortized easily using the results in~\cite{ChuzhoyGLNPS19det} and~\cite{NanongkaiSW17}.
Throughout the section, we use $\bar O(\cdot)$ to hide $(\log\log n)^{O(1)}$ factors. The main result of this section can be summarized as the following theorem.
\begin{thm}
\label{thm:derand_deamort_main}
There is a deterministic algorithm, that, given a fully dynamic unweighted graph $G$ on $n$ vertices, maintains a data structure representing a $(2^{-\bar O(\log^{2/3}n)},2^{-O(\log^{5/6}n)})$-expander hierarchy with slack $2^{\bar O(\log^{1/2}n)}$ of $G$ in $2^{-O(\log^{5/6}n)}$ worst-case update time and the data structure supports the following query: given a vertex $u\in V(G)$, return a leaf-to-root path of $u$ in the hierarchy in $O(\log^{1/6} n)$ time.
\end{thm}
Compared with Theorem~\ref{thm:dynamic_R_tree}, the algorithm in Theorem~\ref{thm:derand_deamort_main} is deterministic, and also gives worst-case update time guarantees. On the flip side, it does not explicitly maintain a single expander hierarchy, but will constantly switch between several expander hierarchies that we maintain in the background, as we will see later.
The proof of the main theorem consists of two parts, that we will show in the following subsections: the first part shows how to derandomize the algorithm in Theorem~\ref{thm:dynamic_R_tree} using a recent result in \cite{ChuzhoyGLNPS19det}; and the second part shows how to de-amortize the algorithm in Theorem~\ref{thm:dynamic_R_tree}, using similar techniques from~\cite{NanongkaiSW17}. 
We note that, by directly combining the methods from the two subsections, we immediately obtain an algorithm for Theorem~\ref{thm:derand_deamort_main}.
We now describe the two parts in more detail.

\subsection{De-randomization}
In this section we provide the proof of the following theorem.
\begin{thm}
\label{thm:deterministic_dynamic_ED}
There is a deterministic algorithm, that, given a fully dynamic unweighted
graph $G$ on $n$ vertices, explicitly maintains a $(2^{-\bar O(\log^{2/3}n)},2^{-O(\log^{5/6}n)})$-expander
hierarchy with slack $2^{\bar O(\log^{1/2}n)}$ of $G$ in amortized update time $2^{\bar O(\log^{5/6}n)}$.
\end{thm}
Recall that the algorithm in Theorem \ref{thm:dynamic_R_tree} utilizes the algorithm in Lemma~\ref{lem:fully_dynamic_expander_decomposition} as a subroutine, and upon this, everything is deterministic.
Recall also that the algorithm in Lemma~\ref{lem:fully_dynamic_expander_decomposition} utilizes as subroutines the algorithm in Theorem~\ref{thm:dynamic_expander_pruning}, which is deterministic, and the algorithm in Theorem~\ref{thm:decomp} as subroutines, which is randomized, and upon this, everything is deterministic.
Observe that the only randomized part in the algorithm in Theorem~\ref{thm:dynamic_R_tree} is the subroutine of the cut-matching game.
Therefore, the only part that is randomized in the algorithm of Theorem \ref{thm:dynamic_R_tree} is also the cut-matching step in Lemma~\ref{thm:cut-matching}.

A recent result by Chuzhoy et al~\cite{ChuzhoyGLNPS19det} gave a deterministic algorithm for the cut-matching step with weaker parameters, that is stated as follows.
\begin{lemma}
\label{thm:cut-matching_deterministic}
There is a deterministic algorithm, that, given an unweighted graph $G=(V,E)$ with $m$ edges and a parameter $\phi>0$,
\begin{enumerate}
\item either certifies that $G$ has conductance $\Phi_G\geq8\phi$; 
\item or finds a cut $(A,\overline{A})$ of $G$ with conductance $\Phi_G(A)\le\gamma^*\phi$, and $\vol_G(A),\vol_G(\bar{A})$ are both at least $m/(16\gamma^*)$, i.e., we find a relatively balanced low conductance cut;
\item or finds a cut $(A,\bar{A})$, such that $\Phi_G(A)\le\gamma^*\phi$ and $\vol_G(\bar{A})\le m/(16\gamma^*)$, and $A$ is a near $8\phi$-expander;
\end{enumerate}
with the parameter $\gamma^*=2^{O(\log^{2/3}m(\log\log n)^{1/3})}=2^{\bar O(\log^{2/3}m)}$. Moreover, the algorithm runs in time $\tilde O(m\gamma^*/\phi)$,
\end{lemma}
The following corollary is immediately obtained by replacing the randomized cut-matching step in the algorithm of Theorem~\ref{thm:decomp} with the algorithm in the above theorem.
\begin{corollary}\label{cor:decomp_deterministic}
There is a deterministic algorithm that, given a graph $G=(V,E)$, a cluster $C\subseteq V$ with $\vol_{G}(C)=m$ and $|E_{G}(C,V\setminus C)|\le b$, and parameters $\alpha,\phi$ such that $\alpha\le 2^{-\bar O(\log^{2/3}m)}$, computes an $(\alpha,\phi)$-expander decomposition of $C$ in $\tilde{O}(b/\phi^2+m\gamma^*/\phi)$ time, with $\gamma^*=2^{O(\log^{2/3}m(\log\log n)^{1/3})}=2^{\bar O(\log^{2/3}m)}$.
\end{corollary}
To obtain the algorithm for Theorem~\ref{thm:deterministic_dynamic_ED}, we simply replace the randomized cut-matching step with the algorithm in Corollary~\ref{cor:decomp_deterministic}. 
We also change the parameters in  Section \ref{sec:dynamic_expander_decomposition} accordingly as 
$\alpha=2^{-\bar O(\log^{2/3}m)}$ and $\phi=2^{-O(\log^{5/6}n)}$.
We keep
$\psi=2^{O(\sqrt{\log n})}$, so $h=\Theta(\log_{\psi}n)=\Theta(\sqrt{\log n})$ as before, the depth of the hierarchy is $O(\log^{1/6}n)$ and $\rho=O(\psi\cdot 38^{h}/\alpha)=2^{\bar O(\log^{2/3}m)}$.
From the same proof of Section \ref{sec:dynamic_expander_decomposition} (with distinct parameters), we can show that we can maintain an $(\alpha,\phi)=(2^{-\bar O(\log^{2/3}n)},2^{- O(\log^{5/6}n)})$-expander
hierarchy with slack $2^{\bar O(\log^{1/2}n)}$ of an $n$-vertex fully-dynamic graph in amortized update time $O(1/\phi^{2})\cdot\rho^{\dep(T)}=2^{O(\log^{5/6}n)}$.

\subsection{De-amortization}
The main result in this section is the following theorem.
\begin{thm}
\label{thm:deamortized_dynamic_ED}
There is a randomized algorithm, that, given a fully dynamic unweighted
graph $G$ on $n$ vertices undergoing adaptive edge insertions and deletions, maintains a data structure representing a $(2^{-O(\log^{1/2} n)},2^{-O(\log^{3/4}n)})$-expander
hierarchy with slack $2^{O(\log^{1/2} n)}$ of $G$ in $2^{O(\log^{3/4}n)}$ worst-case update
time and the data structure supports the following query: given a vertex $u$, return a leaf-to-root path of $u$ in the hierarchy in $O(\log^{1/4} n)$ time
with high probability.
\end{thm}

In order to construct an algorithm for Theorem~\ref{thm:deamortized_dynamic_ED}, we first show that we can de-amortize the algorithm in Lemma~\ref{lem:fully_dynamic_expander_decomposition}, and then we describe how to use this new algorithm of Lemma~\ref{lem:fully_dynamic_expander_decomposition} to construct an algorithm for Theorem~\ref{thm:deamortized_dynamic_ED}.

The crux in de-amortizing the algorithm in Lemma~\ref{lem:fully_dynamic_expander_decomposition} is to de-amortize the core subroutine: \mlp. 
Recall that \mlp extensively uses the algorithm in Theorem~\ref{thm:dynamic_expander_pruning} for expander pruning, that we denote by $\bset$. However, the algorithm in Theorem~\ref{thm:dynamic_expander_pruning} only guarantees small amortized update time, and cannot be de-amortized.
To overcome this issue, the key observation is that, we cannot sequentially feed the algorithm $\bset$ up to the current update and force it to produce information with respect to the current graph. Instead, when we feed the algorithm $\bset$ with a batch of updates that has already shown up in the update sequence, we have to wait for a certain number of updates that is comparable to the length of the batch that we feed to $\bset$, so that it can distribute the running time for processing the batch that we feed to it evenly to the new updates, thus achieving the worst-case update time guarantee.

In the remainder of this section, we first describe how to de-amortize the \mlp process, and then describe how to use it to further de-amortize the algorithm for Lemma~\ref{lem:fully_dynamic_expander_decomposition} and eventually provide an algorithm for Theorem~\ref{thm:deamortized_dynamic_ED}.

\paragraph{De-amortize \mlp.}
Recall that the \mlp consists of a hierarchy of algorithms $A_1,\ldots,A_{\hbar}$, such that, when an higher-level algorithm produces a pruned set, every lower-level algorithm works on the remaining graph where the pruned set is taken out from $U$.
The high-level intuition for de-amortizing \mlp is to ``delay'' the work in each algorithm, so that there are enough updates for the algorithm to distribute their work on. For this to be accomplished, we will have to incur a multiplicative loss of $\psi$ in the update time.

Let $U$ be the cluster that we run the \mlp process on.
We denote by $D$ the sequence of updates on $G$ that are relevant to $U$. 
For a pair of integers $1\le i<j\le N$, we denote $D(i,j)$ as the subsequence of $D$ from the $i$th update to the $j$th update (including both). We call such a subsequence a \emph{batch}.  
For a cluster $W$ of vertices that is $(\alpha',\phi')$-linked in $G$ upon the $(i-1)$th update, we denote by $\mathcal{B}(W,i,j,\alpha',\phi')$ to be the run of the algorithm in Theorem~\ref{thm:dynamic_expander_pruning}, starting with the cluster $W$ in $G$ with boundary-linkedness parameters $\alpha'$ and $\phi'$, handling the updates in $D(i,j)$. Note that, an update on $G$ may be irrelevant of $W$ (i.e., both endpoints of the updated edge are not in $W$). In this case we simply ignore this update in the run of $\mathcal{B}(W,i,j,\alpha',\phi')$.

Recall that the input for the \mlp process is a cluster that is $(\alpha',\phi')$-linked in $G$ for parameters $\alpha',\phi'$ that are known to the algorithm, and the \mlp process handles the next $N=\phi'\vol_G(U)/1200$ updates that are relevant to $U$.
Recall that $\psi=2^{O(\sqrt{\log n})}$ and $\hbar=\log_{\psi}N$.

The new algorithm consists of a hierarchy of $\hbar-1$ sub-algorithms $A'_{\hbar-1},A'_{\hbar-2},\ldots,A'_{1}$.
We first describe the work of sub-algorithm $A'_{\hbar-1}$. Recall that in Section \ref{sec:dynamic_expander_decomposition}, the work of $A_{\hbar}$ is divided into stages with length $\ell_{\hbar-1}=\psi^{\hbar-1}$ each.
Similarly, the work of $A'_{\hbar-1}$ is also divided into stages with length $\ell_{\hbar-1}$ each. However, the work in each stage is now completely different. In the first stage, $A'_{\hbar-1}$ does nothing.
For each $1\le t\le N/\ell_{\hbar-1}$, note that, at the beginning of the $(t+1)$th stage, the batch $D(1,t\ell_{\hbar-1})$ of updates has completely shown up. We simply let $A'_{\hbar-1}$ run $\bset(U,1,t\ell_{\hbar-1},\alpha',\phi')$ in its $(t+1)$th stage, with the work evenly distributed upon all updates in this stage.
And after this stage is finished, $A'_{\hbar-1}$ sends the pruned set $P^{\hbar-1}_{t\ell_{\hbar-1}}$ to $A'_{\hbar-2}$.
Intuitively, the sub-algorithm $A'_{\hbar-1}$ processes batches of size $[\psi^{\hbar-1}, \psi^{\hbar}]$, and is always ``$\psi^{\hbar-1}$ updates late'' compared to the current update.
From the above discussion, in a stage of $A'_{\hbar-1}$, at most $N=\psi^{\hbar}$ updates are handled. 

We now describe, for each $2\le s\le \hbar-2$, the work of sub-algorithm $A'_{s}$, which is similar to $A'_{\hbar-1}$.
The work of sub-algorithm $A'_{s}$ is divided into stages with length $\ell_s$ each. In the first stage, it does nothing. For each $1\le t\le 2\psi-1$, note that at the end of the $t$th stage on $A'_{s}$, the batch $D(1,t\ell_{s})$ of updates has completely shown up, and $A'_{s}$ has not received anything from $A'_{s+1}$.
In the $(t+1)$th stage, we simply let $A'_{s}$ run $\bset(U,1,t\ell_{s},\alpha',\phi')$ in its $(t+1)$th stage, with the work evenly distributed upon all updates in this stage.
And after this stage is finished, $A'_{s}$ sends the pruned set $P^{s}_{t\ell_{s}}$ to $A'_{s-1}$. Starting from the $(2\psi+1)$th stage on $A'_{s}$, we call every next $\psi$ stages on $A'_s$ a \emph{phase} of $A'_s$. Note that, from the description of the work on $A'_{s+1}$, for each $t'\ge 0$, at the beginning of the $(t'+1)$th phase on $A'_{s}$, it receives a set $P^{s+1}_{t'\ell_{s+1}}$ from $A'_{s+1}$. We now describe the work of $A'_s$ within this phase. For each $0\le t\le \psi-1$, in the $(t+1)$ stage within this phase, we let it run the process
$\bset(\overline{P^{s+1}_{t'\ell_{s+1}}},t'\ell_{s+1},(t'+1)\ell_{s+1}+t\ell_{s},\alpha'/38^{\hbar-s-1},\phi'/38^{\hbar-s-1})$, with the work evenly distributed upon all updates in this stage.
This completes the work on $A'_{s}$.
After the $(t+1)$th stage on $A'_s$ is finished, $A'_{s}$ sends the pruned set $P^{s}_{t\ell_s}$ to $A'_{s-1}$.
Intuitively, the subalgorithm $A_s$ process batches of size $[\psi^{s}, 2\psi^{s+1}]$, and is always ``$\psi^{s}$ updates late'' compared with the current update.
From the above discussion, in a stage of $A'_{s}$, at most $2\psi^{s+1}$ updates are handled. 
 
It remains to describe the work of $A'_1$. While the sub-algorithms $A'_2, A'_3,\ldots,A'_{\hbar-1}$ can be one-stage late, the work on $A'_1$ has to be up-to-date.
The work of sub-algorithm $A'_{1}$ is also divided into stages with length $\ell_1=\psi$ each.
In the first $2\psi$ stages, upon the $i$th update, the sub-algorithm $A'_{1}$ simply runs
$\bset(U,1,i,\alpha'/38^{\hbar-2},\phi'/38^{\hbar-2})$.
Starting from the $(2\psi+1)$th stage on $A'_{1}$, we call every next $\psi$ stages on $A'_1$ a \emph{phase} of $A'_1$. Note that, from the description of the work on $A'_{1}$, for each $t'\ge 0$, at the beginning of the $(t'+1)$th phase on $A'_{1}$, it receives a set $P^{2}_{t'\ell_{2}}$ from $A'_{2}$. We now describe the work of $A'_s$ within this phase. For each $0\le i\le \psi^2-1$, upon the $i$th update stage within this phase, we let it run the entire process
$\bset(\overline{P^{2}_{t'\ell_{2}}},t'\ell_{s+1}+1,t'\ell_{s+1}+i,\alpha'/38^{\hbar-2},\phi'/38^{\hbar-2})$.
Put in other words, within the phase, upon each update, the machine $M_1$ makes an individual run of the algorithm in Theorem~\ref{thm:dynamic_expander_pruning} handling all updates in this phase (from the first update in this phase to the current update).
This completes the work on $A'_{1}$.
From the above discussion, upon each update, $A'_1$ handles a batch of at most $2\psi^{2}$ updates.

From the discussion, the worst-case update time in this de-amortized \mlp is at most $O(\psi^2)$-factor larger than the amortized update time of \mlp. Therefore, the worst-case update-time $O(\psi^2\cdot 38^{2h}/\phi'^2)$.

\paragraph{De-amortize \cp.}
Recall that the input to \cp process is a cluster that is $(\alpha,\phi')$-linked in the current graph $G$.
Also recall that, the CD-process contains a main \mlp process, and additionally, for each set $\tilde P^s$ in the collection $\{\tilde P^1,\ldots,\tilde P^{\hbar}\}$ of sets maintained by the \mlp process, the CD-process recomputes an $(\alpha,\phi')$-expander decomposition on $\tilde P^s$ every time it changes (namely, every $\psi^s$ updates on $U$). For obtaining the expander decomposition in time, we tweak the de-amortized \mlp a bit, by letting each sub-algorithm $A'_s$ runs, in each phase, not only a process of $\bset$ handling a batch of updates, but also an $(\alpha,\phi')$-expander decomposition on the pruned out set of vertices, after completing the process of $\bset$, with the total work of both tasks evenly distributed on all updates in this stage. Note that this increase the worst-case update time by $O(1)$-factor.

\paragraph{De-amortize the algorithm for Lemma~\ref{lem:fully_dynamic_expander_decomposition}.}
Recall that the algorithm for Lemma~\ref{lem:fully_dynamic_expander_decomposition} simply first computes an $(\alpha,\phi)$-expander decomposition, and then, for each cluster in the $(\alpha,\phi)$-expander decomposition, it starts a CD-process on it with respect to the well-linkedness parameter $(\alpha,\phi')$ of this cluster. We have already shown how to de-amortize the CD-process. However, to completely de-amortize the algorithm for Lemma~\ref{lem:fully_dynamic_expander_decomposition}, we need one more step. Note that when the CD-process on a cluster $U$ has handled $N=\phi'\vol_G(U)/1200$ updates, the cluster $U$ will be reset. In particular, the algorithm will recompute an $(\alpha,\phi)$-expander decomposition from scratch on $U$, and then starts a new CD-process on each of the cluster in this decomposition. 
This cluster-resetting step needs to be de-amortized as well.

In order to achieve this, we run three CD-processes on the same cluster $U$ in parallel, each maintaining an $(\alpha/38^h,\phi/38^h)$-expander decomposition of $U$. At any time, one of the CD-process is used by the algorithm (that we call \emph{online}), and the others are temporarily not (that we call \emph{in the background}).
When the online CD-process terminates, we switch it into the background, and bring online another CD-process that was in the background.
We carefully choose the ``offset'' between these CD-process and schedule their work so that at any time, the online CD-process maintains an available decomposition of the current cluster.

We now describe the algorithm in more detail. We maintain $3$ tweaked CD-process in parallel. Each tweaked CD-process has three phases: the preparing phase; the chasing phase; and the working phase; each spans the time of a consecutive $N/3$ updates (recall that $N$ is the number of updates that can be handled by a CD-process). The offset between each pair of tweaked CD-process is also $N/3$.
Assume the input is a cluster $U$ that is $(\alpha,\phi')$-linked in $G$.
Assume that some tweaked CD-process starts at the $k$th update, and we denote by $G_k$ the graph after the $k$th update. In the first phase of the tweaked CD-process, the preparing phase, it computes an $(\alpha,\phi')$-expander decomposition of $U$ in $G_k$. In the second phase, the chasing phase, it handles the batch $D(k+1,k+2N/3)$ of updates. Note that, before this phase, the tweaked CD-process is $N/3$ updates behind, and after this phase, the tweaked CD-process manages to maintain an $(\alpha/38^h,\phi'/38^h)$-expander decomposition of the up-to-date graph. Intuitively, this can be achieved by running a normal CD-process with double speed.
In the third phase, the working phase, it runs a normal CD-process to handle the batch $D(k+2N/3+1,k+N)$ of updates.
The work in the first and the second phases is evenly distributed over all updates in that phase. 
Each tweaked CD-process is online only at its working phase.
It is not hard to see that, the combination of three tweaked CD-process defined above maintains an $(\alpha/38^h,\phi'/38^h)$-expander decomposition of the up-to-date graph, and achieves the worst-case update time within a $O(1)$-factor of the worst-case update time of a normal de-amortized CD-process described above. Therefore, the worst-case update time is $O(\psi^2\cdot 38^{2h}/\phi'^2)$.

\paragraph{Constructing the algorithm for Theorem~\ref{thm:deamortized_dynamic_ED}.}
Recall that the algorithm in Lemma~\ref{lem:fully_dynamic_expander_decomposition} maintains an $(\alpha/38^h,\phi'/38^h)$-expander decomposition, such that the amortized recourse in the contracted graph with respect to the decomposition is $\rho=O(38^h\cdot\psi/\alpha)$. However, to construct an algorithm for Theorem~\ref{thm:deamortized_dynamic_ED} using the de-amortized algorithm for Lemma~\ref{lem:fully_dynamic_expander_decomposition}, we need to ensure that the worst-case recourse is $2^{O(\sqrt{\log n})}$, preferably $O(\rho)$.
This can be achieved by further tweaking the de-amortized algorithm for Lemma~\ref{lem:fully_dynamic_expander_decomposition} a bit. Specifically, for each $i$ and in each stage of $G^{i}$, we not only distribute the running time evenly over all updates, but also distribute the recourse that is needed to propagate to the graph $G^{i-1}$ at one-level above. In this way, we ensure that the worst-case recourse for the graph $G^{i}$ is at most $O(\rho)^i$, thus achieving the worst-case update time 
$O(\psi^2\cdot 38^{2h}/\phi'^2)\cdot O(\rho)^i=2^{-O(\log^{3/4}n)}$.

\section{Applications}
\label{sec:application_dynamic}

In this section we show that our dynamic expander hierarchy almost directly leads to a number of applications in dynamic graph algorithms.

\subsection{Dynamic Tree Flow Sparsifier}
We start by reviewing the notion of tree flow sparsifiers. Given a weighted graph $G=(V,E,\capacity)$ and a subset $S\subseteq V$ of
vertices, a set $D$ of \emph{demands} on $S$ is a function
$D:S\times S\to \mathbb{R}_{\ge 0}$, that specifies, for each pair $u,v\in S$ of vertices, a demand $D(u,v)$. Given a subset
$S\subseteq V$ and a set $D$ of demands on $S$, a \emph{routing} of $D$ in $G$ is a flow $F$ on $G$, where for each pair $u,v\in S$, the amount of flow that $F$ sends from $u$ to $v$ is $D(u,v)$. We define the congestion $\eta(G,D)$ of a set $D$ of demands in $G$ to be the minimum congestion of a flow $F$ that is
a routing of $D$ in $G$. We say that a tree $T$ is a \emph{tree flow sparsifier of quality $q$} for $G$ with
respect to $S$, if $S\subseteq V(T)$, and for any set $D$ of demands on $S$,
$\eta(T,D)\le \eta(G,D)\le q\cdot\eta(T,D)$. A tree flow sparsifier $H$ of $G$
w.r.t.\ the subset $V(G)$ is just called a tree flow sparsifier for $G$. 

We design an algorithm that explicitly maintains a tree flow sparsifier for a graph $G$ that undergoes edge insertions and deletions, which proceeds as follows: given an unweighted dynamic graph $G$ on $n$ vertices, maintain a $(2^{-\bar{O}({\log^{2/3}n})},2^{-O(\log^{5/6}n)})$-expander hierarchy with slack $2^{\bar{O}(\log^{1/2} n)}$ of $G$ using Theorem~\ref{thm:deterministic_dynamic_ED}.

We immediately obtain the following result, which proves Corollary~\ref{thm:tree-flow-sparsifier} from the introduction.

\begin{corollary}
	There is a deterministic fully dynamic algorithm on a graph $G$ with $n$ vertices that explicitly maintains a tree flow sparsifier for $G$ with quality $2^{O(\log^{5/6} n)}$ and depth $O(\log^{1/6} n)$ using $2^{O(\log^{5/6} n)}$
	amortized update time. 
\end{corollary}
\label{cor:dynamicTreeFlowSparsifier}
\begin{proof}
	To bound the quality of the tree flow sparsifier, the main observation is that an expander hierarchy of a graph $G$ is itself a tree flow sparsifier for $G$. Concretely, let $\alpha := 2^{-\bar{O}({\log^{2/3}n})}$, $\phi := 2^{-O(\log^{5/6}n)}$ and $s := 2^{\bar{O}(\log^{1/2} n)}$. By Theorem~\ref{thm:deterministic_dynamic_ED}, the depth of $(\alpha,\phi)$-expander hierarchy we maintain is $ t: = O(\log^{1/6} n)$. Using Theorem~\ref{thm:EH_gives_R_tree}, it follows that our $(\alpha,\phi)$-expander hierarchy of $G$ with slack $s$ and depth $t$ is a tree flow sparsifier for $G$ with quality $O(s\log m)^{t} \cdot O(\max\{\frac{1}{\alpha},\frac{1}{\phi}\}/\alpha^{t-1}) = 2^{O(\log^{5/6} n)}$. 
	
	Since we can maintain an $(\alpha,\phi)$-expander hierarchy with slack $s$ of $G$ in $2^{O(\log^{5/6} n)}$ amortized update time~(Theorem~\ref{thm:deterministic_dynamic_ED}), it follows that the amortized update time for maintaining a tree flow sparsifier for $G$ is also bounded by $2^{O(\log^{5/6} n)}$.
\end{proof}

\subsection{Dynamic Vertex Flow Sparsifiers, Maximum Flow, Multi-commodity Flow, Multi-Way Cut and Multicut}

We show that a dynamic tree flow sparsifier can be used to maintain a tree flow sparsifier with w.r.t. a subset $S$~(also known as \emph{vertex flow sparsifiers}), an approximation to the value of the following problems (i) maximum flow/minimum cut, (ii) maximum concurrent (multi-commodity) flow, (iii) multi-way cut and (iv) multicut. 

In the dynamic vertex flow sparsifier\footnote{In general, vertex sparsifiers that preserve the (multi-commodity) flow between terminal vertices are not restricted to tree instances. However, as a byproduct of our techniques, the vertex sparsifiers we consider in this paper are always trees.} problem, the graph $G$ undergoes insertions or deletions of edges and the following queries are supported: given any subset $S \subseteq V(G)$, return a tree flow sparsifier for $G$ w.r.t. $S$. The main idea behind designing an algorithm for this problem is the observation that given a tree flow sparsifier for $G$, one can easily extract a tree flow sparsifier for $G$ w.r.t. any subset $S \subseteq V(G)$. Concretely, given an unweighted dynamic graph $G$ on $n$ vertices, let $T$ be the maintained tree flow sparsifier for $G$ from Corollary~\ref{cor:dynamicTreeFlowSparsifier}. For a vertex pair $u,v$, let $T_{u,v}$ denote the (unique) path between $u$ and $v$ in $T$. Upon receiving a query associated with an arbitrary subset $S \subseteq V(G)$, we do the following:
\begin{itemize}
	\item Construct the subtree $T':= \bigcup_{u \in S} T_{u,r_T}$ that consists of all the paths from vertices in $S$ to the root $r_T$ of $T$.
	\item Return $T'$.
\end{itemize}

We immediately obtain the following result, which proves the third item of Corollary~\ref{cor:cut and flow} from the introduction.

\begin{corollary}
\label{cor:dynamicVertexSparsifier}
	There is a deterministic fully dynamic algorithm on a graph $G$ with $n$ vertices such that given a query associated with an arbitrary $S \subseteq V(G)$ outputs a tree flow sparsifier with quality $2^{O(\log^{5/6} n)}$ for $G$ w.r.t. $S$ using $2^{O(\log^{5/6} n)}$ amortized update time and $O(|S| \log^{1/6} n)$ query time. Moreover, the update time can be made worst-case while keeping the same quality and running time guarantees. 
\end{corollary}
\begin{proof}
	We first show that the output tree $T'$ is a tree flow sparsifier with quality $1$ for $T$ w.r.t. $S$. Since $T$ is a tree, every demand among two leaf vertices $u,v$ in $T$ is routed according to the unique path $T_{u,v}$ between $u$ and $v$ in $T$. If $u,v \in S$, note that $T_{u,v}$ is entirely contained in the sub-tree $T' = \bigcup_{u \in S} T_{u,_rt}$. Therefore, every demand that we route in $T$ between any vertex pair $u,v$ in $S$, can also be routed in $T'$ with the same congestion. For the other reduction, by construction we have that $T' \subseteq T$, i.e., every demand that we route in $T'$ between any vertex pair in $S$ can be routed in $T$ with the same congestion. Combining the above gives that $T'$ is a tree flow sparsifier with quality $1$ for $T$ w.r.t. $S$. As $T$ is a tree flow sparsifier with quality $2^{O(\log^{5/6} n)}$ for $G$~(Corollary~\ref{cor:dynamicTreeFlowSparsifier}), by the transitivity property of flow sparsifiers, it follows that $T'$ is a tree flow sparsifier with quality $2^{O(\log^{5/6} n)}$ for $G$ w.r.t. $S$.
	
	We next analyze the running time. The claimed amortized update time follows directly from Corollary~\ref{cor:dynamicTreeFlowSparsifier}. For the query time, Corollary~\ref{cor:dynamicTreeFlowSparsifier} ensures that at any time the depth of $T$ is $O(\log^{1/6} n)$. The latter guarantees that the length of each path from a leaf vertex to the root in $T$ is $O(\log^{1/6} n)$, which in turn implies that the time to compute $T'$ and its size are both bounded by $O(|S| \log^{1/6} n)$.
	
	To achieve our worst-case update time, we replace the expander hierarchy from Theorem~\ref{thm:deterministic_dynamic_ED} with the one from Thereom~\ref{thm:derand_deamort_main}, which in turn allows us to query for any given vertex $u$, the leaf-to-root path of $u$ in the hierarchy. Since we only need such paths for the construction of $T'$, our claim follows. 
\end{proof}

The above corollary readily implies a fully-dynamic algorithm for the all-pair approximate maximum flow problem: upon receiving a query associated with an arbitrary vertex pair $u,v \in V$ we let $S=\{u,v\}$ and then compute a tree flow sparsifier $T'$ for $G$ w.r.t. $S$ using Corollary~\ref{cor:dynamicVertexSparsifier}. Finally, we compute the maximum flow from $u$ to $v$ in $T'$ and return its value as an estimate. We have the following result, which proves the first item of Corollary~\ref{cor:cut and flow} from the introduction.

\begin{corollary}
	There is a deterministic fully dynamic algorithm on a graph $G$ with $n$ vertices that maintains for every vertex pair $u,v \in V$, an estimate that approximates the maximum flow from $u$ to $v$ in $G$ up to a factor of $2^{O(\log^{5/6} n)}$ using $2^{O(\log^{5/6} n)}$ worst-case update time and $O(\log^{1/6} n)$ query time. 
\end{corollary}

We next show that the same idea extends to the maximum concurrent (multi-commodity) flow problem, which is defined as follows: given an unweighted graph $G$ and $k$ source-sink pairs $s_i,t_i$, each associated with a non-negative demand $D(i)$, compute the congestion $\eta(G,D)$, i.e., the minimum congestion a flow $F$ that is a routing of $D$ in $G$, where $D := (D(1),\ldots,D(k))$. We study a dynamic version of the problem, where $G$ undergoes edge updates and the $k$ source-sink pairs are made available only at query time. Our dynamic construction uses Corollary~\ref{cor:dynamicVertexSparsifier}, and whenever the $k$ source-sink pairs are revealed to us, we define $V_k = \cup_{i} \{s_i,t_i\}$ and then compute a tree flow sparsifier $T'$ for $G$ w.r.t. $V_k$. Finally, we compute the congestion $\eta(T,D)$ in $T$ and return this value as an estimate. The result below follows from the definition of tree flow sparsifiers and proves the third item of Corollary~\ref{cor:cut and flow} from the introduction.

\begin{corollary}
	There is a fully dynamic deterministic algorithm on a graph $G$ with $n$ vertices that maintains for every demand set $D$ defined on $k$ source-sink pairs $s_i,t_i$, an estimate that approximates $\eta(G,D)$ up to a factor of $2^{O(\log^{5/6} n)}$ using $2^{O(\log^{5/6} n)}$ worst-case update time and $O(k \log^{1/6} n)$ query time. 
\end{corollary}

We finally consider a dynamic version of the multi-way cut problem, which is defined as follows. Given an unweighted graph $G$ and $k$ distinguished vertices $s_1,\ldots,s_k$, the goal is to remove a minimum number of edges $F$ such that no pair of distinguished vertices $s_i$ and $s_j$ with $i \neq j$ belong to the same connected component after the removal of $F$ from $G$. We study a dynamic version of the problem, where $G$ undergoes edge updates and the $k$ distinguished vertices are made available only at query time. Similarly to above, we use Corollary~\ref{cor:dynamicVertexSparsifier} and whenever the $k$ distinguished vertices are revealed to us, we define $V_k =\cup_i \{s_i\}$ and then compute a tree flow sparsifier $T'$ for $G$ w.r.t. $V_k$. Finally, we compute an optimal solution to the multi-way cut problem on $T'$ with respect to the queried $k$ distinguished vertices and return this value as an esimate. The result below follows from the definition of tree flow sparsifiers and proves the third item of Corollary~\ref{cor:cut and flow} from the introduction. 

\begin{corollary}
	There is a fully dynamic deterministic algorithm on a graph $G$ with $n$ vertices that maintains for any $k$ distinguished vertices $s_1,\ldots,s_k$, an estimate that approximates an optimal solution to the multi-way cut up to a factor of $2^{O(\log^{5/6} n)}$ using $2^{O(\log^{5/6} n)}$ worst-case update time and $O(k \log^{1/6} n)$ query time. 
\end{corollary}

The dynamic multicut essentially follows the same idea and we omit it here for the sake of brevity.

\subsection{Dynamic Sparsest Cut and Lowest Conductance Cut}
We show that a dynamic tree flow sparsifier can be used to maintain sparsest cuts, multi-cuts and multi-way cuts. Throughout, we only focus on the dynamic sparsest cut problem. An almost identical idea extends to the lowest conductance cut but we omit a detailed description here for the sake of brevity.

Let $G=(V,E,c)$ be a weighted graph. For any cut $(S,\bar{S})$ such that $|S| \leq |\bar{S}|$, let $c(\delta(S))$ be the sum over capacites of all edges with one endpoint in $S$ and the other in $\bar{S}$, where $\bar{S} = V \setminus S$. Let $\alpha(G,S) := c(\delta(S))/|S|$ be the sparsity of $(S,\bar{S})$. The \emph{sparsest cut} problem asks to find a cut $(S,\bar{S})$ such that $|S| \leq |\bar{S}|$  with smallest possible sparsity in $G$, which we denoted by $\alpha(G)$. We study a dynamic version of this problem, where $G$ undergoes edge updates and at query time we need to report the sparsity $\alpha(G)$ of the sparsest cut in the current graph $G$. To design a dynamic algorithm, we follow a well-known approach used to solve the static version of the problem: given a graph $G$, (1) compute a tree flow sparsifier $T$ with quality $q$ for $G$ and (2) solve the sparsest cut problem on $T$. Since a tree flow sparsifier is also a tree cut sparsifier with the same quality, it is easy to verify that $\alpha(T)$ approximates $\alpha(G)$ up to a factor of $q$. The main advantage of this approach is that computing sparsest cut on trees is much easier. 

To see this, consider a (rooted) tree flow sparsifier $T=(V(T),E(T),c^T)$ with quality $q$ and depth $t$ for $G$ such that the leaf nodes of $T$ correspond to the vertices of $G$. It is known that the sparsest cut on a tree must occur at one of the edges in $T$. We can also build a data-structure such that given an internal node $x$ in $T$ (except the root), it reports the number of leaf nodes in the sub-tree rooted at $x$. Using these two observations, an algorithm for computing $\alpha(T)$ works as follows: 
\begin{itemize}
	\item For each edge $e=(x,p(x)) \in E(T)$~(as $T$ is rooted), where $p$ is the parent of $x$, compute the sparsity of the cut $(S,\bar{S})$ obtained by removing $(x,p(x))$ in $T$ using $c^T(e)/|S|$, where $|S|$ is precisely the number of leaf nodes in the sub-tree rooted at $x$.  
	\item Return $\min_{e \in E(T)} c^T(e)/|S|$.
\end{itemize}

In a similar vein, using Corollary~\ref{cor:dynamicTreeFlowSparsifier} we maintain a tree flow sparsifier $T$ for an unweighted dynamic graph $G$. As $T$ undergoes changes, we additionally update the information about the number of leaf nodes at an internal node and the edge with the smallest sparsity in $T$. Since these updates can be implemented in time proportional to the time needed to maintain $T$, we obtain the following result, which proves the second item of Corollary~\ref{cor:cut and flow} from the introduction.

\begin{corollary}
	There is a deterministic fully dynamic algorithm on a graph $G$ with $n$ vertices that maintains an estimate that approximates $\alpha(G)$ up to a factor of $2^{O(\log^{5/6} n)}$ using $2^{O(\log^{5/6} n)}$ amortized update time and $O(\log^{1/6} n)$ query time.
\end{corollary}

\subsection{Dynamic Connectivity}

We observe that the data-structure representation of the expander hierarchy from Theorem~\ref{thm:derand_deamort_main} leads to a dynamic algorithm for maintaining connectivity information of $G$. More precisely, a graph $G$ is connected iff the top level our expander hierarchy consists of a single vertex. Moreover, two vertices $u$ and $v$ are connected iff the roots of $u$ and $v$ in the hierarchy are the same. These observations lead to the following result, which proves Corollary~\ref{cor:conn} from the introduction.

\begin{corollary}
	There is a deterministic fully dynamic algorithm on a $n$-vertex graph $G$ that maintains connectivity of $G$ using $2^{-O(\log^{5/6} n)}$ worst-case update time and also supports pairwise connectivity queries in $O(\log^{1/6} n)$ time.  
\end{corollary}



\subsection{Treewidth decomposition}
\label{subsec:treewidth}

A \emph{treewidth decomposition} $T$ of a graph $G=(V,E)$ is a tree such that
each node $x$ in $T$ corresponds to a set $B_{x}\subseteq V$ of vertices called
a \emph{bag}. For each edge $(u,v)\in E$, there must exist a node $x$
whose bag $B_{x}$ contains both $u$ and $v$. Moreover, for each
vertex $u\in V$, $\{x\mid u\in B_{x}\}$ must induce a connected
subtree of $T$. A \emph{width} of $T$ is $\max_{x}|B_{x}|-1$. The
\emph{treewidth} $\tw(G)$ of $G$ is the minimum width over all treewidth
decomposition of $G$. 

We obtain the first dynamic algorithm for maintaining a tree width decomposition. The main observation behind our construction is that a treewidth decomposition of a graph can be directly derived from an expander hierarchy, which works as follows. Let $T$ be a expander hierarchy of a graph $G=(V,E)$. We simply let $T$ itself be the treewidith decomposition. It remains to define a bag $B_{x}$ for each node $x\in T$.

%



To this end, for each node $x\in T$, let $U$ be a cluster from $T$ corresponding
to a node $x$. Recall that $U\subseteq V(G^{i})$ for some $i$.
Let $E_{G_{i}}(U,V(G^{i}))$ denote the set of edges in $G^{i}$ incident
to a vertex from $U$. For each $e^{i}\in E_{G_{i}}(U,V(G^{i}))$,
there is a corresponding ``original'' edge $e$ of $G$. The bag
$B_{x}\subseteq V$ consists of the endpoints of all ``original''
edges correspond to edges from $E_{G_{i}}(U,V(G^{i}))$. See \Cref{fig:bag}
for an example. 

\begin{figure}[h]
\centering
\subfigure{\scalebox{0.24}{\includegraphics{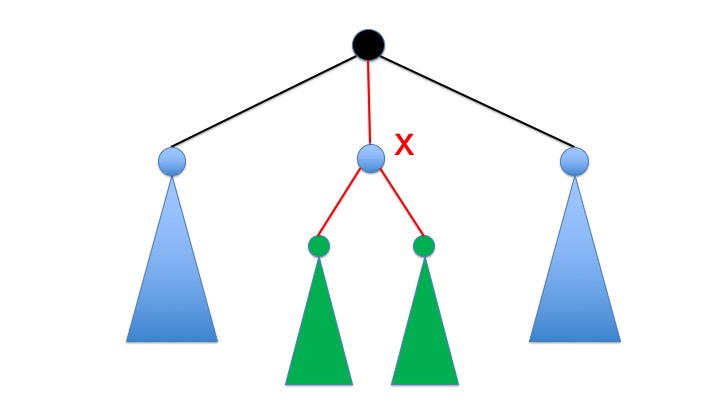}}}
\hspace{0.5cm}
\subfigure{\scalebox{0.24}{\includegraphics{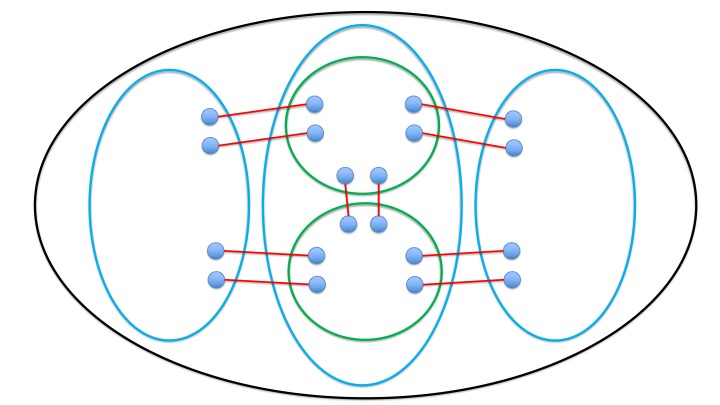}}}
\caption{An illustration of a bag $B_{x}$ of a node $x$ in an expander hierarchy. \label{fig:bag}}
\end{figure}

\begin{lem}	
The expander hierarchy $T$ is a treewith decomposition of $G$.
\end{lem}

\begin{proof}
Observe that for each edge $(u,v)\in E$ there is a unique cluster graph
$C$ in $T$ containing a edge $e'\in E(C)$ corresponding to $e$.
Let $x$ be the node in $T$ corresponding to $C$. It is clear that
the bag $B_{x}$ containing both $u$ and $v$. 

Next, suppose towards a contradiction that there is a vertex $u$ where
the set $\{x\mid u\in B_{x}\}$ does not induce a connected subtree
of $T$. Let $T_{1}$ and $T_{2}$ be two disconnected induced subtrees.
Let $y$ be a node in a path connecting $T_{1}$ and $T_{2}$ such that $y$ is neither in $T_{1}$ nor in $T_{2}$. Observe that the bag $B_{y}$ is a separator that separates vertices in the bags in $T_{1}$ and $T_{2}$.
More precisely, let $V_{1}=\bigcup_{x\in T_{1}}B_{x}$ and $V_{2}=\bigcup_{x\in T_{2}}B_{x}$.
Observe that in the graph $G[V\setminus B_{y}]$, no pair of vertices
between $V_{1}\setminus B_{y}$ and $V_{2}\setminus B_{y}$ can be
connected. However, we have that $u\in V_{1}\setminus B_{y}$ and $u\in V_{2}\setminus B_{y}$, which is a contradiction.
\end{proof}

To bound the width of our treewidth decomposition, we need the notions of well-linkedness and flow-linkedness. 
\begin{definition}
A set of $S \subset V(G)$ is \emph{$\gamma$-well-linked} in $G$ iff any
cut $(A,B)$ in $G$, $|E(A,B)|\ge\gamma\cdot\min\{|A\cap S|,|B\cap S|\}$.
\end{definition}
%
\begin{definition}
	A set $S \subset V(G)$ is \emph{$\gamma$-flow-linked} in $G$ if given any multi-commodity
	flow demand $D$ on $S$ where the total demand on each vertex $v\in S$
	is at most $1$, the congestion for routing $D$ in $G$ is at
	most $\eta(G,D)\le1/\gamma$. 
\end{definition}
It is easy to see that any $\gamma$-flow-linked set in $G$ is a $\gamma$-well-linked set in $G$.
The next fact relates well-linkedness and treewidth in a bounded degree graph. 

\begin{fact}[A paraphrase of Corollary 2.1 from \cite{ChekuriC13}]\label{fact:tw linked}
Let $G=(V,E)$
be a graph with maximum degree $\Delta$. Let $B\subseteq V$ be a
set of vertices that is $\gamma$-well-linked in $G$. Then $\tw(G)\ge\frac{\gamma|B|}{3\Delta}-1$ or equivalently $|B|=O(\frac{\Delta}{\gamma}\cdot\tw(G))$. 
\end{fact}

The following key technical lemma relates the notion of tree flow sparsifiers to the notion of treewidth via flow-linkedness.

\begin{lem}
	\label{lem:bag is linked}
	If $T$ is a tree flow sparsifier with quality $q$, then each bag
	$B_{x}$ is $\Omega(1/q)$-flow-linked.
\end{lem}

\begin{proof}
	Let $D$ be any demand $D$ on $B_{x}$ where the total demand on
	each vertex $v\in B_{x}$ is at most $1$. It suffices to show $D$
	is routable in $T$, i.e.~$\eta(T,D)\le1$. This is because $T$
	has quality $q$, so $D$ can be routed in $G$ with congestion
	$q$, i.e.~$\eta(G,D)\le q$.
	
	The crucial observation is that all vertices $v\in B_{x}$ can route
	one unit of flow in $T$ to $x$ \emph{simultaneously} without congestion.
	The latter holds since each $v\in B_{x}$ is an endpoint of some boundary
	edge $e$ of a cluster correspond to the node $x$ or the children
	of $x$ in $T$. Therefore, the edge $e$ contributes to one unit capacity to
	every tree-edge in the path from $v$ to $x$ in $T$. 
	
	Now, to route $D$ in $T$, each vertex $v\in B_{x}$ just sends flow
	(equal to its total demand) of at most one unit to $x$, which causes
	no congestion. Connecting the all flow paths that meet at $x$ completes the proof of the lemma. 
\end{proof}

As an expander hierarchy is a good quality tree flow sparsifiers, our construction of treewidth decomposition has small width.
This fact is summarized as follows:

\begin{cor}
	\label{cor:helpfulTreeWidth}
	Let $G$ be a constant degree graph and let $T$ an $(\alpha,\phi)$-expander
	hierarchy of $G$ with depth $t$ and slack $s$. Then each bag $B_{x}$
	has size at most $\tw(G)\cdot O(s\log m)^{t}\cdot O(\max\{\frac{1}{\alpha},\frac{1}{\phi}\}/ \alpha^{t-1})$. 
\end{cor}

\begin{proof}
	By \Cref{thm:EH_gives_R_tree}, $T$ has quality $q=O(s\log m)^{t}\cdot O(\max\{\frac{1}{\alpha},\frac{1}{\phi}\}\cdot\frac{1}{\alpha^{t-1}})$,
	and thus each bag $B_{x}$ is $\Omega(1/q)$-flow-linked by \Cref{lem:bag is linked}, and hence also
	$\Omega(1/q)$-well-linked. Finally, by \Cref{fact:tw linked}, $|B_{x}|\le O(\tw(G)\cdot q)$
	as desired.
\end{proof}

Our dynamic algorithm for treewidth decomposition proceeds as follows. We maintain a $(\alpha, \phi)$-expander hierarchy $T$ with slack $s:=2^{\bar{O}(\log^{1/2} n)}$ and depth $t:= O(\log^{1/6} n)$ of $G$ using Theorem~\ref{thm:deterministic_dynamic_ED}, where $\alpha := 2^{-\bar{O}(\log^{2/3} n)}$, $\phi := 2^{-O(\log^{5/6} n)}$. Using Corollary~\ref{cor:helpfulTreeWidth} and observing that we can explicitly update all the bags within the same running time guaranteed by Theorem~\ref{thm:deterministic_dynamic_ED}, we get the following result which proves Corollary~\ref{cor:treewidth} from the introduction.

\begin{cor}
	There is a deterministic fully dynamic algorithm
	on a constant degree graph $G$ with $n$ vertices that maintains a treewidth decomposition of $G$ with width $tw(G) \cdot 2^{O(\log^{5/6} n)}$ using $2^{O(\log^{5/6} n)}$ amortized update time.
\end{cor}

\clearpage
\appendix
\section{Proof of Lemma~\ref{lem:restricted_demand}}
\label{apd:restricted}

	We show this via an approximate maxflow-mincut theorem. It is well
	known \cite{LeightonR99}
	that the optimum congestion required for solving a
	multicommodity flow problem with demands $D$ in an undirected graph $G=(V,E)$ is at most
	$O(\log n / \operatorname{sparsity}(G,D))$, where
	$\operatorname{sparsity}(G,D)=\max_{X\subseteq{V}}{|E_G(X,V\setminus X)|/ D(X,V\setminus X)}$, and $D(X,V\setminus X)=\sum_{(x,y)\in X\times
		V\setminus X}(D(x,y)+D(y,x))$ is the demand that has to cross the cut $X$.
	
	With this in mind we prove the lemma by showing that $D$ has low sparsity in
	$G[S]$. Fix a subset $X$. The demand that originates at vertices in $X$ is at
	most $\sum_x\gamma\deg_G(x)$ because the demand is $\gamma$-restricted. The
	same holds for the demand that ends at vertices in $X$. This means that the
	total demand that can cross the cut can be at most
	$\gamma\min\{\vol_G(X),\vol_G(V\setminus X)\}$.
	But since $\induced{S}^{\alpha/\phi}$ is a $\phi$-expander we know that
	$|E_G(X,S\setminus
	X)|\ge
	\phi
	\min\{\vol_{\induced{S}[\alpha/\phi]}(X),\vol_{\induced{S}[\alpha/\phi]}(S\setminus
	X)\}
	\ge
	\phi\min\{\vol_G(X),\vol_{G}(S\setminus
	X)\}$. Using approximate maxflow-mincut gives the first statement.
	
	For the second statement the demand that has to cross the cut can be at most
	
	\begin{equation*}
	\begin{split}
	D(X,S\setminus X) &\le \gamma\cdot\min\{{\textstyle\sum}_{v\in X}|E_G(\{v\},V\setminus S)|,{\textstyle\sum}_{v\in
		S\setminus X}|E_G(\{v\},V\setminus S)|\}\\
	&\le\gamma\cdot\tfrac{\phi}{\alpha}\cdot\min\{\vol_{\induced{S}[\alpha/\phi]}(X),\vol_{\induced{S}[\alpha/\phi]}(S\setminus X)\}\\
	&\le\gamma\cdot\tfrac{\phi}{\alpha}\cdot\tfrac{1}{\phi}|E_G(X,S\setminus X)|\enspace,
	\end{split}
	\end{equation*}
	where the first inequality is due to the $\gamma$-boundary restriction, the
	second due to the reweighting of boundary edges in the graph
	$\induced{S}[\alpha/\phi]$ and the last inequality follows from the
	$\phi$-expansion of $\induced{S}[\alpha/\phi]$.

\clearpage		  
\bibliographystyle{alpha}
\bibliography{references} 
         
\end{document}